\documentclass[11pt,letterpaper]{article}

\def\showauthornotes{0}
\def\showtableofcontents{1}
\def\showkeys{0}
\def\showdraftbox{0}
\def\usemicrotype{1}
\def\showfixme{0}

\usepackage{etex}

\usepackage[l2tabu, orthodox]{nag}

\usepackage{xspace,enumerate}

\usepackage[dvipsnames]{xcolor}

\usepackage[T1]{fontenc}
\usepackage[full]{textcomp}

\usepackage[american]{babel}

\usepackage{mathtools}

\usepackage{amsthm}

\newtheorem{theorem}{Theorem}[section]
\newtheorem*{theorem*}{Theorem}

\newtheorem{proposition}[theorem]{Proposition}
\newtheorem*{proposition*}{Proposition}
\newtheorem{lemma}[theorem]{Lemma}
\newtheorem*{lemma*}{Lemma}
\newtheorem{corollary}[theorem]{Corollary}
\newtheorem*{conjecture*}{Conjecture}
\newtheorem{fact}[theorem]{Fact}
\newtheorem*{fact*}{Fact}

\newtheorem*{hypothesis*}{Hypothesis}

\theoremstyle{definition}
\newtheorem{definition}[theorem]{Definition}

\theoremstyle{remark}
\newtheorem{claim}[theorem]{Claim}
\newtheorem*{claim*}{Claim}

\newtheorem*{remark*}{Remark}

\newtheorem*{observation*}{Observation}

\usepackage[
letterpaper,
top=0.7in,
bottom=0.9in,
left=1in,
right=1in]{geometry}

\usepackage[varg]{pxfonts}
\usepackage{textcomp}
\usepackage[scr=rsfso]{mathalfa}
\usepackage{bm} 
\let\mathbb\varmathbb

\ifnum\showkeys=1
\usepackage[color]{showkeys}
\fi

\usepackage[
pagebackref,
colorlinks=true,
urlcolor=blue,
linkcolor=blue,
citecolor=OliveGreen,
]{hyperref}
\usepackage{prettyref}

\let\pref=\prettyref
\newcommand{\savehyperref}[2]{\texorpdfstring{\hyperref[#1]{#2}}{#2}}

\newrefformat{eq}{\savehyperref{#1}{\textup{(\ref*{#1})}}}
\newrefformat{lem}{\savehyperref{#1}{Lemma~\ref*{#1}}}
\newrefformat{def}{\savehyperref{#1}{Definition~\ref*{#1}}}
\newrefformat{thm}{\savehyperref{#1}{Theorem~\ref*{#1}}}
\newrefformat{cor}{\savehyperref{#1}{Corollary~\ref*{#1}}}
\newrefformat{cha}{\savehyperref{#1}{Chapter~\ref*{#1}}}
\newrefformat{sec}{\savehyperref{#1}{Section~\ref*{#1}}}
\newrefformat{app}{\savehyperref{#1}{Appendix~\ref*{#1}}}
\newrefformat{tab}{\savehyperref{#1}{Table~\ref*{#1}}}
\newrefformat{fig}{\savehyperref{#1}{Figure~\ref*{#1}}}
\newrefformat{hyp}{\savehyperref{#1}{Hypothesis~\ref*{#1}}}
\newrefformat{alg}{\savehyperref{#1}{Algorithm~\ref*{#1}}}
\newrefformat{rem}{\savehyperref{#1}{Remark~\ref*{#1}}}
\newrefformat{item}{\savehyperref{#1}{Item~\ref*{#1}}}
\newrefformat{step}{\savehyperref{#1}{step~\ref*{#1}}}
\newrefformat{conj}{\savehyperref{#1}{Conjecture~\ref*{#1}}}
\newrefformat{fact}{\savehyperref{#1}{Fact~\ref*{#1}}}
\newrefformat{prop}{\savehyperref{#1}{Proposition~\ref*{#1}}}
\newrefformat{prob}{\savehyperref{#1}{Problem~\ref*{#1}}}
\newrefformat{claim}{\savehyperref{#1}{Claim~\ref*{#1}}}
\newrefformat{relax}{\savehyperref{#1}{Relaxation~\ref*{#1}}}
\newrefformat{red}{\savehyperref{#1}{Reduction~\ref*{#1}}}
\newrefformat{part}{\savehyperref{#1}{Part~\ref*{#1}}}

\newcommand{\Sref}[1]{\hyperref[#1]{\S\ref*{#1}}}

\usepackage{nicefrac}

\ifnum\usemicrotype=1
\usepackage{microtype}
\fi

\ifnum\showauthornotes=1
\newcommand{\Authornote}[2]{{\sffamily\small\color{red}{[#1: #2]}}}
\newcommand{\Authornotecolored}[3]{{\sffamily\small\color{#1}{[#2: #3]}}}
\newcommand{\Authorcomment}[2]{{\sffamily\small\color{gray}{[#1: #2]}}}
\newcommand{\Authorstartcomment}[1]{\sffamily\small\color{gray}[#1: }

\newcommand{\Authorfnote}[2]{\footnote{\color{red}{#1: #2}}}
\newcommand{\Authorfixme}[1]{\Authornote{#1}{\textbf{??}}}
\newcommand{\Authormarginmark}[1]{\marginpar{\textcolor{red}{\fbox{\Large #1:!}}}}
\else
\newcommand{\Authornote}[2]{}
\newcommand{\Authornotecolored}[3]{}
\newcommand{\Authorcomment}[2]{}
\newcommand{\Authorstartcomment}[1]{}

\newcommand{\Authorfnote}[2]{}
\newcommand{\Authorfixme}[1]{}
\newcommand{\Authormarginmark}[1]{}
\fi

\ifnum\showfixme=0

\fi

\usepackage{boxedminipage}

\newcommand{\Paren}[1]{\left(#1\right)}

\newcommand{\Abs}[1]{\left\lvert#1\right\rvert}

\newcommand{\Norm}[1]{\left\lVert#1\right\rVert}

\newcommand{\iprod}[1]{\langle#1\rangle}

\newcommand{\Esymb}{\mathbb{E}}
\newcommand{\Psymb}{\mathbb{P}}

\DeclareMathOperator*{\E}{\Esymb}

\DeclareMathOperator*{\ProbOp}{\Psymb}

\renewcommand{\Pr}{\ProbOp}

\newcommand{\tensor}{\otimes}

\newcommand{\defeq}{\stackrel{\mathrm{def}}=}

\newcommand{\seteq}{\mathrel{\mathop:}=}

\newcommand{\mper}{\,.}
\newcommand{\mcom}{\,,}

\newcommand\bdot\bullet

\ifx\mathds\undefined 

\else

\fi

\DeclareMathOperator{\Tr}{Tr}

\DeclareMathOperator{\poly}{poly}

\DeclareMathOperator{\polylog}{polylog}

\let\varprod\undefined

\DeclareMathOperator*{\varprod}{{\textstyle \prod}}

\newcommand{\Erdos}{Erd\H{o}s\xspace}
\newcommand{\Renyi}{R\'enyi\xspace}

\newcommand{\N}{\mathbb N}
\newcommand{\R}{\mathbb R}

\newcommand{\cB}{\mathcal B}

\newcommand{\cL}{\mathcal L}
\newcommand{\cM}{\mathcal M}
\newcommand{\cN}{\mathcal N}

\newcommand{\cP}{\mathcal P}

\newcommand{\cR}{\mathcal R}

\renewcommand{\leq}{\leqslant}

\renewcommand{\geq}{\geqslant}

\ifnum\showdraftbox=1
\newcommand{\draftbox}{\begin{center}
  \fbox{%
    \begin{minipage}{2in}%
      \begin{center}%
          \Large\textsc{Working Draft}\\%
        Please do not distribute%
      \end{center}%
    \end{minipage}%
  }%
\end{center}
\vspace{0.2cm}}
\else
\newcommand{\draftbox}{}
\fi

\let\epsilon=\varepsilon

\numberwithin{equation}{section}

\newcommand\MYcurrentlabel{xxx}

\newcommand{\MYstore}[2]{%
  \global\expandafter \def \csname MYMEMORY #1 \endcsname{#2}%
}

\newcommand{\MYload}[1]{%
  \csname MYMEMORY #1 \endcsname%
}

\newcommand{\MYnewlabel}[1]{%
  \renewcommand\MYcurrentlabel{#1}%
  \MYoldlabel{#1}%
}

\newcommand{\MYdummylabel}[1]{}

\newcommand{\torestate}[1]{%
  \let\MYoldlabel\label%
  \let\label\MYnewlabel%
  #1%
  \MYstore{\MYcurrentlabel}{#1}%
  \let\label\MYoldlabel%
}

\newcommand{\restatetheorem}[1]{%
  \let\MYoldlabel\label
  \let\label\MYdummylabel
  \begin{theorem*}[Restatement of \prettyref{#1}]
    \MYload{#1}
  \end{theorem*}
  \let\label\MYoldlabel
}

\newcommand{\restatelemma}[1]{%
  \let\MYoldlabel\label
  \let\label\MYdummylabel
  \begin{lemma*}[Restatement of \prettyref{#1}]
    \MYload{#1}
  \end{lemma*}
  \let\label\MYoldlabel
}

\newcommand{\restateprop}[1]{%
  \let\MYoldlabel\label
  \let\label\MYdummylabel
  \begin{proposition*}[Restatement of \prettyref{#1}]
    \MYload{#1}
  \end{proposition*}
  \let\label\MYoldlabel
}

\newcommand{\restatefact}[1]{%
  \let\MYoldlabel\label
  \let\label\MYdummylabel
  \begin{fact*}[Restatement of \prettyref{#1}]
    \MYload{#1}
  \end{fact*}
  \let\label\MYoldlabel
}

\newcommand{\restate}[1]{%
  \let\MYoldlabel\label
  \let\label\MYdummylabel
  \MYload{#1}
  \let\label\MYoldlabel
}

\newcommand{\addreferencesection}{
  \phantomsection
  \addcontentsline{toc}{section}{References}
}

\let\origparagraph\paragraph
\renewcommand{\paragraph}[1]{\origparagraph{#1.}}

\allowdisplaybreaks

\usepackage{paralist}

\usepackage{comment}

\let\citet\cite

\theoremstyle{definition}

\usepackage{relsize}
\usepackage[font=footnotesize]{caption}

\DeclareMathOperator{\pE}{\tilde {\mathbb E}}

\let\cL\relax
\DeclareMathOperator{\cL}{\mathcal L}

\DeclareMathOperator{\Span}{Span}

\DeclareUrlCommand\email{}

\newcommand{\tO}{\tilde O}

\newcommand{\Err}{\text{Err}}

\usepackage{skull}
\usepackage{appendix}
\usepackage{amssymb}
\usepackage{amsfonts}
\usepackage{latexsym}
\usepackage{amsmath}
\usepackage[T1]{fontenc}
\usepackage{fullpage,appendix}
\usepackage{dsfont}
\usepackage[capitalize]{cleveref}

\usepackage{color}
\usepackage{tikz}
\newcommand{\eqdef}{\stackrel{\textrm{def}}{=}}

\newcommand{\ignore}[1]{}
\definecolor{corlinks}{RGB}{64,128,128}
\definecolor{cormenu}{RGB}{0,37,94}
\definecolor{corurl}{RGB}{0,46,91}

\newcommand{\on}{\{-1,1\}}

\renewcommand{\P}{\mathcal{P}}

\newcommand{\B}{\mathcal{B}}

\renewcommand{\L}{\mathcal{L}}

\newcommand{\M}{\mathcal{M}}
\renewcommand{\S}{\mathbb{S}}

\newcommand{\tE}{\tilde{\mathbb{E}}}

\renewcommand{\int}{\mathsf{int}}

\newcommand{\MaxC}{\textsc{MAX CLIQUE}}
\newcommand{\nchoose}[1]{{{[n]} \choose {#1}}}
\newcommand{\sE}{\mathcal{E}}

\renewcommand{\top}{\dagger}

\title{SoS and Planted Clique: Tight Analysis of MPW Moments at all Degrees and an Optimal Lower Bound at Degree Four}

\author{%
Samuel B. Hopkins\thanks{Department of Computer Science, Cornell University. Work done as an intern at Microsoft Research New England.}
\and Pravesh K. Kothari \thanks{Department of Computer Science, UT Austin. Work done as an intern at Microsoft Research New England.}
\and Aaron Potechin \thanks{Simons Fellow. Work done as an intern at Microsoft Research New England and at MIT with an NSF graduate research fellowship under grant No. 0645960.} }

\begin{document}

\maketitle
\draftbox
\thispagestyle{empty}

\begin{abstract}
The problem of finding large cliques in random graphs and its ``planted" variant, where one wants to recover a clique of size $\omega \gg \log{(n)}$ added to an \Erdos-\Renyi graph $G \sim G(n,\frac{1}{2})$, have been intensely studied.
Nevertheless, existing polynomial time algorithms can only recover planted cliques of size $\omega = \Omega(\sqrt{n})$.
By contrast, information theoretically, one can recover planted cliques so long as $\omega \gg \log{(n)}$.

In this work, we continue the investigation of algorithms from the sum of squares hierarchy for solving the planted clique problem begun by Meka, Potechin, and Wigderson \cite{MPW15} and Deshpande and Montanari \cite{DM15}.
Our main results improve upon both these previous works by showing:
\begin{enumerate}
  \item Degree four SoS does not recover the planted clique unless $\omega \gg \sqrt n / \polylog n$, improving upon the bound $\omega \gg n^{1/3}$ due to \cite{DM15}. A similar result was obtained independently by Raghavendra and Schramm \cite{RS15}.
  
  \item For $2 < d = o(\sqrt{\log{(n)}})$, degree $2d$ SoS does not recover the planted clique unless $\omega \gg n^{1/(d + 1)} /(2^d \polylog n)$, improving upon the bound due to \cite{MPW15}.
\end{enumerate}
Our proof for the second result is based on a fine spectral analysis of the certificate used in the prior works \cite{MPW15,DM15,DBLP:journals/siamcomp/FeigeK03} by decomposing it along an appropriately chosen basis. Along the way, we develop combinatorial tools to analyze the spectrum of random matrices with dependent entries and to understand the symmetries in the eigenspaces of the set symmetric matrices inspired by work of Grigoriev \cite{Gri01}.

An argument of Kelner shows that the first result cannot be proved using the same certificate. Rather, our proof involves constructing and analyzing a new certificate that yields the nearly tight lower bound by ``correcting" the certificate of \cite{MPW15,DM15,DBLP:journals/siamcomp/FeigeK03}.

\end{abstract}

\clearpage

\ifnum\showtableofcontents=1
{
\tableofcontents
\thispagestyle{empty}
 }
\fi

\clearpage

\setcounter{page}{1}

\section{Introduction}
Let $G(n,p)$ be the \Erdos-\Renyi random graph where each edge is present in $G$ with probability $p$ independently of others. It is an easy calculation that the largest clique in $G \sim G(n,\frac{1}{2})$ is of size $(2+o(1))\cdot \log{(n)}$ with high probability. Recovering such a clique using an efficient algorithm has been a long standing open question in theoretical computer science. As early as 1976, Karp \cite{Kar76} suggested the impossibility of finding cliques of size even $(1+\epsilon) \log{(n)}$ for any constant $\epsilon > 0$ in polynomial time. Karp's conjecture was remarkably prescient and has stood ground after nearly $4$ decades of research.

Lack of algorithmic progress on the question motivated Jerrum \cite{DBLP:journals/rsa/Jerrum92} and Kucera \cite{DBLP:journals/dam/Kucera95} to consider a relaxed version known as the \emph{planted clique} problem. In this setting, we are given a graph $G$ obtained by planting a clique of size $\omega$ on a graph sampled according to $G(n,\frac{1}{2})$. Information theoretically, the added clique is identifiable as long a $\omega \gg \log{(n)}$. The goal is to recover the added clique via an efficient algorithm for as small an $\omega$ as possible. This variant is also connected to the question of finding large communities in social networks and the problem of \emph{signal finding} in molecular biology \cite{PS00}. Despite attracting a lot of attention, the best known polynomial time algorithm can only find planted cliques when their size $\omega = \Omega(\sqrt{n})$ \cite{DBLP:conf/soda/AlonKS98,DBLP:journals/siamcomp/FeigeK03}. The LS+ semi-definite programming hierarchy leads to the state of the art trade off: planted cliques of size $\omega \approx \sqrt{\frac{n}{2^{d}}}$ can be recovered in time $n^{O(d)}$ for any $d = O(\log{(n)}$.

Recently, this difficulty of finding cliques of size $\omega \ll \sqrt{n}$ has led to an increasing confidence in planted clique being a candidate for an average case hard problem and has inspired new research directions in cryptography \cite{DBLP:conf/stoc/ApplebaumBW10}, property testing \cite{DBLP:conf/stoc/AlonAKMRX07}, machine learning \cite{BR13}, algorithmic game theory \cite{DBLP:conf/soda/HazanK09a,DBLP:conf/approx/AustrinBC11} and  mathematical finance \cite{DBLP:conf/innovations/AroraBBG10}. 
%

In this paper we are interested in understanding the \emph{Sum of Squares} (SoS, also known as \emph{Lasserre}) semi-definite programming (SDP) hierachy \cite{Las01,Par00} for the planted clique problem. This is a family of algorithms, paramterized by a number $d$ called the \emph{degree}, where the $d^{th}$ algorithm takes $n^{O(d)}$ time to execute. The sum of squares hierarchy can be viewed as a common generalization and extension of both linear programming and spectral techniques, and as such has been remarkably successful in combinatorial optimization. In particular it captures the state of the art algorithms for problems such as Sparsest Cut \cite{ARV09}, MAX CUT \cite{GW95}, Unique Games/Small Set Expansion \cite{ABS10, BRS12,GS12}. 
 Recently, \cite{BBHKSZ12} showed that a polynomial time algorithm from this hierarchy solves all known integrality gap instances of the Unique Games problem, and similar results have been shown for the hard instances of MAX-CUT \cite{DMN13} and Balanced Separator \cite{OZ13}.  Very recently, \cite{LRS15} showed that the sum of squares algorithm is in fact optimal amongst all efficient algorithms based on semidefinite programming for a large class of problems that includes constraint satisfaction and the traveling salesman problem. Moreover, Barak, Kelner and Steurer~\cite{BKS14,BKS15} used the SoS hierarchy to give improved algorithms to \emph{average case} problems such as the dictionary learning problem and the planted sparse vector problem that have at least some similarity to the planted clique problem.Thus several researchers have asked whether the SoS hierarchy can yield improved algorithms for this problem as well.


The first published work along these lines was of Meka, Potechin and Wigderson \cite{MPW15} who showed that for every $d \geq 2$, the degree $2d$ SoS cannot find planted cliques of size smaller than $\approx n^{\frac{1}{2d}}$.\footnote{We use $\approx$ to denote equality up to factors polylogarithmic in $n$ (the size of the graph) and with an arbitrary dependence on the degree parameter $d$.} Deshpande and Montanari \cite{DM15} independently proved a tighter lower bound of  $\approx n^{1/3}$ for the case of degree $4$. In the main result of this paper, we extend the prior works and show that the first non trivial extension of the spectral algorithm, namely the SoS algorithm of degree $4$, cannot find cliques of size $\approx \sqrt{n}$, a bound optimal within $\poly \log {(n)}$ factors. Our lower bound for degree $4$ is obtained by a careful ``correction'' to the certificate used by $\cite{MPW15} $ and $\cite{DM15}$ in their lower bounds. 

\begin{theorem}[Main Theorem 1]
  \label{thm:deg4-intro}
  The canonical degree $4$ SoS relaxation of the planted clique problem (\pref{eq:prog-formulation}) has an integrality gap of at least $\tilde{O}(\sqrt{n})$ with high probability.\footnote{Throughout this paper, we use $\tilde{O}$ to hide polylogarithmic factors in $n$}.
\end{theorem}
A similar result was obtained in an independent work by Raghavendra and Schramm \cite{RS15}. In our second main result, we give a tight analysis of the certificate considered by \cite{MPW15} and \cite{DM15} and show that it yields a lower bound of $n^{\frac{1}{d +1}}$.

\begin{theorem}[Main Theorem 2]
  \label{thm:degd-intro}
 For every $d = o(\sqrt{\log{(n)}})$, the canonical degree $2d$ SoS relaxation for the planted clique problem (\pref{eq:prog-formulation}) has an integrality gap of at least $\tilde{O}(n^{\frac{1}{d+1}})$. 
\end{theorem}

The certificate of \cite{MPW15,DM15} is sufficient to show an $\Omega(\sqrt{n/2^d})$ lower bound for the degree $d$ $LS+$ hierarchy~\cite{DBLP:journals/siamcomp/FeigeK03} (which is a weaker SDP that also runs in time $n^{O(d)}$).
However, a generalization of an argument of Kelner (see \pref{sec:kelner}) shows that this is \emph{not} the case for the SoS hierarchy, and our analysis of this certificate is tight.
Hence our work shows that to get stronger lower bounds for higher degree SOS it is necessary and sufficient to utilize more complicated constructions of certificates than those used for weaker hierarchies. Whether this additional complexity results in an asymptotically better tradeoff between the running time and clique size remains a tantalizing open question. 
\section{Technical Overview}
\newcommand{\J}{\mathcal{J}}
The SoS semidefinite programming hierarchy yields a convex programming relaxation for the planted clique problem. That is, we derive from the input graph $G$ a convex program $\cP_G$ such that if
the graph had a clique of size $\omega$ then $\cP_G$ is feasible. To show that the program \emph{fails} to solve the planted clique problem with parameter $\omega$, we show that with high probability there is a solution (known as a \emph{certificate}) for the program $\cP_G$ even when $G$ is a random graph from $G(n,1/2)$ (which in particular will not have a clique of size $\gg \log n$).

The solution to degree $d$ hierarchies can be thought of as a vector $X \in \R^{n^d}$. For \emph{linear programming} hierarchies this vector needs to satisfy various linear constraints, while for \emph{semi-definite programming} languages it also needs to satisfy constraints of the form $M \succeq 0$ where $M$ is a matrix where each entry is a linear function of $X$. In previous SoS lower bounds for problems such as random 3XOR/3SAT, Knapsack, and random constraint-satisfaction problems~\cite{Gri01,Schoenebeck08,BCK15}, the certificate $X$ was obtained in a fairly natural way, and the bulk of the work was in the analysis. In fact, in all those cases the certificate used in the SoS lower bounds was the same one that was used before for obtaining lower bounds for weaker hierarchies~\cite{DBLP:journals/siamcomp/FeigeK03}. The same holds for the previous works for the planted clique problem, where the works of \cite{MPW15,DM15} used a natural certificate which is a close variant of the certificate used by Feige and Krauthgamer~\cite{DBLP:journals/siamcomp/FeigeK03} for LS+ lower bounds, and showed that it satisfies the stronger conditions of the SoS program.

It is known that such an approach \emph{cannot} work to obtain a $\approx \sqrt{n}$ lower bound for the SoS program of degree $4$ and higher. That is, this natural certificate does \emph{not} satisfy the conditions of the SoS program. Hence  to obtain our lower bound for degree $4$ SoS we need to consider a more complicated certificate, that can be thought of as making a global ``correction'' to the simple certificate of \cite{MPW15,DM15}.  For higher degrees, we have not yet been able to analyze the corresponding complex certificate, but we are able to give a tight analysis of the simple certificate, showing that it certifies an $\omega \approx n^{1/(d + 1)}$ lower bound on degree $2d$ SoS relaxation. The key technical difficulty in both our work and prior ones is analyzing the spectrum of random matrices that have \emph{dependent} entries. Deshpande and Montanari~\cite{DM15} achieved such an analysis by a tour-de-force of case analysis specifically tailored for the degree $4$ case. However, the complexity of this argument makes it unwieldy to extend to either the case of the more complex certificate or the case of analyzing the simple certificate at higher degrees. Thus, key to our analysis is a more principled representation-theoretic approach, inspired by Grigoriev~\cite{Gri01}, to analyzing the spectrum of these kind of matrices. We hope this approach would be of use in further results for both the planted clique and other problems.

We now give an informal overview of the SoS program for planted clique, the \cite{MPW15,DM15} certificate, our correction to it, and our analysis. See \pref{sec:prelims} and \cite{BS14} for details.
\subsection{The SoS program for \MaxC}

Let $G = G([n],E)$ be any graph with the vertex set $[n]$ and edge set $E$. The following polynomial equations ensure that any assignment $x\in \R^n$ must be the characteristic vector of an $\omega$-sized clique in $G$:
\begin{align}
   x_i^2 = x_i & \text{  for all } i \in [n] \notag \\
  x_i \cdot x_j = 0 & \text{ for all } \{ i, j\} \notin E \notag \\
  \sum_{i = 1}^n x_i = & \omega\mper \label{eq:prog-formulation}
\end{align}
There is a related formulation (which we refer to as the ``optimization version") where the constraint $\sum_{i=1}^n x_i = \omega$ is not present and instead we have an objective function $ \sum_{i = 1}^n x_i$ to maximize. This latter formulation is used by \cite{DM15} in their work. It is also the program of focus in the work of Raghavendra and Schramm  \cite{RS15} who independently of us, show an almost optimal lower bound for the planted clique problem for the case of degree $4$ SoS relaxation. A point feasible for $\pref{eq:prog-formulation}$ is easily seen to be feasible for the optimization version with value $\omega$ and hence using the variant \pref{eq:prog-formulation} only makes our results stronger. It is unclear, however, whether an explicit constraint of $\sum_{i = 1}^n x_i = \omega$ adds more refutation power to the program. \footnote{The reason, as we describe when discussing pseudoexpectations, is that adding $\{p = 0\}$ as a constraint ensures that $\tE[qp] = 0$ for every $\deg(q) \leq d-\deg(p)$ in addition to $\tE[p] = 0$.}

The degree $d$ SoS hierarchy optimizes over an object called as degree-$d$ \emph{pseudo-expectation} or \emph{pseudo-distribution}. A \emph{degree-$d$ pseudo-expectation operator} for (\ref{eq:prog-formulation}) is a linear operator $\tE$ that behaves to some extent as the expectation operator for some distribution over $x\in\R^n$ that satisfies the conditions (\ref{eq:prog-formulation}). For example this operator will satisfy that $\tE \sum_{i=1}^n x_i = \omega$, $\tE x_i^2 = \tE x_i$, etc.. More formally, $\tE$ is a linear operator mapping every polynomial $P$ of degree at most $d$ into a number $\tE P$ such that $\tE 1 =1$, $\tE P^2 \geq 0$ for every $P$ of degree at most $d/2$, $\tE PQ = 0$ for every $Q$ of degree at most $d- \deg P$ and $P$ such that the constraint $\{P=0 \}$ appears in (\ref{eq:prog-formulation}).  Note that since the dimension of the set of $n$-variate polynomials of degree at most $d$ is at most $n^d$, the operator $\tE$ can be described as a vector in $\R^d$. Moreover, the constraints on $\tE$ can be captured by a semidefinite program, and this semidefinite program is in fact the SoS program. See \pref{sec:prelims}, the survey \cite{BS14} or the lecture notes \cite{Bar14} for more on the SoS hierarchy.

\subsection{The ``Simple Moments''}

To show a lower bound of $\omega$, we need to show that for a random graph $G$, we can find a degree $d$ pseudo-expectation operator that satisfies (\ref{eq:prog-formulation}). Both the papers \cite{MPW15} and \cite{DM15} utilize essentially the same operator, which we call here the ``simple moments''. It is arguably the most straightforward way to satisfy the conditions of (\ref{eq:prog-formulation}), and the bulk of the work is then in showing the positivity constraint that $\tE P^2 \geq 0$ for every $P$ of degree $\leq 2$ (in the degree $4$ case). \cite{DM15} shows that this will hold as long as $\omega \ll n^{1/3}$ and an argument of Kelner (see \pref{sec:Kelner-poly} below) shows that this is tight and in fact these simple moments fail to satisfy the positivity conditions for $\omega \gg n^{1/3}$.

To define a degree $d$ pseudo-expectation operator $\tE$, we need to choose some basis $\{ P_1 , \ldots, P_N \}$ for the set of polynomials of degree at most $d$ and define $\tE P_i$ for every $i$. The simplest basis is simply the monomial basis. Moreover, since our pseudo-expectation satisfies the constraints $\{ x_i^2=x_i \}$, we can restrict attention to \emph{multilinear} monomials, of the form $x_S = \prod_{i\in S} x_i$ for some $S \subseteq [n]$. Note also that the constraints $x_ix_j = 0$ for $\{i,j\}\notin E$ imply that we must define $\tE x_S = 0$ for every $S$ that is not a clique in $G$. Indeed, the pseudo-distribution $\{ x \}$ is supposed to mimic an actual distribution over the characteristic vectors of $\omega$-sized cliques in $G$, and note that in any such distribution it would hold that $\tE x_S =0$ when $S$ is not a clique.

The simplest form of such a pseudo-distribution is to set 
\[
\tE x_S = \begin{cases}
          0 & \text{$S$ is not a clique} \\
          \alpha_{|S|} & \text{otherwise} 
          \end{cases}
\]

where $\alpha_{|S|}$ is a constant depending only on the size of $S$. We can compute the value $\alpha_{|S|}$ by noting that we need to satisfy $\tE (\sum_i x_i)^\ell = \sum_{i_1,\ldots,i_\ell} \tE x_{i_1}\cdots x_{i_\ell} = \omega^\ell$ for every $\ell=1,\ldots, d$. Since there would be about
$\binom{n}{\ell}2^{-\binom{\ell}{2}}$ $\ell$-sized cliques in the graph $G$, the value $\alpha_\ell$ will be $\approx \left(\tfrac{\omega}{n}\right)^\ell$.\footnote{One actually needs to make some minor modifications to these moments to ensure they satisfy exactly the constraint $\sum x_i = \omega$ as is done in \cite{MPW15} and in our technical section. However, these corrections have very small magnitudes and so all the observations below apply equally well to the modified moments, and so we ignore this issue in this informal overview. }

This pseudo-distribution is essentially the same one used by Feige and Krauthgamer~\cite{DBLP:journals/siamcomp/FeigeK03} for LS+, where they were shown to be valid for the constraints of this problem as long as $\omega < \sqrt{n/2^{d+1}}$. Initially Meka and Wigderson conjectured that a similar bound holds for the SoS program, or in other words, that the $\binom{n}{\leq d/2} \times \binom{n}{\leq d/2}$ matrix $M$ where $M_{S,T} = \tE x_Sx_T$ for every $S,T \subseteq [n]$ of size $\leq d/2$ is positive semidefinite as long as $\omega \ll \sqrt{n}$. The Meka-Wigderson conjecture would have held if the off-diagonal part of $M$, which is a random matrix with \emph{dependent} entries, would have a spectral norm comparable with an \emph{independent} random matrix with entries of a similar magnitude. However, this turns out to fail in quite a strong way. An argument due to Jonathan Kelner, described in \cite{Bar14}, shows that (for $d=4$) the matrix $M$ is \emph{not} positive semidefinite as long as $\omega \gg n^{1/3}$. We review this argument below, as it is instructive for our correction.

\subsection{There is such a thing as too simple} \label{sec:Kelner-poly}

In the simple moments, every $4$-clique $S$ gets the same pseudo-expectation. In some sense these moments turn out to be ``too random'' in that they fail to account for some structure that the graph possesses. Specifically, for $i\in [n]$, consider the linear function $r_i(x) = \sum_j r_{i,j} x_j$ where $r_{i,j}$ equals $+1$ when $\{i,j\} \in E$, equals $-1$ when $\{i,j\}\notin E$, and equals $0$ when $i=j$. Now, consider the polynomial $P(x) = \sum_{i=1}^n r_i(x)^4$. For every $x$ that is the characteristic vector of an $\omega$-clique in $G$,  $P(x) \geq \omega(\omega-1)^4 \geq \omega^5/2$; 
indeed for every $i$ in the clique, $r_i(x)^4$ would equal $(\omega-1)^4$. On the other hand, for every $i$, let us consider the \emph{expectation} of $\tE r_i(x)^4$ taken over the choice of the random graph $G$. Note that in a random graph the  $r_{i,j}$'s are i.i.d. $\pm 1$ random variables, and hence
\begin{equation}
\E \tE r_i(x)^4 = \E \tE (\sum_j r_{i,j} x_j )^4 = \sum_{j_1,j_2,j_3,j_4 \neq i} \E r_{i,j_1} r_{i,j_2} r_{i,j_3} r_{i,j_4} \tE x_{j_1}x_{j_2}x_{j_3}x_{j_4} \;. \label{eq:rijs}
\end{equation}

Let us group the terms on the RHS of (\ref{eq:rijs}) based on the number of distinct $j_k$'s. There are $O(n^2)$ terms corresponding to two distinct $j_k$'s, each of them is multiplied by $\alpha_2 \approx \left(\tfrac{\omega}{n}\right)^2$ and so they contribute a total of $C \omega^2$ to the expectation for some constant $C$. In the terms corresponding to three or four $j_k$'s, there is always one variable $r_{i,j}$ that is not squared, and hence their contribution to the expectation is zero. There are $O(n)$ terms corresponding to a single $j_k$, each multiplied by $\tfrac{\omega}{n}$ and so their total contribution is at most $\omega$. We get that in expectation $\tE r_i(x)^4 \leq C \omega^2$ for some constant $C$ and by Markov this holds with high probability as well (for some different choice of the constant). 

The conclusion is that while for every $\omega$-sized clique $x$, $P(x) \geq \omega^5/2$, the simple moments satisfy that $\tE P(X) \leq C n \omega^2$.
When $\omega \gg n^{1/3}$ this yields a strong discrepancy between the value the simple moments give $P$ and the value that they should have given, had they corresponded to an actual distribution on $\omega$-sized cliques. This discrepancy can be massaged into a degree $2$ polynomial $Q$ such that $\tE Q^2 < 0$ for the simple moments when $\omega \gg n^{1/3}$, thus showing that in this case these moments do not satisfy the SoS program.

\subsection{Fixing the simple moments}

Our fix for the simple moments is directly motivated by the example above. We want to ensure that the polynomial $P$ will get a pseudo-expectation of $\approx \omega^5$, and that in fact for every $i$ $\tE r_i(x)^4$ will be roughly $\omega^5/n$.  The idea is to break the symmetry between different equal-sized cliques and give a significantly higher pseudo-expectation to cliques that are somewhat over-represented by these polynomials. Specifically for every set $S$, define $r_S = \sum_i \prod_{j\in S} r_{i,j}$. Note that $r_S$ is a sum of $n$  entries in $\{ \pm 1 \}$, and in a random graph it behaves roughly like a normal variable with mean $0$ and variance $n$. Roughly speaking, the corrected moments will set 
\[
\tE [x_S] = \alpha_{|S|} (1+r_S \omega / n)
\]
for every clique $S$. Note that when $\omega = \epsilon \sqrt{n}$, the correction factor would typically be of the form $1 \pm \Theta(\epsilon)$.\footnote{While it might seem that there is a chance for these pseudo-expectations to be negative, if $\omega < \sqrt{n}/\polylog(n)$ then it is exceedingly unlikely that there will exists an $S$ such that $|r_S|>n/\omega$, and so we ignore this issue in this overview.}

Computing the pseudo-expectation $\tE r_i(x)^4$ under the new moments we again get the expression
\[
\sum_{j_1,j_2,j_3,j_4} \E r_{\{ j_1,\ldots,j_4 \}} \tE x_{\{j_1,\ldots,j_4\}} \;.
\]
If we now focus on the contribution of the $n^4$ terms where the $j_k$'s are all distinct, we see that each such set $S$ yields the term
\[
\E r_S^2  \alpha_4 \omega / n \;.
\]
Since $\alpha_4 = C \omega^4/n^4$ we get that $\tE r_i(x)^4 = C \omega^5 / n$ as desired.

\subsection{Analyzing the corrected moments}

The above gives some intuition why the corrected moments might be better than the simple moments for one set of polynomials. But a priori it is not at all clear that those polynomials encapsulated all the issues with the simple moments. Moreover, it is also unclear whether or not the correction itself could introduce additional issues, and create new types of negative eigenvectors. Ruling out these two possibilities is the crux of our analysis. 


Here we discuss some key points from our analysis of $\tE$.
Since \cite{DM15} carried out a thorough analysis of the degree-4 simple moments, we begin by reviewing their approach.

\paragraph{Approach of \cite{DM15}}
The PSDness of $\tE$ reduces to proving PSDness of a related matrix $\M \in \R^{\nchoose{2} \times \nchoose{2}}, \M(S,T) = \pE x^S x^T$.
The eigendecomposition of $E = \E_G[\M]$ has three eigenspaces $V_0, V_1, V_2$ and eigenvalues $\lambda_0 \approx \omega^4 / n^2$ on $V_0$, $\lambda_1 \approx \omega^3 / n^2$ on $V_1$, and $\lambda_2 \approx \omega^2 / n^2$ on $V_2$. 
Next, write $\M$ as a block matrix with blocks $\M_{ij} = \Pi_{V_i} \M \Pi_{V_j}$ where $\Pi_V$ projects to the subspace $V$.
On the diagonal blocks, $E$ contributes large positive eigenvalues.
If the on-diagonal blocks $M_{ii} \succeq \lambda_i I$ for some $\lambda_i$'s so that $\|M_{ij}\| \ll \sqrt{\lambda_i \lambda_j}$, then $\M$ will be PSD.
\begin{figure}[h]
\begin{center}
\includegraphics[scale=0.5]{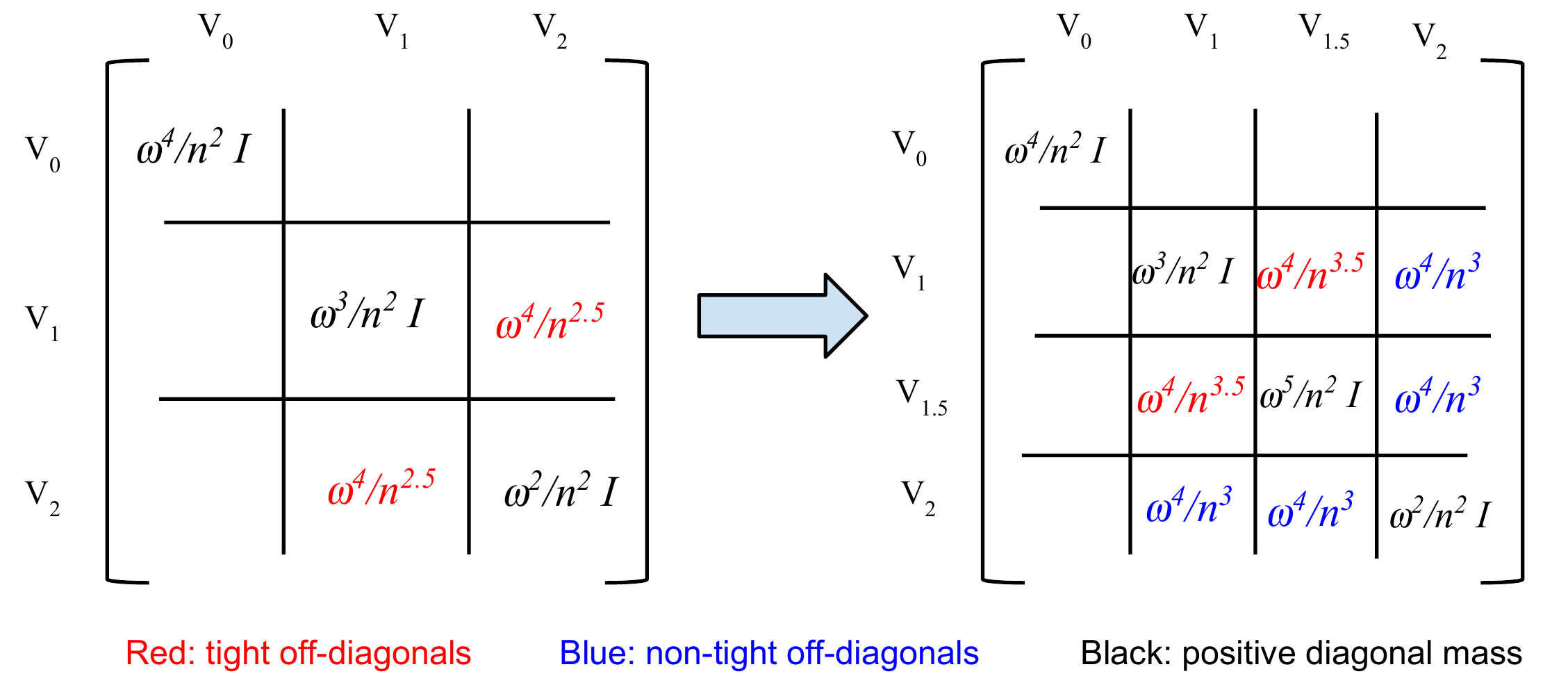}
\end{center}
\caption{The block matrix / subspace decomposition view, before and after the correction. Uninteresting entries left empty.}
\end{figure}

Because of the dependencies in the random matrix $\cM$, the deviation from expectation varies according to the eigenspace.
Thus, in \cite{DM15}, the deviation from the expectation is analyzed by first decomposing along the eigenspaces $V_0,V_1,V_2$.

A second technical idea is required to carry out this decomposition.
Because of the symmetries present in the spaces $V_0, V_1, V_2$, this decomposition is \emph{very nearly the same} as splitting up the matrix $\cM$ in an ostensibly unrelated way.
Each entry $\cM(I,J)$ of $\cM$ for $I, J \in \nchoose{2}$ is the $0/1$ indicator for the presence of a clique on $I \cup J$.
This indicator is just the AND of all the $\pm 1$ indicators $g_b$ for the presence of the edges $b \in \sE(I \cup J)$.
Taking a Fourier decomposition of this suggests a way to decompose $\M$ as $\cM = \sum_{\text{subsets $S$ of edges on $4$ nodes}} \cM_S$,\footnote{This is not quite the whole picture, see \pref{sec:D-decomp}.} where the matrix $\M_S$ corresponds to the Fourier character $S$.
The matrices $\cM_S$ can be matched up to the subspaces $V_0, V_1, V_2$ in such a way that those matrices with \emph{larger spectral norm} (corresponding to larger deviations from expectation) have subspaces with \emph{smaller eigenvalues} in their kernels!

\paragraph{Pinpointing the failure of the simple moments}
The foregoing is missing one subtlety.
Some monomial matrices $\cM_S$ do not match nicely to a single subspace.
Instead they form cross terms: for example, having $V_2$ in the left kernel but not in the right kernel (not all these matrices need be symmetric).
In fact, it is just such a matrix which keeps the simple moments from remaining PSD beyond $\omega \approx n^{1/3}$.

For the four nodes $a_1,a_2,b_1,b_2$, consider the monomial $g_{a_1,b_1}g_{a_1,b_2}$ and the corresponding matrix $\cM_S(\{a_1,a_2\},\{b_1,b_2\}) \approx (\omega^4 / n^4) g_{a_1,b_1}g_{a_1,b_2}$.
The entries $\cM_S$ do not depend on $a_1$, and so there are many repeated rows, which creates a much larger spectral norm for $\cM_S$ than if it had independent entries.
\cite{DM15} prove the (tight) bound $\|\cM_S\| \approx \omega^4 / n^{7/2}$.
At the same time, it turns out only to have $V_2$ in its left kernel, not its right one.
Appealing to the above picture, in order to have $\omega^4 / n^{7/2} \ll \omega^{5/2} / n^2$, we must have $\omega \ll n^{1/3}$.

\paragraph{Analysis of the correction}
We make one further observation about the matrix $\cM_S'$ from the previous section: its rows are the tensor squares of the $\pm 1$ neighborhood indicator vectors $r_i$ from above.
Our fix to the simple moments, described above as adjusting individual pseudo-expectations, amounts roughly to adding to $\cM$ the matrix $(\omega^5 / n^5) \sum_i (r_i^{\tensor 2})(r_i^{\tensor 2})^\top$.
This carves out of $V_2$ (our worst subspace from an eigenvalue perspective) a new subspaces space $V_{1.5}$ with eigenvalues lower-bounded by $\lambda_{1.5} \approx \omega^5 /n \gg \omega^2 / n^2$.
Now instead of matching the bad matrix $\cM_S$ to $V_1$ and $V_2$ as a cross term, we can match it to $V_1$ and $V_{1.5}$ as a cross term.
Then we only need $\omega^4 / n^{7/2} \ll \sqrt{\lambda_1 \lambda_{1.5}} \approx \omega^4 / n^{7/2}$.
With some care in the details, the above picture can be made precise.

However, a crucial point is that the matrix $N = (\omega^5 / n^5) \sum_i (r_i^{\tensor 2})(r_i^{\tensor 2})^\top$ doesn't satisfy the clique constraints (in that all entries $I,J$ with $I \cup J$ not a clique should be $0$). A chunk of our proof goes into analyzing the discrepancy between the matrix $N$ and its zeroed out version. Our analysis here requires the use of new combinatorial tools (\pref{sec:spectralnorm}) combined with the trace moment method.

\paragraph{Symmetries of Eigenspaces and Tight Analysis of MPW Operator for Higher Degrees} As we have alluded to already, a key technical step in our proof is to show that certain Fourier-decomposed matrices of the form discussed above have some of the subspaces $V_0, V_1, V_2$ in their kernels.
In the analysis of simple moments for degree $4$, \cite{DM15} use explicit entries for canonical forms of eigenvectors in $V_0, V_1$ to accomplish this.
However this approach hits analytical roadblocks for the analysis in case of higher degrees.
Canonical forms for the eigenvectors are hard to pin down explicitly from the literature in algebraic combinatorics.

To mitigate this difficulty, we take a more principled approach to understand the eigenspaces $V_0,V_1, \ldots, V_d \subseteq \R^{\nchoose{d}}$ in terms of their symmetries.
Using techniques from basic representation theory of finite groups, we arrive at an explicit family of symmetries that express any vector in $V_i \subseteq \R^{\nchoose{d}}$ as an explicit linear transformation of some vector in $\R^{\nchoose{i}}$.
It also shows that any $v \in V_i \subseteq \R^{\nchoose{d}}$ has the form $\iprod{v,x^{\tensor d}}$ that's essentially the mutilinearization of $(\sum_j x_j)^{d-i} p(x)$ for some $p$.

A similar approach was utilized heavily by Grigoriev \cite{Gri01} to prove a sum of squares lower bound for the knapsack problem.
While for degree $4$ either explicit eigenvectors or our approach will work, although the latter takes some more elbow grease, ours is absolutely vital for our tight analysis of the MPW moments for the higher degrees.
We hope that such an approach will be useful for proving better (approaching $\omega \approx \sqrt n$) integrality gap for Sum of Squares relaxations of higher degree for Planted Clique and other related problems.

The analysis of the MPW operator at higher degrees also presents other new challenges that do not show up in the special case of degree $4$ analyzed in \cite{DM15}. \cite{DM15} deal with the optimization version of the degree $4$ SOS program which could be potentially weaker than the one we analyze here (and thus our lower bound is technically stronger). Working with the ``optimization'' version simplifies the analysis in \cite{DM15} a little bit as the matrix $\M$ has entries that only have local dependence on the graph $G$. We explicitly work with the feasibility version of the degree 4 SOS program and thus, must deal with the additional complexity of the entries of $\M$ having a global dependence. As in \cite{MPW15}, we deal with this situation by separating $\M$ into matrices $L$ and $\Delta$ such that $L$ has only local dependence on the graph $G$. \cite{MPW15} deal with $\Delta$ by a simple entrywise bound, however, employing such a bound yields no improvement over the bound proved in \cite{MPW15} for us. It turns out that we have to do a fine grained analysis of the $\Delta$ matrix itself by a decomposition for $\Delta$ such that each piece is essentially only locally dependent on the graph. Once we have such a decomposition, our ideas from the analysis of $L$ can be extended to that setting as well.

Finally, our argument for analyzing the spectral norms of each of the pieces of encountered in the decompositions also needs to be much more general than in case of \cite{DM15} to handle higher degrees. For this, we identify a simple combinatorial structure that controls the norm bounds and allows a general hammer for computing the norms of all the matrices that appear in this analysis. Our proofs here are based on the trace power method and build on the combinatorial techniques in \cite{DM15}.

\subsection{Preview of Technical Toolkit}
In this section we give a preview of the key lemmas that allow us to carry out the analyses described thus far.
We have simplified some issues for the sake of exposition; details may be found in \pref{sec:tools}.

We are concerned with the matrices in the aforementioned Fourier decomposition.
Let $B \subseteq [d] \times [d]$ be a bipartite graph on $2d$ vertices.
Let $Q_B$ be an $\nchoose{d} \times \nchoose{d}$ matrix with entries
\[
  Q_B(I,J) = (-1)^{\text{number of $B$-edges which are not $G$-edges when the left vertex set of $B$ is replaced by $I$ and the right one with $J$}}
\]
(We are ignoring what happens if $I \cap J \neq \emptyset$.)
\begin{figure}[h]
\begin{center}
\includegraphics[scale=0.5]{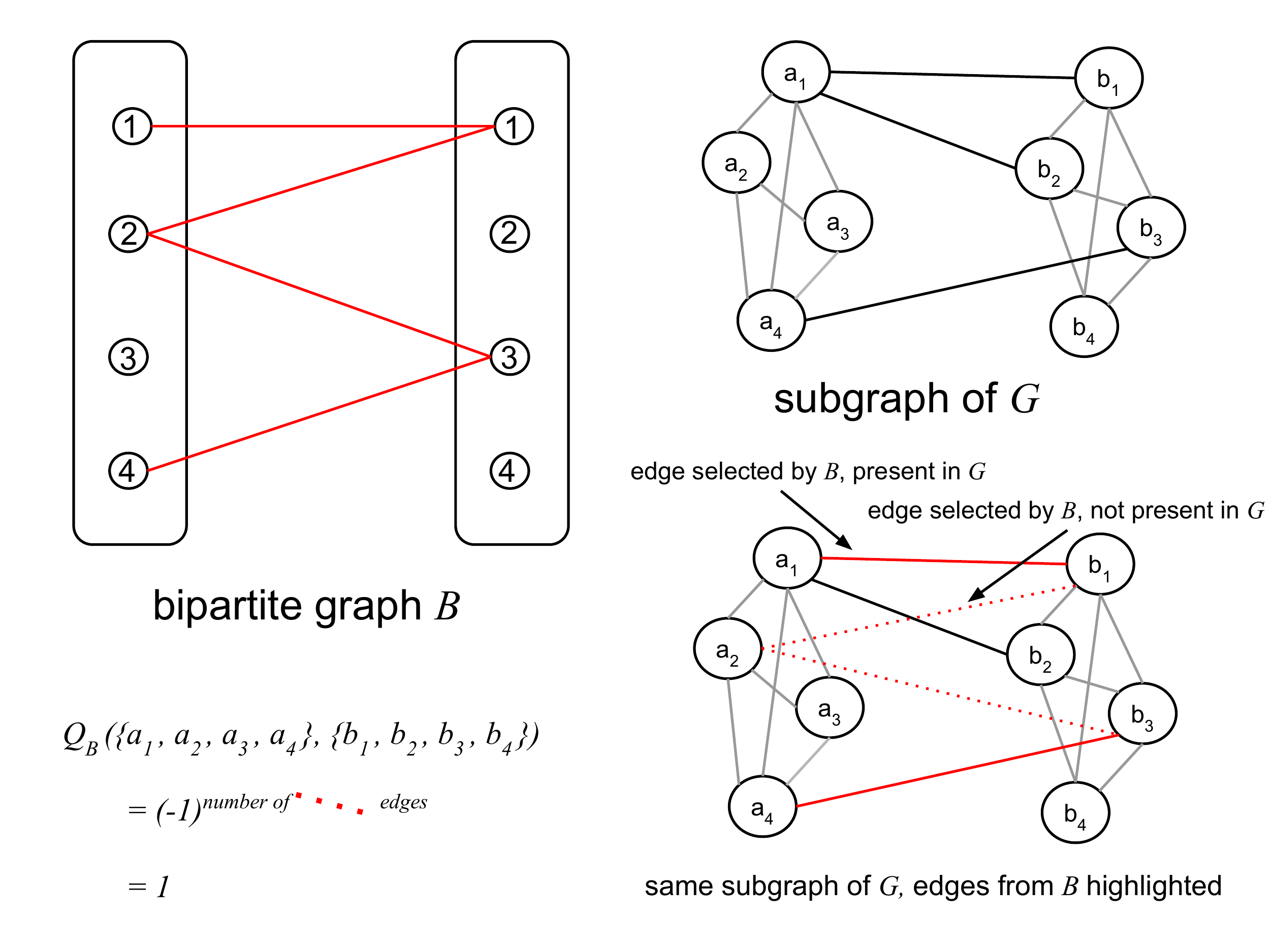}
\end{center}
\caption{Example $B$ and $Q_B$ where $f$ is parity of edges, $d = 4$. \pref{lem:kernelsymoperators-intro} says that $\Pi_4 Q_B = 0$ and $Q_Br \Pi_4 = Q_B \Pi_3 = 0$. \pref{lem:fourier-norm-mpw-disjoint-intro} says that $\|Q_B\|  \approx n^3$ with high probability when $G \sim G(n,1/2)$, since $B$ contains a $2$-matching.}
\end{figure}

This first lemma bounds the spectral norm of such a matrix in terms of the shape of $B$.
\begin{lemma}[Informal version of \pref{lem:fourier-norm-mpw-disjoint}\footnote{We use and prove only the cases $c = 1, c = 2$, but the general version follows from almost identical techniques.}]
  \label{lem:fourier-norm-mpw-disjoint-intro}
  Let $c$ be the number of edges in the maximum matching in $B$.
  With high probability, $\|Q_B\| = \tO(n^{d - c/2})$.
\end{lemma}

We also want to show that these matrices have nontrivial kernels, so we can bound their negative eigenvalues against the parts of the simple moments with larger positive eigenvalues.
The following allows us to carry out this matching of Fourier decomposition matrices to eigenspaces $V_0,\ldots,V_d$ of the expectation matrix.

\begin{lemma}[Informal version of \pref{lem:kernelsymoperators}] \label{lem:kernelsymoperators-intro}
 Let $B_{\ell}$ ($B_r$) be the subset of vertices on the left (right, respectively) hand side with non zero degrees in $B$.
 Let $\Pi_i$ be the projector to $V_i$.
 Then,
  \begin{enumerate}
    \item For every $j > |B_\ell|$,  $$\Pi_j Q_B  = 0,$$
    \item For every $i > |B_r|$, $$ Q_B \Pi_i = 0.$$
  \end{enumerate}
\end{lemma}

The maximum matching cannot be too small when $|B_\ell| + |B_r|$ is large, which allows us to combine these lemmas for every $B$, either $Q_B$ has small spectral norm or its kernel contains the spaces where the diagonal of the expectation matrix is small.


%

\subsection{Related Work}
There's a large amount of work on understanding Linear and Semidefinite Programming based hierarchies. A detailed survey on the sum of squares hierarchy and references to works related can be found in \cite{BS14}. 
The earliest works on proving SoS lower bounds were due to Grigoriev \cite{Gri01, Gri01b} who showed that degree $\Omega(n)$ SoS does not beat the random assignment for 3SAT or 3XOR even on random instances from a natural distribution. Some of these lower bounds were rediscovered by Schoenebeck \cite{Schoenebeck08}. Lower bounds for SOS essentially rely on gadget reductions from 3SAT or 3XOR and this approach has been understood in some detail \cite{Tul09,Bhak12}. An exception to this methodology is the recent work of Barak et al. in proving SoS lower bounds for pairwise independent CSPs \cite{BCK15}. Even though the lower bounds for CSPs are for random instances, the average-case nature of the problem does not show up as a main analytic issue. There has recently been a surge of interest in understanding the performance of SoS on average-case problems of interest in machine learning, both in proving upper and lower bounds \cite{Hop15, BarakM15, MaW15, GeM15, BKS15}.

For the planted clique problem, Feige and Krauthgamer gave an analysis of the performance of the LS+ semidefinite heirarchy tight to within constants \cite{FK00, DBLP:journals/siamcomp/FeigeK03} giving the state of the art algorithm for finding planted cliques in any fixed polynomial time. Other algorithmic techniques not based on convex relaxations have been studied and shown to fail for planted clique beyond $\omega \approx \sqrt n$,
most prominantly Markov Chain Monte Carlo (MCMC) \cite{DBLP:journals/rsa/Jerrum92}. Recently, Feldman et. al.  \cite{Feldman13} showed a lower bound for (a variant of) the planted clique problem in the restricted class of \emph{statistical algorithms} that generalize MCMC based methods and many other algorithmic techniques. Frieze and Kannan \cite{FKan} proposed an approach for the planted clique problem through optimizing a degree-$3$ polynomial related to the random graph. Such polynomials are NP hard to optimize in the worst case but the belief is that the random nature of the polynomials might be helpful. This approach was generalized to higher degree polynomials by Brubaker and Vempala \cite{BV09}. 

There has also been recent work on variants of the problem that define Gaussian versions of the planted clique and more generally, the hidden submatrix problems showing, for example, strong inditinguishability results about the spectrum of the associated matrices with and without planting \cite{Mon14}. Finally, the present work builds heavily on independent papers of Meka, Potechin, and Wigderson \cite{MPW15} and Deshpande and Montanari \cite{DM15}, which we have already thoroughly discussed.

\paragraph{Overview of Rest of the Paper}
\pref{sec:prelims} contains preliminaries.
\pref{sec:MPWop} contains definitions and the necessary background on the simple moments, a.k.a. the MPW operator.
\pref{sec:deg4} contains the formal definition of our corrected degree $4$ moments.
\pref{sec:tools} lays out the technical framework for the analyses of the corrected degree $4$ moments and for the tightened bounds on the MPW operator at higher degrees.
Here we define the Fourier decompositions alluded to above and carry out representation-theoretic arguments about their kernels.
\pref{sec:degree-d} and \pref{sec:deg4-psd} use the tools we have built thus far to prove the main theorems.
In \pref{sec:degreevariancecalculations} we prove a technical concentration result for small subgraphs of $G(n,1/2)$ required for the analysis.
In \pref{sec:kelner} we sketch Kelner's argument showing that our analysis of the MPW moments is nearly optimal.

%
%

\section{Preliminaries}
\label{sec:prelims}
We will use the following general notation in the paper.
\begin{enumerate}
\item $G$ will denote a draw from $G(n,\frac{1}{2})$ unless otherwise stated. 
\item $||x||_2= ||x||$ denotes the Euclidean $2$ norm of a vector $x \in \R^m$.
\item For a square symmetric matrices $Q,R$, we write $Q \succeq R$ to mean $Q-R$ is positive semidefinite.
\item For any matrix $M$, $||M||$ denotes its largest singular value, or, equivalently, $||M|| = \max_{x:||x||_2 = 1} ||Mx||^2.$
\item For matrices $M,N$ of same dimensions, $M \odot N$ denotes their entrywise or \emph{Hadamard} product, i.e., $(M \odot N )(I,J) = M(I,J) \cdot N(I,J)$ for every $I,J$.
\item For a graph $G$ and any set of two vertices of $G$, $e$, $g_e$ denotes the $\on$ indicator of the edge $e$ being present in $G$. That is, $g_e = +1$ if $e$ is an edge in $G$ and $-1$ otherwise.
\item For a set $I$ of vertices of $G$, $\sE(I) = {I \choose 2}$, the set of all pairs from $I$.
\item For a pair of subsets of vertices $I,J$ of $G$, $\sE_{ext}(I,J) = \sE(I \cup J) \setminus (\sE(I) \cup \sE(J)),$ the set of cut edges between $I$ and $J$.
\item For a subspace $V$, $\Pi_V$ denotes the projector to $V$.
\end{enumerate}

Following \cite{BBHKSZ12} and many subsequent papers, we work with SoS using the language of \emph{pseudo-expectations}.
\begin{definition}[Pseudo Expectation]
A linear operator $\tE: \P_d^n \rightarrow \R$ is a degree $d$ pseudo-expectation operator if it satisfies: \\
\textbf{Normalization: } $\tE[1] = 1$ where on the LHS $1$ denotes the constant polynomial $p$ such that $p(x)  = 1$. \\
\textbf{Positivity (or positive semidefiniteness):}  $\tE[ p^2] \geq 0$ for every $p \in \P_{d/2}^n$.  \label{def:pseudo-expectation}
\end{definition}
For every polynomial $p \in \P_{d}^n$, we say that $\tE$ satisfies the constraint $\{p = 0\}$ if $\tE[ pq] = 0$ for every $q \in \P_{d-\deg(p)}^n$. The sum of squares hierarchy can be thought of as optimizing over pseudo expectations (see \cite{BS14} and the lecture notes \cite{Bar14}).
\begin{fact}
  Let $p_0,\ldots,p_k \in \P_d^n$.
  If a pseudo-expectation satisfying the constraints $\{ p_0 =0 ,\ldots,p_k = 0 \}$ exists, it can be found in time $n^{O(d)}$.
  If none exists, a certificate of infeasibility of these equations is found instead.
\end{fact}

\begin{fact}[Special Case of Gershgorin Circle Theorem] \label{fact:Gershgorin}
For any square matrix $M \in \R^{N\times N}$, $$||M|| \leq \max_{i \in [N]} \left( \sum_{j = 1}^N |M_{ij}| \right).$$
\end{fact}
The following observations (actually both the same observation in different forms) will come in handy in our analysis.
\begin{lemma}
  \label{lem:blockpsd}
  Let $M \in \R^{n \times n}$ be self-adjoint.
  Let $W_1,\ldots,W_k$ be an orthogonal decomposition of $\R^n$ into subspaces.
  Let $P_i$ be the projector to $W_i$.
  Let $\lambda_1,\ldots,\lambda_k \geq 0$ and suppose for all $i,j \leq k$
  \[
    P_i M P_i \succeq \lambda_i P_i \quad \text{and when $i \neq j$} \quad \|P_i M P_j\| \leq \frac 2 k \sqrt{\lambda_i \lambda_j}\mper
  \]
  Then $M$ is PSD.
\end{lemma}
\begin{proof}
  Consider a unit vector $x \in \R^n$ and write it as $x = \sum_{i \in [k]} P_i x$.
  We expand
  \[
    \iprod{x, Mx} = \sum_{i,j \in [k]} \iprod{x, P_i M P_j x} \geq \sum_i \lambda_i \|P_i x\|^2 - \frac 2 k \sum_{i \neq j} \|P_i x\| \|P_j x\| \sqrt{\lambda_i \lambda_i}
  \]
  by our assumptions on $P_i M P_j$ and Cauchy-Schwarz.
  For each $i,j$, we know $\frac 1 k (\|P_i x\| \lambda_i + \|P_i x\| \lambda_j) \geq \frac 2 k \|P_i x \| \|P_j x\| \sqrt{\lambda_i \lambda_j}$, which implies that the whole expression is nonnegative.
\end{proof}

\begin{lemma}
  \label{lem:matrix-cs}
  Let $M \in \R^{n \times n}$ be self-adjoint.
  Let $V_1, V_2$ be subspaces of $\R^n$.
  Let $\Pi_{V_i}$ be the projector to $V_i$.
  If $\sqrt{\lambda_1 \lambda_2 } \geq \|\Pi_{V_1} M \Pi_{V_2} \|$, then
  \[
    \Pi_{V_1} M \Pi_{V_2} + \Pi_{V_2} M \Pi_{V_1} \preceq \lambda_1 \Pi_1 + \lambda_2 \Pi_2\mper
  \]
\end{lemma}
\begin{proof}
  For any $x \in \R^n$ we have
  \begin{align*}
    2 \iprod{x, \Pi_{V_1} M \Pi_{V_2} x} & \leq  2 \|\Pi_{V_1} M \Pi_{V_2}\| \cdot \|\Pi_{V_1} x \| \cdot \|\Pi_{V_2} x\|\\ 
    & \leq 2 (\sqrt{\lambda_1} \|\Pi_{V_1} x \|) (\sqrt{\lambda_2}) \|\Pi_{V_1} x\|)\\
    & \leq \lambda_1 \|\Pi_{V_1} x\|^2 + \lambda_2 \|\Pi_{V_2} x \|^2 \\
    & = \lambda_1 \iprod{x, \Pi_{V_1} x} + \lambda_2 \iprod{x, \Pi_{V_2} x}\mper\qedhere
  \end{align*}
\end{proof}

\section{The MPW Operator}
\label{sec:MPWop}
In this section, we describe the linear operator $\tE : \P_{2d}^n \rightarrow \R$ for every $d \leq O(\log{n})$ used by \cite{MPW15} and \cite{DM15}.
It is this operator which we will show gives an integrality gap for degree-$2d$ SoS when $\omega \ll n^{1/(d + 1)}$, and it will also form the basis for our improved integrality gap witness at degree four.

The main task in such a setting is to show that $\tE$ is positive semidefinite.
$\tE$ is same as the operator used by $\cite{MPW15}$ who showed that for $\omega \leq \Theta(n^{\frac{1}{2d}})$ and graph $G$ drawn at random from $G(n,\frac{1}{2})$, $\tE$ is a degree $d$ pseudo-expectation that satisfies all the constraints with high probability over the draw of $G$.
In other words, they showed that $\tE$ is a `cheating" solution that ``thinks" that a random graph has a clique of size $\sim n^{1/2d}$ with high probability.

For any set $I \subseteq [n]$, let $x_I$ be the monomial $\Pi_{i \in I} x_i$.
\begin{definition}[MPW operator for clique size $\omega$]
\label{def:pseudoexpectation}
For a graph $G$ on $n$ vertices and a parameter $\omega > 0$ we define a linear functional $\pE : \P_{2d}^n \rightarrow \R$.
To describe $\tE$ it is enough to describe its values on every monomial $x_I$ for $I \subseteq [n]$, $|I| \leq d$.
Towards this goal, for any set $I \subseteq [n]$, $|I| \leq d$, we define
$$\deg_G(I) = \left| \{S \subseteq [n]: I \subseteq S \text{, } |S|  = 2d \text{, $S$ is a clique  in $G$ }\} \right|.$$
Further, we set $C_{2d} = C_{2d}(G)$ be the number of $2d$-cliques in $G$.

For every $I \subseteq [n]$, we define:
\begin{equation}
  \tE[x_I] = \frac{\deg_G(I)}{C_{2d}} \cdot \frac{{\omega \choose {|I|}}}{{{2d} \choose {|I|}}}. \end{equation}
Our definition of $\tE$ is the same as the one used in $\cite{MPW15}$ up to normalization (we explicitly satisfy the normalization condition $\tE[1]=1$).
When $\omega$ is chosen so that $\pE$ is PSD, we often call it the MPW pseudo distribution.
\end{definition}
It is easy to check that the linear operator $\tE$ satisfies the constraints in \pref{eq:prog-formulation} which we record as the following fact:
\begin{fact}
For any graph $G$, $\tE$ defined by \pref{def:pseudoexpectation} satisfies the constraints described in \pref{eq:prog-formulation}.
\end{fact}

The main task then is to show that $\tE$ is PSD for appropriate range of $\omega$. This task is simplified by another observation from \cite{MPW15} that we state next.
\begin{fact}[Corollary 2.4 in \cite{MPW15}]
\label{fact:homog}
For $\tE$ of degree $d$ defined in \pref{def:pseudoexpectation}, $\tE$  is positive semidefinite iff $\tE[ p^2] \geq 0$ for every multilinear, homogeneous polynomial $p$ of degree $d$.
\end{fact}

We define a matrix $\M \in \R^{{\nchoose{d} } \times {\nchoose {d}}}$ such that  $I, J \in \nchoose{d}$, \begin{equation} \M(I,J) = \deg_G(I \cup J) \frac{{\omega \choose {|I \cup J|}}}{{{2d} \choose {|I \cup J|}}}. \label{def:momentmatrix} \end{equation}
Then, from the fact above, showing that $\tE$ is PSD is equivalent to proving that $\frac{1}{C_d} \M$ is PSD. The goal of the next section is to establish that with high probability over the draw of $G \sim G(n,1/2)$, $\M$ is PSD for $\omega \leq \tO(n^{\frac{1}{d +1}})$. This immediately also shows that $\frac{1}{C_d} \M$ is PSD with high probability completing the proof. 

\begin{theorem}
\label{thm:main-degree-d}
With probability at least $1- 1/n$ over the draw of $G \sim G(n,\frac{1}{2})$ and $d  = o(\sqrt{\log{(n)}})$ $\M$ defined by \pref{def:momentmatrix} is PSD whenever $\omega \leq \tO(n^{\frac{1}{d +1}})$. \label{thm:main1}
\end{theorem}

Our analysis improves upon the analysis in \cite{MPW15} and generalizes the improved analysis for the special case of $d = 2$ done in $\cite{DM15}$. By a generalization of the counter example due to Kelner, our analysis can actually be shown to be tight. We defer the details of the counter example to the full version. In the remaining part of this section, we begin the task of proving \pref{thm:main1} by introducing certain simplifications and computing the eigen values of the expected value of the matrix  under $G(n, \frac{1}{2})$.
\subsection{Reduction to PSDness of $\M'$}
The presence of some zero rows in $\M$ (corresponding to index sets $S$ that are not cliques in $G$) poses a problem in analyzing its spectrum.  As in \cite{MPW15}, we evade this issue by working with $\M' \in \R^{\nchoose{d} \times \nchoose{d}}$ obtained by filling in the zero rows of $\M$ while not affecting the non zero rows of $\M$. Since the non zero part of $\M$ (for any $G$) is a sub matrix of $\M'$, proving PSDness of $\M'$ is enough. We describe $\M'$ next and begin by setting up some notation towards that goal:

For any $0 \leq i \leq d$, let $$\beta(i) = \frac{{\omega \choose {2d-i}}}{{{2d} \choose {2d-i}}}.$$ 

\begin{definition}[``Filled-in'' matrix]
  \label{def:MPW-Mprime}
 Let $\M_T \in \R^{\nchoose{d} \times \nchoose{d}}$ be defined by $\M_T(I, J) = \beta( |I \cap J|)$ whenever $I \cup J \subseteq T$ and $\sE(T) \setminus \left(\sE(I) \cup \sE(J) \right) \subseteq E$. We define the filled in matrix $\M'$ as: $$\M' = \sum_{T: |T| = 2d} \M_T.$$
\end{definition}

Observe that for any $I,J$, $\M'(I,J)$ is chosen so as to depend only on the edges with one end point in $I$ and the other in $J$. Intuitively, this corresponds to thinking of $I$ and $J$ as being cliques in $G$ by addition of some edges. Moreoever, $\M'(I,J)$ is chosen so that $\M'(I,J)=M(I,J)$ whenever $I,J$ are actually cliques in $G$. Thus, as noted above, we have the following fact (which is Lemma 5.1 in \cite{MPW15}).

\begin{fact}[Lemma 5.1 in \cite{MPW15}]
$\M$ is PSD if $\M'$ is PSD.
\end{fact}

To analyze $\M'$ we decompose into two parts initially writing $\M' = E + D$ where $E = \E_{G \sim G(n, \frac{1}{2})}[ \M']$. We show that $E$ is PSD with all eigenvalues bounded away from $0$ in the next subsection following which we analyze the deviation $D = \M' - E$ by writing it as a sum of various pieces and decomposing the action of each piece along the eigenspaces of $E$ in \pref{sec:degree-d}. 

Thus, the following lemma completes the proof of \pref{thm:main1}.
\begin{lemma}[$\M'$ is PSD] \label{lem:M'-is-PSD}
With probability at least $1-1/n$ over the draw of $G \sim G(n,\frac{1}{2})$ for $d = o(\sqrt{\log{(n)}})$, $\M'$ defined by \pref{def:MPW-Mprime} satisfies $\M' \succeq 0$ whenever $\omega \leq \tO(n^{\frac{1}{d +1}}).$
\end{lemma} 

\subsection{The Expectation Matrix}
The minimum eigen values of the expectation matrix $E = \E[ \M']$ was analyzed in \cite{MPW15} via known results about the Johnson scheme matrices. The same proof also yields all the eigenvalues of $E$ which we note here. 

We first describe the entries of the matrix $\M'$.
\begin{fact}[Entries of $\M'$, see Claim 7.3 in \cite{MPW15}]
For every $I, J \in \nchoose{d}$ and $E = \E[M']$, $$E(I,J) = {{n-|I \cup J|} \choose {2d - |I \cup J|}} \cdot \frac{{\omega \choose {|I \cup J|}}}{{2d \choose {|I \cup J|}}} \cdot 2^{-d^2 - {{|I \cap J|} \choose 2}}.$$
\end{fact}

Next, we need a basic fact about the (shared) eigenspaces of all the set symmetric matrices, in particular, their number and dimensions which follows from the following well known result from classical theory of Johnson schemes.
\begin{fact}[Lemma 6.6 of \cite{MPW15}] \label{fact:johnsondecomp}
Fix $n,d \leq n/2$ and let $\J= \J(n,d)$ be the set of all set symmetric matrices in $\R^{\nchoose{d} \times \nchoose{d}}$. Then, there exist subspaces $V_0, V_1, \ldots, V_d \in \R^{\nchoose{d}}$ that are orthogonal to each other such that:
\begin{enumerate}
\item $V_0, V_1, \ldots, V_d$ are eigen spaces for every $J \in \J$ and are isomorphic to distinct irreducible representations of the symmetric group $\S_n$ (See \pref{sec:tools} for definitions). \
\item For $0 \leq j \leq d$, $dim(V_j) = {n \choose j} - {n \choose {j-1}}$. 
\end{enumerate}
\end{fact}
Using a nice basis for the matrices in $\J$, one can obtain the following estimates of the eigenvalues of $E$ on $V_i$ for each $0 \leq i \leq d$:

\begin{lemma}[Eigenvalues of $E$]
\label{lem:Eeigenvals}
Let $\omega < \frac{n-2d}{3d2^{d-1}}$ and $d \leq \omega/2$. Let $\lambda_j(E)$ be the eigenvalue of $E$ on $V_j$ as defined in \pref{fact:johnsondecomp}. Then, $$\lambda_j(E) \geq \frac{1}{2} \cdot {{n-2d+j} \choose j} \cdot \frac{{\omega \choose {2d-j}}}{{{2d} \choose {2d-j}}} \cdot 2^{-d^2 - {j \choose 2}} \cdot {{n-t-j} \choose {d-t}} \cdot {{d-j} \choose  {t-j}} \geq 2^{-O(d^2)} \cdot n^d \cdot \omega^{2d-j}.$$
\end{lemma}

\section{The Corrected Operator for Degree Four}
\label{sec:deg4}
In this section, we present the pseudodistribution that we will use to show an almost optimal lower bound on the degree $4$ SOS algorithm. Our pseudodistribution is obtained by ``correcting" the one described in the previous section. The correction itself is inspired by an explicit polynomial described by Kelner who showed that the pseudodistribution from the previous section for degree $4$ does not satisfy positivity for $\omega \gg n^{1/3}$.

We now lay some groundwork for defining our modified operator. In the following, we will always work on a fixed graph $G$ on $[n]$ and use $\tE_0$ to denote the MPW pseudoexpectation operator for $d=2$. We start by defining a specific neighborhood indicator vector for every vertex in $G$. For a vertex $s \in G$, let the vector $r_s \in \R^n$ be given by
  \[
    r_s(j) = \begin{cases}
    1 & \text{ if $s \sim j$}\\
    -1 & \text{ if $s \not \sim j$}\\
    0 & \text{ if $s = j$}
    \end{cases}\mper
  \]
  
Next, we define the additive correction $\cL$ to the MPW pseudoexpectation operator $\tE$. $\cL$ will be a linear operator on the space of homogeneous degree $4$ polynomials. Because of linearity, it is enough to define $\L$ on the basis of all monomials of degree $4$.
\begin{definition}[Correction Term]
  \label{def:update}Let $\gamma > 0$ be a real parameter to be chosen later.
Let $\cL$ be the linear operator on the linear space of homogeneous, multilinear polynomials of degree $4$ such that:
  \[
    \cL [x_i x_j x_k x_\ell] = \begin{cases}
            \gamma (\tfrac {\omega} n )^5 \sum_s r_s(i) r_s(j) r_s(k) r_s(\ell) & \text{ if $i,j,k,\ell$ form a clique in $G$}\\
            0 & \text{ otherwise}
            \end{cases}\mper
  \]
\end{definition}

The following is easy to prove.
\begin{fact}
  \label{fact:so-ugly}
For $\omega > 0$ and $c \ll \omega$, there exists $x$, $x = \omega \pm O(c/\omega^3)$, such that ${x \choose 4} = {\omega \choose 4} + c$.
\end{fact}
We now go on to define the corrected moments $\tE:\P_4^n \rightarrow \R$. 

\begin{definition}[Corrected Pseudoexpectation]
\label{def:corrected-moms}
We first use the correction operator $\cL = \cL_\gamma$ to define the corrected moments on all multilinear monomials of degree $4$.
 For every $S \subseteq [n]$, $|S| = 4$, we set:
  \[
    \pE [x_S] = \pE_0[x_S] + \cL [x_S]\mper
  \]
Next, we want to extend $\tE$ to all the monomials so that $\tE[1] = 1$ and $\tE[ \sum_{i} x_i ] = \omega'$ for some $\omega' \approx \omega.$
 Towards this, we let $c = \sum_{S:  \text{$S$ is a $4$-clique in $G$}} \cL[x_S] \mper$ Then, observe that:
  \[
    \pE \sum_{S: \text{$S$ is a $4$-clique in $G$}} x_S
     = \sum_{S: \text{$S$ is a $4$-clique in $G$}} \pE_0 x_S + \cL x_S = {\omega \choose 4} + c
  \]
Then we know there exists $\omega'= \omega \pm O(c/\omega^3) > 0$ satisfying ${\omega' \choose 4} \eqdef {\omega \choose 4} + c$ (using \pref{fact:so-ugly}).
 Thus, the degree $4$ moments we defined ``think'' $\sum_i x_i = \omega'$.
 We use this relationship to extend the definition to all the monomials.
 For every $S \subseteq [n]$, $|S| = 3$, 
  \[
    \pE[x_S] = \frac {1} {\omega' -3} \sum_{\ell \not \in S} \pE [x_{S \cup \ell}] \mper
  \]
Similarly, for each $S$: $|S| = 2$, 
  \[
    \pE [x_S] = \frac {1} {\omega' - 2} \sum_{\ell \not \in S} \pE [x_{S \cup \ell}]\mper
  \]
  Finally, we set
  \[
    \pE[x_i] = \frac 1 {\omega' -1} \sum_{\ell \neq i} \pE[x_i x_\ell]
  \]
  and $\pE 1 = 1$.
\end{definition}
 
\begin{theorem}[ \pref{thm:deg4-intro}, formal]
  \label{thm:deg-4-sqrtn}
  Let $G \sim G(n,1/2)$.
  There is $\omega = \Omega(\sqrt n / \polylog n)$ so that
  with probability $1 - 1/n$ the operator $\pE$ of \pref{def:corrected-moms} is a valid a degree-4 pseudo-expectation satisfying \pref{eq:prog-formulation} for $d = 2$.
\end{theorem}

It is not hard to show that $\pE$ of \pref{def:corrected-moms} satisfies the constraints in \pref{eq:prog-formulation} and the correction above doesn't change $\omega$ by a lot. We defer the proofs to  \pref{sec:deg4-technical}.
\begin{lemma}
  \label{lem:deg4-constraints}
Let $\pE$ be the degree-$4$ corrected moments for clique size $\omega$ (\pref{def:corrected-moms}). Then, there is $\omega'$ such that $\pE$ satisfies
  \[
    \{x_i^2 = x_i\}_{i \in [n]}\mcom \quad \{x_i x_j = 0\}_{i \not \sim j \text{ in } G}\mcom \quad \{ \sum_i x_i = \omega' \}\mper
  \]
Furthermore, if $G \sim G(n,1/2)$, then with probability $1 - O(n^{-25})$, $|\omega' - \omega| < O(\gamma \log(n)^2 \omega^2/n^{5/2})$.
\end{lemma}

Thus, to show \pref{thm:deg-4-sqrtn}, the remaining task is to show that $\tE$ satisfies positive semidefiniteness. Using Fact \pref{fact:homog}, it is enough to show that $\tE[p^2] \geq 0$ for every homogeneous, multilinear polynomial $p$ of degree $2$. This is equivalent to showing that the matrix $\cN' \in \R^{\nchoose{2} \times \nchoose{2}}$ defined by 
$\cN'(I,J) = \tE[ x_{I \cup J}]$ is PSD. Thus, to complete the proof of Theorem \pref{thm:deg-4-sqrtn} we will show the following lemma which is the most technical part of the proof. 

\begin{lemma}\label{lem:deg4-PSD-claim}
  There is $\omega_0 = \Omega(\sqrt{n} / \polylog n)$ and $\gamma = \Theta(1)$ so that for $\omega \leq \omega_0$, with probability at least $1-1/n$ over the draw of $G \sim G(n,\frac{1}{2})$, $\cN' \succ 0$.
\end{lemma}
%
%
%
\subsection{Technical Lemmas and Proofs}
\label{sec:deg4-technical}
We proceed here to show that $\pE$ from \pref{def:corrected-moms} satisfies the appropriate constraits.
We will need the following lemma giving concentration for certain scalar random variables, including the extent to which the correction changes the (pseudo)-expected clique size when $G \sim G(n,1/2)$.
\begin{lemma}
  \label{lem:deg4-scalar}
  Let $G \sim G(n,1/2)$.
  Let the vectors $r_s \in \R^n$ be as in \pref{def:corrected-moms}.
  There is a universal constant $C$ so that with probability $1 - O(n^{-25})$,
  \begin{enumerate}
  \item
   \[
     \left | \sum_s \sum_{i,j,k,\ell \text{ a clique}} r_s(i) r_s(j) r_s(k) r_s(\ell) \right | \leq C n^{5/2} \log(n)^2\mper
   \]
  \item For every $i,j,k$ distinct,
  \[
    \left | \sum_s \sum_{\ell \text{ in a clique with $i,j,k$}} r_s(i)r_s(j)r_s(k)r_s(\ell) \right | \leq  C n \log(n) \mper
  \]
  \item For every $i,j$ distinct and every $s$,
  \[
    \left | \sum_{k \text{ in a clique with $i,j$}} r_s(i)r_s(k) \right | \leq  C \sqrt{n \log n}\mper
  \]
  \end{enumerate}
\end{lemma}

\begin{proof}
  We prove the first item; the others are similar.

  The proof is by several applications of McDiarmid's inequality.
  By a standard Chernoff bound there is a universal constant $C_0$ so that for every $s \in [n]$ and every $i,j,k \in [n]$, with probabilty $1 - O(n^{-40})$,
  \[
    \left | \sum_{\ell \text{ in a clique with $i,j,k$}} r_s(\ell) \right | \leq C_0 \sqrt{n \log n}\mper
  \]
  Call $E_1$ the event that this occurs for every $s,i,j,k$.
  Clearly $\Pr(E_1) \geq 1 - O(n^{-36})$.

  Now for every $s,i,j \in [n]$, we apply McDiarmid's inequality to $| \sum_{k,\ell \text{ in a clique with $i,j$}} r_s(k) r_s(\ell) |$.
  We truncate to get rid of the bad event $\neg E_1$.
  For a graph $G$, let $f(G) = \sum_{k,\ell \text{ in a clique with $i,j$}} r_s(k) r_s(\ell)$ if $E_1$ occurs for $G$ and $f(G) = 0$ otherwise.
  Now consider any pair of graphs $G,G'$ differing on a single edge $(u,v)$.
  It is straightforward to show that if $\{ u,v \} \cap \{s,i,j\} = \emptyset$ then $|f(G) - f(G')| = O(1)$, while otherwise
  \[
    |f(G) - f(G')| \leq \left | \sum_{\ell \text{ in a clique with $i,j,k$}} r_s(\ell) \right | \leq C_1 \sqrt{n \log n}
  \]
  for some other universal constant $C_1$.
  So by McDiarmid's inequality there is $C_2$ so that with probability $1 - O(n^{-34})$,
  \[
    \left | \sum_{k,\ell \text{ in a clique with $i,j$}} r_s(k) r_s(\ell) \right | \leq  C_2 n \log n\mper
  \]
  By a similar argument there is $C_3$ so that for every $s,i \in [n]$, with probability $1 - O(n^{-30})$,
  \[
    \left | \sum_{j,k,\ell \text{ in a clique with $i$}} r_s(j)r_s(k) r_s(\ell) \right | \leq C_3 n^{3/2} (\log n)^{3/2}\mper
  \]
  Let $E_2$ be the event that this bound holds for every $s,i \in [n]$.
  So $\Pr(E_2) \geq 1 - O(n^{-27})$.
  Then, letting $f'(G) = \sum_s \sum_{i,j,k,\ell \text{ a clique}} r_s(i) r_s(j) r_s(k) r_s(\ell)$ if $E_2$ occurs for $G$ and $0$ otherwise,
  we get that on graphs $G,G'$ differing on an edge $(u,v)$
  \[
    |g(G) - g(G')| \leq C_4 n^{3/2} \log(n)^{3/2}
  \]
  for some other constant $C_4$.
  The result follows by a final application of McDiarmid's inequality (we lose a factor of $n$ at this step as opposed to the $\sqrt n$ at previous steps because there are $\approx n^2$ edges to be revealed).
\end{proof}

We can now complete the proof of \pref{lem:deg4-constraints} using \pref{lem:deg4-scalar} and \pref{fact:so-ugly}.
\begin{proof}[Proof of \pref{lem:deg4-constraints}]
  The functional $\pE$ satisfies the constraints $\{x_i^2 = x_i\}_{i \in [n]}, \{x_i x_j = 0\}_{i \not \sim j \text{ in } G}$ by construction.
  Let $\omega'$ be as in \pref{def:corrected-moms}.
  It is routine to check that for $p(x)$ homogeneous of degree $1,2,$ or $3$ that $\pE p(x) \sum_i x_i = \omega' \pE p(x)$ by definition, so it will be enough to check that $\pE \sum_i x_i = \omega'$.

  Recall that $\omega'$ satisfies $\omega'(\omega'-1)(\omega'-2)(\omega'-3) = 4! \cdot \pE \sum_{i,j,k,\ell \text{ a 4 clique}} x_i x_j x_k x_\ell$.
  Now we expand:
  \begin{align*}
    \pE \sum_i x_i & = \pE \frac 1 {\omega' -1} \sum_i \sum_{j \neq i} x_i x_j\\
    & = \pE \frac{1}{(\omega'-1)(\omega'-2)}\sum_{\substack{i,j,k\\\text{all distinct}}} x_i x_j x_k\\
    & = \frac{1}{(\omega'-1)(\omega'-2)(\omega'-3)}\sum_{\substack{i,j,k,\ell\\\text{all distinct}}} x_i x_j x_k x_\ell\\
    & = \frac{1}{(\omega'-1)(\omega'-2)(\omega'-3)} \cdot 4! \cdot \pE \sum_{i,j,k,\ell \text{ a 4 clique}} x_i x_j x_k x_\ell\\
    & = \omega'\mper
  \end{align*}
  It remains just to show our claim on $|\omega - \omega'|$.
  By our choice of $\omega'$ and the guarantees of \pref{fact:so-ugly}, we get that $|\omega - \omega'| \leq |\cL \sum_{i,j,k,\ell \text{ a 4-clique}} x_i x_j x_k x_\ell| / \omega^3$ where $\cL$ is the correction operator from \pref{def:update}.
  By \pref{lem:deg4-scalar}, this is with probability $1 - O(n^{-25})$ at most $O(\gamma \omega^2 \log(n)^2 / n^{5/2})$ when $G \sim G(n,1/2)$.
\end{proof}

In \pref{sec:tools}, we develop some general tools for analyzing the matrices that we encounter before going on to prove \pref{thm:main-degree-d} and \pref{lem:deg4-PSD-claim}.

\section{Tools} \label{sec:tools}
In this section, we build some general purpose tools helpful in the analysis of the matrices of interest to us. We give statements and proofs that are more or less independent of the rest of the paper with an eye towards future work on planted clique and related problems where one deals with random matrices with structure dependencies. The first three sections focus on building an understanding of the symmetries of the eigenspaces $V_0,V_1, \ldots, V_d$ of set symmetric matrices on $\R^{\nchoose{d} \times \nchoose{d}}$. The last section uses moment method with some combinatorial techniques to obtain tight estimates of spectral norm for certain random matrices with dependent entries. 

\subsection{Background on Representations of Finite Groups}
We provide background in the required tools from basic representation theory below.

\begin{definition}[Representation]
For a finite dimensional complex vector space $V$, let $Hom(V,V)$ be the set of all linear maps from $V$ into $V$. For any finite group $G$ and $\pi:G \rightarrow Hom(V,V)$, the pair $(\pi, V)$ is said to be a representation of $G$ if $\pi$ satisfies, for any $g_1, g_2 \in G$, $$\pi(g_1 \cdot g_2) = \pi(g_1) \cdot \pi(g_2),$$ where the ``$\cdot$'' on the LHS corresponds to the group operation and on the RHS, the composition of linear maps on $V$. When the map $\pi$ is clear from the context (as some natural action of the group $G$ on $V$), we abuse notation and just say that $V$ is a representation of $G$.
\end{definition}

Let $(\pi, V)$ be a representation of a group $G$. A subspace $W \subseteq V$ is said to be a \emph{subrepresentation} if for every $w \in W$, $\pi(g) w \in W$ for every $g \in G$. That is, $W$ is a stable or invariant subspace for all the linear maps $\pi(g)$, one for each $g \in G$. Observe that in this case, $(\pi,W)$ is another representation of $G$. A representation $(\pi,V)$ of $G$ is said to be \emph{irreducible} if for any subspace $W$ invariant under all the linear maps $\pi(g)$ for $g \in G$, $W = V$ or $W = \{0\}$. 

Every representation $V$ of $G$ can be decomposed as a direct sum of subspaces each of which is an irreducible representation of $G$. Further, for any finite group $G$, there are at most $|G|$ distinct irreducible representation up to isomorphism. For well studied finite groups such as the symmetric group on $n$ elements $S_n$, the set of irreducible representations are well known and well studied. The power of representation theory in the present context comes from understanding the eigenspace structure of linear operators that are invariant under some action of the group $G$ (in our case $\S_n$). 

There is a natural linear action of the permutation group $\S_n$ on $\P_q$ for any $q$, denoted by $\pi:\S_n \rightarrow Hom(\P_q, \P_q)$:  A permutation $\sigma \in \S_n$ when applied to a vector 
$v \in \P_q$ produces the vector $v' \in \P_q$ such that $v'_I = v_{\sigma(I)}$ for every $I \in \nchoose{q}$. This can be alternately described as multiplication by the permutation matrix associated with $\sigma$. Observe that $(\sigma_1 \cdot \sigma_2) \cdot v = \sigma_1 \cdot (\sigma_2 \cdot v)$ and thus, $(\pi, \P_q)$ is a representation of $\S_n$. It is known (See Section 3.2, \cite{Bannai-Ito}) that under this action, $\P_d$ can be decomposed as direct sum of subspaces $V_0, V_1, \ldots, V_q$ such that each $V_i$ is an irreducible representation of $\S_n$ and none of $V_i, V_j$ for $i \neq j$ are isomorphic to each other. 

The expected moment matrix $E = \E[ \M']$ is set symmetric and therefore commutes with the action of $\S_n$ on $\P_d$ described above. This can be easily used to obtain that $E$ has the eigenspaces $V_0, V_1, \ldots, V_d$ discussed above.
We need the following consequence of a basic representation-theoretic result.

\begin{fact}[Consequence of Schur's Lemma \cite{Serre}] \label{lem:schur}
Suppose $(\pi, V)$ and $(\pi', W)$ are representations of a group $G$. Suppose $L
:V \rightarrow W$ is a linear map such that for any $g \in G$ and $ v \in V$, $$L ( \pi(g)\cdot v)  = \pi'(g) \cdot L (v).$$ Then, for any irreducible representation $V_i \subseteq V$ under $\pi$, $L(V_i) \subseteq W$ is an irreducible representation in $W$ under $\pi'$.
\end{fact}

\subsection{Eigenspaces of the Set Symmetric Matrices}
We often encounter random ${n \choose d} \times {n \choose d}$ matrices $M$ indexed by subsets of $[n]$ of size $d$.
For example, a common feature in our setting (as observed in \cite{MPW15}) is that $E= \E [M ]$ depends only on $|I \cap J|$.

\begin{definition}[Set Symmetry]
A matrix $A \in \R^{\nchoose{d} \times \nchoose{d}}$ is said to be \emph{set symmetric} if for every $S,T, S', T' \in \nchoose{d}$ such that $|S \cap T| = |S' \cap T'|$, $A(S,T) = A(S',T')$.
\end{definition}

The set of all set symmetric matrices is known as the \emph{Johnson} scheme in algebraic combinatorics.
 All such matrices commute and thus share eigenspaces. While the matrices in the Johnson scheme are well studied, the description of the eigenspaces in the literature is hard to use for the purpose of our proofs.
  We thus take a more direct approach and use basic representation theory in what follows to identify a simple symmetry condition on the eigenspaces of set symmetric matrices which will be useful to understand the spectral properties of the matrices we study.

\begin{lemma} \label{lem:eigenbasisstructure}
Let $V_0,V_1, \ldots, V_{d} \subseteq \R^{\nchoose{d}}$ be the eigenspaces of set symmetric matrices on $\R^{\nchoose{d} \times \nchoose{d}}$ described in the previous section. For any $u \in \R^{ \nchoose{t}}$, let $v \in \R^{\nchoose{d}}$ for $d \geq t$ be defined so that for each $I \in \nchoose{d}$, $$v_I = \sum_{I' \subseteq I, |I'| = t} u_{I'}.$$ Then, $v \in V_{0} \oplus V_1 \oplus \cdots \oplus V_{t}$. 
\end{lemma}

\begin{proof}
Let $\P_q$ for any positive integer $q$ be the space of all vectors indexed by elements of $\nchoose{q}$. Consider the standard action of $S_n$ on $[n]$ that sends $i \rightarrow \sigma(i)$ for any $\sigma\in \S_n$. This induces a natural action on $\nchoose{q}$ where any $I \in \nchoose{q}$ is sent to $\sigma(I) = \{ j \mid \exists i \in I \text{ such that } \sigma(i) = j \}$. This further induces a natural action on $\P_q$ by taking $v = \{v_I\}_{I \in \nchoose{q}}$ and sending it to $v'$ where $v'_{I} = v_{\sigma^{-1}(I)}$ for every $I$. A quick check ensures that the action defined above satisfies $\sigma \circ \tau (v) = \sigma ( \tau (v))$ for any $\sigma, \tau \in \S_n$. Thus, $\P_q$ is a representation of $\S_n$ under the action defined above for any $q$. It is easy to check that left multiplication by any set symmetric matrix from $\R^{\nchoose{q} \times \nchoose{q}}$ commutes with the action of $\S_n$ defined above. From \pref{fact:johnsondecomp}, the eigenspaces of any set symmetric matrix acting on $\P_q$ are given by $V_0, V_1, \ldots, V_q$ such that $\dim(V_i) = {n \choose i} - {n \choose {i-1}}$. By \pref{fact:johnsondecomp} each $V_i$ is isomorphic to distinct irreducible representations of $\S_n$. 

Next, consider the map $C:\P_t \rightarrow \P_d$ such that for any $u \in \R^{\nchoose {t}}$, the value $C(u) \in \R^{ \nchoose{d}}$ is given by $v$ such that $v_I = \sum_{I' \subseteq I \text{, } |I'| = t} u_I$. Then, $C$ is linear and we claim that $C$ commutes with the action of $\S_n$ defined above: $\sigma (C(u)) = C ( \sigma(u))$. Note that on the LHS, $\sigma$ refers to the action of $\S_n$ on $\P_d$ while on the RHS, it refers to the action on $\P_{t}$. We follow the definition to verify this: $$\left(\sigma (C(u)) \right)_{I} = \sigma ( \sum_{I' \subseteq I, |I'| = t} u_{I'}) = \sum_{I' \subseteq I, |I'| = t} u_{\sigma^{-1}(I')} = \sum_{I' \subseteq \sigma(I)} u_{I'} = C( \sigma(u))_{I}.$$

Suppose $u \in V_i^t$ for $i \leq t$ where $V^t_i$ is some eigenspace of a set-symmetric ${n \choose t} \times {n \choose t}$ matrix.
 Then, by \pref{lem:schur} $C(V^t_i)$ is an irreducible representation of $\S_n$ and is thus an invariant subspace for the action of $\S_n$ in $\P_d$.
  By a dimension argument, $C(V^t_i)= V_i$.
 Thus, $C(u) \in V_i$. 

\end{proof}

\subsection{Kernels of Patterned Matrices}\label{sec:patterned}
In this section we design some general tools to understand the spectral structure of matrices that have restricted variations around the set symmetric structure discussed in the previous section. The main tool we will use to establish these results is \pref{lem:eigenbasisstructure} shown in the previous section. Before moving on to this task, we describe a high level overview of what we intend to do. The following paragraph can be skipped to dive directly into the technical details without the loss of continuity.

The study of the eigenspaces of set symmetric matrices lets us completely understand the spectral structure of the expectation matrix $E$. In the next section when we analyze the spectrum of $\M'$, we will encounter matrices that depend on the underlying graph $G$ and thus are not set symmetric. However, if the dependence on the underlying graph $G$ is in some sense limited, we hope that some of the nice algebraic properties that set symmetry grants us should perhaps continue to hold. In our case, we will be able to decompose $E$ into various pieces and for each of these pieces, the entry at $(I,J)$ has dependence on the graph $G$ based only on the status (edge vs no edge) of a small number of pairs $(i,j) \in I \times J$. The goal of this section is to develop tools to understand certain (coarse) spectral properties of such matrices. 

Our aim is to study matrices in $Q = Q(G) \in \R^{\nchoose{d} \times \nchoose{d}}$ for a graph $G$ on $[d]$ such that $Q(I,J)$ depends on a) the intersections between $I$ and $J$ b) the values of $g_b$ (the edge indicator of $G$) for pairs $b$ of vertices (from the non intersecting parts of $I$ and $J$). We first develop some notation to talk about such matrices.

Next, we define \emph{patterns}:
\begin{definition}[Pattern]
For $Z_\ell, Z_r \subseteq [d]$, let $\B_{Z_\ell,Z_r}$ be the set of all non-empty bipartite graphs on left and right vertex sets each given by $[d] \setminus Z_\ell $ on the left and  $[d] \setminus Z_r$ on the right. Define $\B^q = \cup_{|Z_\ell|,|Z_r| = q} \B_{Z_\ell,Z_r}$. Then, a tuple $(B,Z_\ell,Z_r)$ for $B \in \B_{Z_\ell,Z_r}$ is said to be a $q$ \emph{pattern} where $q = |Z_\ell| = |Z_r|$. When $q = 0$, we call $B$ itself a pattern.
\end{definition}

For any set $I$, $J$, consider the ``sorting maps'' $\zeta_I:[d] \rightarrow I$, $\zeta_J: [d] \rightarrow J$ i.e, $\zeta_I(1)$ is the least element of $I$, $\zeta_I(2)$, the next to the least and so on. We can extend $\zeta_I$, $\zeta_J$ to subsets of $[d]$ in the natural way. Let $B_{\ell}$ ($B_r$) be the subset of vertices on the left (right) hand side with non zero degrees in $B \in \B_{Z_\ell,Z_r}$. For any $I, J \in \nchoose{d}$, there is a natural map that takes $B$ and obtains a copy of $B$ on vertex sets $I$ and $J$, via the sorting maps $\zeta_I$ and $\zeta_J$ from above: $\zeta_{I,J}(B)$ is the bipartite graph on $I$, $J$ with the edges obtained by taking every edge $b = \{i,j\} \in B$ and adding the edge $\{\zeta_I(i), \zeta_J(j)\}$ to $\zeta_{I,J}(B)$.

We need to understand the effect of applying a permutation $\sigma \in \S_d$ to $(B,Z_\ell, Z_r)$ for $B \sim \B_{Z_\ell,Z_r}$. Let $\sigma \in \S_d$ be a permutation on $[d]$. Given $(B,Z_\ell,Z_r)$, $\sigma$ has two natural actions. The left action of $\sigma$ on $(B,Z_\ell, Z_r)$ produces $\sigma \circ (B,Z_\ell, Z_r) \eqdef (\sigma \circ B, \sigma(Z_\ell), Z_r)$ where each edge $(i,j) \in B$ is sent into $(\sigma(i), j)$ in $\sigma \circ B$. We similarly define the right action of $\sigma$ on $(B,Z_\ell, Z_r)$ that produces $(B,Z_\ell, Z_r) \circ \sigma \eqdef (B \circ \sigma, Z_\ell, \sigma(Z_r))$. Each of these two actions defines a subgroup that leaves $B$ invariant. 

\begin{definition}[Automorphism Groups]
Let $B \in \B_{Z,Z'}$ be a labeled bipartite graph. We define the left automorphism group of $(B,Z_\ell, Z_r)$ as $$Aut_\ell(B,Z_\ell, Z_r) = \{ \sigma \in \S_d \mid  \sigma \circ (B,Z_\ell, Z_r) = (B,Z_\ell, Z_r)\},$$ and the right automorphism group of $B$ as $$Aut_r(B,Z_\ell, Z_r) = \{ \sigma \in \S_d \mid (B,Z_\ell, Z_r) \circ \sigma = (B,Z_\ell, Z_r) \}.$$
\end{definition}

Next, we define equivalence classes of the patterns $(B,Z_\ell, Z_r)$.
\begin{definition}[Similar Patterns]
For patterns $(B,Z_\ell, Z_r)$ is \emph{left} similar to $(B',Z'_\ell, Z'_r)$ and write $(B,Z_\ell, Z_r)\sim_\ell (B',Z'_\ell, Z'_r)$ if there exists a $\sigma \in \S_d$ such that $\sigma \circ (B,Z_\ell, Z_r)= (B',Z'_\ell, Z'_r)$. Similarly, we say that $(B,Z_\ell, Z_r)$ is \emph{right} similar to $ (B',Z'_\ell, Z'_r)$ and write $(B,Z_\ell, Z_r)\sim_r (B',Z'_\ell, Z'_r)$ if there exists a $\sigma \in \S_d$ such that $(B',Z'_\ell, Z'_r) = (B,Z_\ell, Z_r)\circ \sigma.$ 
\end{definition}

We are now ready to define patterned matrices:
\begin{definition}[Patterned Matrices] \label{def:patterned}
Let $(B,Z_\ell,Z_r)$ be a $q$-pattern. Let $f: \{-1,1\}^{B} \rightarrow \R$ be a function that maps a $\{-1,1\}$ labeling of the pairs in $B$ to $\R$. For a graph $G$ on $[n]$ vertices, the patterned matrix with pattern $(B,Z_\ell,Z_r)$ defined by $f$ is a matrix in $Q = Q_{B,Z_\ell,Z_r,f}(G) \in \R^{\nchoose{d} \times\nchoose{d}}$ such that 
\[
Q(I,J) = \begin{cases}
f( \{g_b\}_{b \in \zeta_{I,J}(B)} ), \text{  for every $I, J$ } \zeta_I(Z_\ell) = \zeta_J(Z_r)\\
0 \text{,  otherwise.} 
\end{cases}
\]
When $q=0$, we write $Q_{B,f}$ for the corresponding patterned matrix.
\end{definition}
The following result describes the kernels of certain symmetrized sums of $Q_{B,Z,f}$ and is the main claim of this section. 
\begin{lemma} \label{lem:kernelsymoperators}
For graph $G$, a $q$-pattern $(B,Z_\ell,Z_r)$ and $f:\{-1,1\}^B \rightarrow \R$, let $Q = Q_{B,Z_\ell,Z_r,f}(G) \in \R^{\nchoose{d} \times \nchoose{d}}$ be the corresponding patterened matrix. Define the left and right symmetrized version of $Q$ by:
 $$Q^{\ell} = \sum_{(B',Z_\ell',Z_r') \sim_{\ell} (B,Z_\ell,Z_r)} Q_{B',Z_\ell',Z_r',f},$$ and $$Q^r = \sum_{(B',Z_\ell',Z_r') \sim_{r} (B,Z_\ell,Z_r)} Q_{B',Z_\ell,Z_r,f},$$ respectively. Let $B_{\ell}$ ($B_r$) be the subset of vertices on the left (right, respectively) hand side with non zero degrees in $B$.
 Then,
  \begin{enumerate}
    \item For every $j > |B_\ell| +q$,  $$\Pi_j^{\dagger} Q^{\ell}_{B,Z,f}  = 0,$$
    \item For every $i > |B_r| +q$, $$ Q^{r}_{B,Z,f} \Pi_i = 0.$$
  \end{enumerate}
\end{lemma}

\begin{proof}
Observe that $Q^{\ell}_{B,Z_\ell,Z_r,f} = Q^{\ell}_{B',Z_\ell',Z_r',f}$ for any $(B,Z_\ell,Z_r) \sim_{\ell} (B',Z_\ell',Z_r')$. This motivates us to first obtain a more symmetric looking expression for $Q^{\ell}$ and $Q^r$. Let $\S_d/ Aut_{\ell}((B,Z_\ell,Z_r))$ (and correspondingly, $\S_d/Aut_r((B,Z_\ell,Z_r))$) be the group of left (right) cosets of $ Aut_{\ell}((B,Z_\ell,Z_r))$ ($Aut_r((B,Z_\ell,Z_r))$ respectively). We have: 
\begin{align*}
Q^{\ell} = \sum_{(B,Z_\ell,Z_r) \sim_{\ell} (B',Z_\ell',Z_r')} Q_{B',Z_\ell',Z_r',f}
= \sum_{\tau \in \S_d/Aut_{\ell}((B,Z_\ell,Z_r))} Q_{\tau \circ (B,Z_\ell,Z_r) }
= \frac{1}{|Aut_{\ell}((B,Z_\ell,Z_r))|} \sum_{\sigma \in \S_d} Q_{\sigma \circ (B,Z_\ell,Z_r)}.
\end{align*}

Similarly, we have: $$Q^r_{B,Z_\ell,Z_r,f} = \frac{1}{|Aut_r((B,Z_\ell,Z_r))|} \sum_{\sigma \in \S_d} Q_{(B,Z_\ell,Z_r) \circ \sigma, f}.$$ We now begin the argument for proving the first claim. The second claim has an analogous proof. Consider an arbitrary $v = \{v_{I}\}_{I \in \nchoose{d}} \in \R^{\nchoose{d}}$. We will show that $Q^{\ell} v \in V_0 \oplus V_1 \oplus \cdots \oplus V_{q+|B_{\ell}|}$. Towards this goal, define a vector $u = \{u_T\}_{T \in \nchoose{k}} \in \R^{\nchoose{k}}$ as follows: for each $T \in \nchoose{k}$, let $I_T \in \nchoose{d}$ be arbitrary subject to the constraint that $I_T \supseteq T$. We define:
$$ u_T =  \sum_{\begin{subarray}[(B',Z'_\ell,Z'_r) \\\sim_{\ell} (B,Z_\ell,Z_r)\\ B'_\ell \cup Z_\ell' = T \end{subarray}} \sum_{J:\zeta_J(Z_r) = \zeta_{I_T}(Z'_{\ell})} v_J f( \{g_b\}_{b \in \zeta_{I_T,J}(B')}).$$ 

We first show that $u$ above is well defined in that the definition does not depend on the specific subset $I_T$ used so long as $I_T \supseteq T$. We adopt the notation (that ignores the ``direction'' of action) $B^{\sigma} \eqdef \sigma \circ B$ only for the calculations that follow. 

\begin{claim}
Fix $T \in \nchoose{q+|B_\ell|}$ and let $I_1, I_2 \in \nchoose{d}$ such that $T \subseteq I_1, I_2$. Then, $$ \sum_{\begin{subarray}[(B',Z'_\ell,Z'_r) \\\sim_{\ell} (B,Z_\ell,Z_r)\\ B'_\ell \cup Z_\ell' = T \end{subarray}} \sum_{J:\zeta_J(Z_r) = \zeta_{I_1}(Z'_{\ell})} v_J f( \{g_b\}_{b \in \zeta_{I_1,J}(B')}) = \sum_{\begin{subarray}[(B',Z'_\ell,Z'_r) \\\sim_{\ell} (B,Z_\ell,Z_r)\\ B'_\ell \cup Z_\ell' = T \end{subarray}} \sum_{J:\zeta_J(Z_r) = \zeta_{I_2}(Z'_{\ell})} v_J f( \{g_b\}_{b \in \zeta_{I_2,J}(B')}) .$$ \label{claim:well-defined}
\end{claim}
\begin{proof}[Proof of Claim] 
We can equivalently write the claim above as: 
\begin{equation}
\sum_{\sigma \in \S_d} \sum_{J: \zeta_J(Z_r) = \zeta_{I_1}(\sigma(Z_\ell))} v_J \cdot f( \{g_b\}_{b \in \zeta_{I_1,J}(B^{\sigma})}) = \sum_{\sigma \in \S_d} \sum_{J: \zeta_J(Z_r) = \zeta_{I_2}(\sigma(Z_\ell))} v_J \cdot f( \{g_b\}_{b \in \zeta_{I_2,J}(B^{\sigma})})
\end{equation}

We start with the LHS and observe that for any $\tau \in \S_d$ it equals: $$\sum_{\sigma \in \S_d} \sum_{J: \zeta_J(Z_r) = \zeta_{I_1}(\sigma \circ \tau (Z_\ell))} v_J \cdot f( \{g_b\}_{b \in \zeta_{I_1,J}(B^{\sigma \circ \tau})}).$$ We show that there exists a $\tau$ such that $\zeta_{I_1}(\sigma \circ \tau(Z_\ell)) = \zeta_{I_2}(\sigma(Z_\ell))$ and $\zeta_{I_1,J}( \sigma \circ \tau  \circ B ) = \zeta_{I_2,J}( \sigma \circ B).$

For each $i \in I'$, let $b^2_i \in [d]$ be such that $\zeta_{I_2} (b^2_i) = i$. Similarly, for each $i \in I'$, let $b^1_i \in [d]$ be such that $\zeta_{I_1}(b^1_i) = i$. For each $i \in I'$, choose $\tau$ such that $\tau(b^2_i) = b^1_i$. Then, $\zeta_{I_1}( \tau (b^2_i)) = \zeta_{I_1}(b^2_i) = i$. Thus, $\zeta_{I_1,J}( B'^{\tau}) = \zeta_{I_2,J}(B')$ for every $(B',Z_\ell',Z_r') \sim_\ell (B,Z_\ell,Z_r).$ 
\end{proof}

We can now show that $(Q^{\ell} v)_{I} = \sum_{I' \subseteq I, |I'| = q+|B_\ell|} u_{I'}$. We now observe:
\begin{align}
(Q^{\ell} v)_I  &= \sum_{J \in \nchoose{d}} Q^{\ell} (I,J) \cdot v_J \\
&= \frac{1}{|Aut_{\ell}(B,Z_\ell,Z_r)|} \sum_{J \in \nchoose{d}} \sum_{\sigma \in \S_d}  Q_{B',Z_\ell',Z_r',f} v_J \notag \\
&\text{ Keeping $J$s that correspond to non-zero entries in $Q_{B',Z_\ell',Z_r\,f}$}\\
&= \frac{1}{|Aut_{\ell}(B,Z_\ell,Z_r)|} \sum_{\sigma \in \S_d} \sum_{J:\zeta_J(Z_r) = \zeta_I(Z'_{\ell})} Q_{B',Z_\ell',Z_r',f} v_J \notag \\
&= \frac{1}{|Aut_{\ell}(B,Z_\ell,Z_r)|} \sum_{\sigma \in \S_d} \sum_{J:\zeta_J(Z_r) = \zeta_I(Z'_{\ell})} f(\{g_b\}_{b \in \zeta_{I,J}(B^{\sigma})}) v_J \notag \\ 
&\text{ Using that} \cup_{I' \subseteq I} \{ \sigma \mid \zeta_{I}(B^{\sigma}_{\ell}) \cup \sigma(Z_{\ell}) = I' \} \text{ forms a partition of $\S_d$ indexed by $I'$,}\\
&= \sum_{I' \subseteq I} \frac{1}{|Aut_{\ell}(B,Z_\ell,Z_r)|} \sum_{\sigma \in \S_d \text{, } B^{\sigma}_{\ell} \cup \sigma(Z_\ell) = I' \text{, }  J:\zeta_J(Z_r) = \zeta_I(Z'_{\ell})}  f(\{g_b\}_{b \in \zeta_{I,J}(B^{\sigma})}) v_J  \notag\\
&\text{ Since each $(B',Z_\ell',Z_r') \sim_\ell (B,Z_\ell,Z_r)$ s.t. $B'_{\ell} \cup Z_\ell' = I'$ occur $|Aut_{\ell}(B,Z_\ell,Z_r)|$ times in inner sum,} \notag\\
&= \sum_{I' \subseteq I} \sum_{(B',Z_\ell',Z_r') \sim_\ell (B,Z_\ell,Z_r) \text{, } B'_{\ell} \cup Z_\ell' = I' }  \sum_{J:\zeta_J(Z_r) = \zeta_I(Z'_{\ell})}f(\{g_b\}_{b \in \zeta_{I}( B'_{\ell})}) v_J \notag\\
&\text{ Using \pref{claim:well-defined}: }\\
&= \sum_{I' \subseteq I} u_{I'}.
\end{align}
This completes the proof using \pref{lem:eigenbasisstructure}.
\end{proof}

\subsection{Concentration for Locally Random Matrices over $G(n,\frac{1}{2})$}
\label{sec:spectralnorm}
The goal of this section is to prove strong concentration bounds for the matrices will encounter in our analysis. The first result is a spectral concentration bound for the patterned matrices $Q = Q_{B,f}(G)$ for $G \sim G(n, \frac{1}{2})$ when $f:\{-1,0,1\}^B \rightarrow \R$ is given by $f(x) = \Pi_{b \in B} x_b$. In other words, the entry $Q(I,J)$ is the product of the edge indicator variables $g_b$ for $b \in \zeta_{I,J}(B)$. These bounds will be used in \pref{sec:degree-d}.

\begin{lemma}
  \label{lem:fourier-norm-mpw-disjoint}
  For $d \geq 2$, $d = O(\log{(n)})$, and a bipartite graph $B \in \B$, let $Q = Q_{B,f}$ be a patterned matrix with $f(x) = \prod_{b \in B} x_b$. That is, 
  \[
    Q(I,J) = \begin{cases}
      \prod_{b \in \zeta_{I,J}(B)} g_{b} & \text{ if $I \cap J = \emptyset$}\\
      0 & \text{ otherwise}
     \end{cases}\mcom
  \]
  Then:
  \begin{enumerate}
    \item When $B$ contains a $2$-matching, then $\Pr (\|Q \| \geq n^{d-1}(\log n)^3) \leq O(n^{-10})$.
    \item When $B$ is not the empty graph, $\Pr(\|Q\| \geq n^{d-1/2}(\log n)^3) \leq O(n^{-10})$.
  \end{enumerate}
\end{lemma}

The next main result of this section considers a different class of matrices that appear in the analysis in \pref{sec:deg4-psd}. 
\begin{lemma}
  \label{lem:correction-disjoint}
  Let $U \subseteq [2] \times [2]$ be a bipartite graph on $4$ vertices and suppose $U$ is nonempty.
  Let $M  \in \nchoose{2} \times \nchoose{2}$ be a matrix with the entry at $\{a_1,a_2\}$, $\{b_1,b_2\}$ for $a_1 \leq a_2$ and $b_1 \leq b_2$ are given by
  \[
    M[\{a_1,a_2 \},\{b_1,b_2 \}] = \begin{cases}
      \sum_{k \in [n]} g_{k,a_1}g_{k,a_2}g_{k,b_1}g_{k,b_2}\prod_{(i,j) \in U} g_{a_i, b_j} & \text{ if $|\{ a_1,a_2,b_1,b_2 \}| = 4$}\\
      0 & \text{ otherwise}
     \end{cases}\mcom
  \]
  Recall that we set $g_{aa} = 0$ for every $a\in [n]$ by convention. Then, whenever $U$ is non-empty, $\Pr(\|M\| \geq n^{3/2} (\log n)^3) \leq O(n^{-10})$. If $U$ is the empty graph, then $\Pr(\|M \| \geq n^2 (\log n)^3) \leq O(n^{-10})$.
\end{lemma}

The proofs of both these results are based on the standard idea of analyzing the trace of higher powers of a matrix to prove bounds on its spectral norm. The proof of \pref{lem:fourier-norm-mpw-disjoint} is similar to the proofs via the trace power method for bounding the norms of matrices as presented in \cite{DM15}. The general format we present here will come in handy for multiple applications to various matrices in \pref{sec:degree-d}. \pref{lem:correction-disjoint} deals with somewhat more complicated matrices that appear in the analysis of the corrected operator for degree $4$ lower bound. Nevertheless, as is common in such proofs, the analysis is based on a combinatorial analysis of the terms that make non zero contribution to the trace powers combined with the simplifying effect of random partitioning based arguments. We describe the details of the proof in the following section.

\subsubsection{General Tools}
Before diving into the details, we present three general purpose tools that we will employ repeatedly in our analyses. For analyzing the spectral norm of a matrix $Q \in \R^{\nchoose{d} \times \nchoose{d}}$, the first tool allows us to analyze instead a related matrix $Q' \in \R^{n^d \times n^d}$. That is, instead of rows and columns being indexed by subsets of vertices as in $Q$, $Q'$ has rows and columns indexed by ordered tuples of vertices of size $d$. This transformation is not hard as one can find $Q$ as a principal submatrix of $Q'$.

\begin{lemma}[Sets to Ordered Tuples]
  \label{lem:embed}
For any $Q  \in \R^{\nchoose{d} \times \nchoose{d}}$ define the matrix $Q' \in \R^{n^d \times n^d}$ such that for any ordered tuple $S = (a_1, a_2, \ldots, a_d), T = (b_1, b_2, \ldots, b_d) \in [n]^d$, $Q'(S,T) = Q(\{a_1, a_2, \ldots, a_d\}, \{b_1, b_2, \ldots, b_d\})$. Then, $\|Q\| \leq \|Q'\|$.  
\end{lemma}
\begin{proof}
  It is enough to show that $Q'$ occurs as a principal submatrix of $Q$. For this, take the submatrix of rows and columns of $M$ indexed by tuples $(a_1,\ldots,a_d)$ in sorted order, i.e., with $a_1 \leq a_2 \leq \ldots a_d$.
\end{proof}

We will use the following lemma to break dependencies in certain random matrices by decomposing them into matrices whose entries, while still dependent, have additional structure.

\begin{lemma}[Random Partitioning]
  \label{lem:partitioning}
 For $d \in \N$, let $Q \in \R^{n^d \times n^d}$.
  \begin{enumerate}
  \item
  \label{itm:partition-basic}
  Suppose $Q(I,J) = 0$ when $I \cap J \neq \emptyset$.
 Let $(S_1^1,\ldots,S_k^1),\ldots,(S_1^r,\ldots,S_k^r)$ be a sequence of partition of $[n]$ into $k$ bins. Each partition induces a matrix based on $Q$ as follows:
  \begin{align*}
    & Q_{i}[(a_1,\ldots,a_d),(b_1,\ldots,b_d)]  \\
    & \quad = \begin{cases}
      Q[(a_1,\ldots,a_d),(b_1,\ldots,b_d)] & \text{ if $a_j, b_j \in S_j^i$ for $j <k$}\\
      & \text{ and $a_j, b_j \in S_k$ for $j \geq k$}\\
      & \text{ and for all $i' < i$, $M_{i'}[(a_1,\ldots,a_d),(b_1,\ldots,b_d)] = 0$}\\
      0 & \text{ otherwise}
      \end{cases}\mper
  \end{align*}
  Then, there is a family of partitions $(S_1^1,\ldots,S_k^1),\ldots,(S_1^r,\ldots,S_k^r)$ such that $Q = \sum_{i = 1}^r Q_{i}$ with $r \leq O(k^k \log n)$.
  \item 
  \label{itm:partition-fancy}
 Let $Q^j \in \R^{n^d \times n^d}$, for each $1 \leq j \leq n$, be matrices such that $Q^j(I,J) = 0$ whenever $I \cap J \neq \emptyset$ or $j \in I \cup J$. Suppose $Q = \sum_{j = 1}^n Q^j$. For a partition $(S_1,\ldots,S_k,T)$ of $[n]$ into $k+1$ parts, say that $j,(a_1,\ldots,a_d),(b_1,\ldots,b_d)$ respect the partition if $j \in T$, $a_i,b_i \in S_i$ for all $i$.
  Let a sequence of partitions $(S_1^1,\ldots,S_k^1,T^1),\ldots,(S_1^r,\ldots,S_k^r,T^r)$ of $[n]$ into $k+1$ parts induce matrices $M_{i}$ in the following way.
  \begin{align*}
   Q_{i}[(a_1,\ldots,a_d),(b_1,\ldots,b_d)] = \sum_{j \in T^i_{(a_1,\ldots,a_d),(b_1,\ldots,b_d)}} Q^j[(a_1,\ldots,a_d),(b_1,\ldots,b_d)] 
  \end{align*}
  where $T^i_{(a_1,\ldots,a_d),(b_1,\ldots,b_d)}$ is the set of indices $j$ so that $j,(a_1,\ldots,a_d),(b_1,\ldots,b_d)$ respect the partition $(S_1^i,\ldots,S_k^i,T^i)$ and do not respect any partition $(S_1^{i'},\ldots,S_k^{i'},T^{i'})$ for any $i' < i$.
  \end{enumerate}
  Then, there is a family $(S_1^1,\ldots,S_k^1,T^1),\ldots,(S_1^r,\ldots,S_k^r,T^r)$ of partitions of $[n]$ so that $Q = \sum_{i = 1}^r Q_{i}$ with $r \leq O((k+1)^{k+1} \log n)$.
\end{lemma}
\begin{proof}[Proof of \pref{lem:partitioning}]
  We present the proof of \pref{itm:partition-basic}; the proof of \pref{itm:partition-fancy} is almost identical.
  For $r$ to be chosen later, we pick partitions $(S_1^1,\ldots,S_k^1),\ldots,(S_1^r,\ldots,S_k^r)$ 
 uniformly at random and independently so that each is partition of $[n]$ into sets of size $n/k$ each.

  Call $(a_1,\ldots,a_d),(b_1,\ldots,b_d)$ \emph{good} at step $i$ if $a_j, b_j \in S_j^i$ for every $j < k$ and $a_j, b_j \in S_k^i$ if $j \geq k$.
  It is enough to show that after $r \leq O(k^k \log n)$ steps the probability that every $\{a_1,\ldots,a_d,b_1,\ldots,b_d \}$ of size $2d$ is good at some step $i \leq r$.

  Fix some $(a_1,\ldots,a_d),(b_1,\ldots,b_d)$ with $| \{ a_1,\ldots,a_d,b_1,\ldots,b_d \}| = 2d$.
  It is good at step $i$ with probability at least $k^{-k}$.
  Since the steps are independent, after $r$ steps
  \begin{align*}
    \Pr((a_1,\ldots,a_d),(b_1,\ldots,b_d) \text{ is good}) & \geq (1 - \tfrac 1 {k^k})^r\\
    & = ((1 - \tfrac 1 {k^k})^{k^k})^{r/k^k}\\
    & \leq (\tfrac 1 e)^{r/k^k}
  \end{align*}
  which is at most $1/n^{10d}$ for some $r = O(k^k \log n)$.

  Taking a union bound over all $O(n^{2d})$ tuples $(a_1,\ldots, a_d),(b_1,\ldots,b_d)$ with $|\{a_1,\ldots,a_d, b_1,\ldots,b_d\}| = 2d$ completes the proof.
\end{proof}

Finally, the following lemma relates the norms of certain matrices in $X \in \R^{\nchoose{d} \times \nchoose{d}}$ that have non zero entry $(I,J)$ only if $|I \cap J | = q$ to a certain \emph{lift} of $X$ that lives in $\R^{\nchoose{d-q} \times \nchoose{d-q}}$ and has non zero entries $I,J$ only when $I \cap J = \emptyset$.  The latter case is easier to handle and the idea of lifts helps reducing the norm computation for lifts of $X$ to that of $X$.

\begin{definition}[Lifts of Matrices, Equation 8.5 in \cite{MPW15}] \label{def:liftsofnondisjoint}
For a matrix $X \in \R^{\nchoose{d-i} \times \nchoose{d-i}}$ for some $0 \leq i \leq d$ such that $X(I',J') = 0$ whenever $I' \cap J' \neq \emptyset$, define the lift $X^{(i)} \in \R^{\nchoose{d} \times \nchoose{d}}$ to be the matrix defined by:
\[
X^{(i)}(I,J) = \begin{cases}
X(I \setminus (I \cap J), J \setminus (I \cap J)), \text{ if} |I \cap J| = i\\
0, \text{ otherwise.}
\end{cases}
\]
\end{definition}

The usefulness of the above definition is captured by the following claim:
\begin{fact}[Lemma 8.4 in \cite{MPW15}]\label{fact:liftsofnondisjoint}
Let $X \in \R^{\nchoose{d-i} \times \nchoose{d-i}}$ for some $0 \leq i \leq d$ such that $X(I',J') = 0$ whenever $I' \cap J' \neq \emptyset$. Then, for the lift $X^{(i)}$ of $X$, we have: $$\| X^{(i)}\| \leq {d \choose i}^2 \cdot \|X\|.$$
\end{fact}



\subsubsection{Graph-Theoretic Definitions and Lemmas}
In this section, we set up some notation and definitions helpful in our proofs of the main results of this section. The next few definitions and notation are generalization of the ones used in \cite{DM15} to general degrees $d$ and are useful in the proof of \pref{lem:fourier-norm-mpw-disjoint}.
%
%

\begin{definition}
  Let $U$ be a bipartite graph on vertices $\{1,2, \ldots, d\} \times \{1',2',\ldots, d'\}$.
  A $U$-\emph{ribbon} of length $2\ell$ is a graph $R$ on $2\ell d$ vertices
  \begin{align*}
    a_1^1,\ldots,a_d^1,\ldots,a_1^\ell,a_d^\ell\\
    b_1^1,\ldots,b_d^1,\ldots,b_1^\ell,b_d^\ell\mper
  \end{align*}
  We install edges in $R$ by placing a copy of $U$ on vertices $1,2,\ldots, d $ and $1',2',\ldots, d'$ (with the label $i$ or $i'$ matching the upper index of $a$s and $b$s respectively)
  on $a_1^i,\ldots,a_d^i, b_1^{i-1},\ldots,b_d^{i-1}$ for every $i \leq d$. For $i = 0$, we treat $i-1$ as $d$ (modular addition). Often we will omit the length parameter $2\ell$ when it is clear from context.
\end{definition}

\begin{definition}
  Let $G$ be a graph.
  A \emph{labeled} $U$-ribbon $R$ is a tuple $(R,F)$ where $R$ is a $U$-ribbon and $F : R \rightarrow G$ is a map labeling each vertex of $R$ with a vertex in $G$.
  We require that for $(u,v)$ an edge in $R$, $F(u) \neq F(v)$.
\end{definition}

\begin{definition}
  Let $(R,F)$ be a labeled $U$-ribbon where $U$ has $2d$ vertices.
  We say $(R,F)$ is \emph{disjoint} if for every $i$,
  \[
    |\{ F(a_{i}^1),\ldots,F(a_{i}^d),F(b_{i}^1),\ldots,F(b_{i}^d) \}| = |\{ F(a_{i-1}^1),\ldots,F(a_{i-1}^d),F(b_{i-1}^1),\ldots,F(b_{i-1}^d) \}| = 2d\mper
  \]
\end{definition}
\begin{definition}
  Let $(R,F)$ be a labeled $U$-ribbon where $U$ has $2d$ vertices.
  We say that $(R,F)$ is \emph{contributing} if no element of the multiset $\{ (F(u),F(v)) \, : \, (u,v) \in R \}$ occurs with odd multiplicity.
\end{definition}

The following combinatorial lemma will serve as a tool in the proofs of the main results for this section.
\begin{lemma}
  \label{lem:ribbon-unique-bound}
  Let $(R,F)$ be a contributing labeled $U$-ribbon of length $2\ell$.
  Recall that $R$ has vertex set $a_j^i, b_j^i$ for $i \in \ell$ and $j \in [d]$.
  Let $k \leq d$.
  Suppose that the sets
  \[
    \{F(a^i_1),F(b^i_1) \}_{i \in [\ell]}, \ldots, \{F(a_k^i),F(b_k^i) \}_{i \in [\ell]}, \quad
    \{F(a_j^i),F(b_j^i) \}_{i \in [\ell], j \in [k+1,d]}
  \]
  are disjoint.
  Then if $U$ contains the edges $\{ (1,1),\ldots,(k,k)\}$ (where we identify the vertex set of $U$ with $[d] \times [d]$),
  $\{ F(u) \, : \, u \in R\}$ has size at most $(2d - k)\ell + k$.
\end{lemma}
\begin{proof}
  The assumption on $U$ implies that $R$ contains the cycles
  \begin{align*}
    & C_1 \defeq (a_1^1,b_1^1,\ldots,b_\ell^1,a_1^1)\\
    & \ldots\\
    & C_k \defeq (a_1^k,b_1^k,\ldots,b_\ell^k,a_1^k)\mper
  \end{align*}
  In order for $(R,F)$ to be contributing, every edge $(u,v) \in R$ must have a partner $(u',v') \neq (u,v)$ so that $F(u') = F(u)$ and $F(v') = F(v)$.
  By our disjointness assumption, every edge in cycle $C_i$ must be partnered with another edge in $C_i$.
  Thus, now temporarily identifying edges when they are labeled identically, each $C_i$ is a connected graph with at most $\ell$ unique edges (since each of the $2\ell$ edges must be partnered).
  It thus has at most $\ell + 1$ unique vertex labels.
  Among the cycles $C_1,\ldots,C_k$, there are thus at most $k(\ell + 1)$ unique vertex labels.
  In the rest of the ribbon $R$ there can be at most $2\ell (d -k)$ unique vertex labels, because once the cycles $C_1,\ldots,C_k$ are removed there are only that many vertices left in $R$.
  So in total there are at most $k(\ell + 1) + 2\ell(d - k) = (2d-k)\ell + k$ unique labels.
\end{proof}

The next few definitions and notation are needed in  the proof of \pref{lem:correction-disjoint}.

\begin{definition}
  Let $U$ be a bipartite graph on vertices $a_1,a_2, b_1, b_2$.
  A \emph{fancy $U$-ribbon} $R$ of length $2\ell$ is a graph on vertices $c_1,\ldots,c_{2\ell}, a^1_1, a_1^2,\ldots,a^1_{\ell},a^2_\ell, b_1^1,b_1^2,\ldots,b_\ell^1, b_\ell^2$.
  On the $a$ and $b$ vertices, $R$ restricts to a $U$-ribbon of length $2\ell$.
  Additionally, it has edges $(c_i, a_i^1),(c_i, a_i^2),(c_i, b_i^1),(c_i,b_i^2)$.

  Where $G$ is a graph, a labeled fancy $U$-ribbon is a tuple $(R,F)$ where $R$ is a fancy $U$-ribbon and $F : R \rightarrow G$ labels each vertex of $R$ with a vertex in $G$.
  We require for any edge $(u,v) \in R$ that $F(u) \neq F(v)$.
\end{definition}

\begin{lemma}
  \label{lem:fancy-ribbon-unique-bound}
  Let $U$ be a nonempty bipartite graph on vertices $a_1,a_2,b_1,b_2$.
  Let $(R,F)$ be a contributing fancy $U$-ribbon of length $2\ell$.
  Suppose that the sets
  \[
    \{ F(a_i^1), F(b_i^1) \}, \quad \{ F(a_i^2), F(b_i^2) \}, \quad \{F(c_i)\}
  \]
  are disjoint.
  Then $\{ F(u) \, : \, u \in R \}$ contains at most $3\ell + 2$ distinct labels.
  If $U$ is empty, then $\{ F(u) \, : \, u \in R \}$ contains at most $4\ell + 2$ distinct labels.
\end{lemma}
\begin{proof}
  First suppose $U$ is nonempty.
  By swapping $a_i^1, a_i^2$ or $b_i^1, b_i^2$ or both as necessary (which does not change whether $(R,F)$ is contributing),
  we may assume that $U$ contains the edge $(a_1, b_1)$ and thus that $R$ contains the edges $(a_i^1, b_i^1)$ and $(b_i^1, a_{i+1}^1)$
  (where as usual addition is modulo $\ell$).

  Because $(R,F)$ is contributing, every edge must have an identically-labeled partner.
  By our disjointness assumptions, edges among $\{a_i^1, b_i^1\}$ may be partnered only to edges similarly among $\{a_i^1, b_i^1\}$.
  Also, edges between $\{c_i\}$ and $\{ a_i^2, b_i^2 \}$ may be partnered only to edges between $\{c_i\}$ and $\{a_i^2, b_i^2 \}$.
  Thus, the $2\ell$-edge-long cycle on vertices $\{ a_i^1, b_i^1 \}$ may have at most $\ell$ uniquely-labeled edges,
  and the $4\ell$-edge-long cycle on vertices $\{ a_i^2, b_i^2, c_i \}$ may have at most $2\ell$ uniquely-labeled edges.
  Since both are connected, the former may have at most $\ell + 1$ unique vertex labels and the latter at most $2\ell + 1$ unique vertex labels.
  Thus there are at most $3\ell + 2$ unique vertex labels in $(R,F)$.

  When $U$ is empty the proof is similar: there are two paths, $\{a_i^1, b_i^1, c_i \}$ and $\{a_i^2, b_i^2, c_2 \}$.
\end{proof}

\subsubsection{Proofs of \pref{lem:fourier-norm-mpw-disjoint} and \pref{lem:correction-disjoint}}

\begin{proof}[Proof of \pref{lem:fourier-norm-mpw-disjoint}]
  By \pref{lem:embed} it is enough to prove the analogous claims for the $n^d \times n^d$ matrix $Q$ with entries given by
  \[
    Q[(a_1,\ldots,a_d), (b_1,\ldots,b_d)] = \begin{cases}
      \prod_{(i,j) \in B} g_{a_i, b_j} & \text{ if $|\{ a_1,\ldots,a_d,b_1,\ldots,b_d \}| = 2d$}\\
      0 & \text{ otherwise}
     \end{cases}\mcom
  \]

  By multiplying $Q$ by suitable permutation matrices $P,P'$ to give $PQP'$, we may assume in the $2$-matching case above that the matching is $\{(1,1), (2,2)\}$ and in the nonempty graph case that the edge contained is $(1,1)$ (where we think of the vertex set of $U$ as $[d] \times [d]$).
  Note that $\|Q\| = \|PQP'\|$.

  We apply \pref{lem:partitioning} to obtain a family of matrices $\{ Q_i \}_{i \in [r]}$ 
  for some $r = O(3^3 \log n) = O(\log n)$ satisfying $Q = \sum_i Q_i$.
  On any entry $(a_1,\ldots,a_d),(b_1,\ldots,b_d)$ on which $Q_i$ is nonzero it is equal to $Q$ at that entry,
  and furthermore for each $Q_i$ there is a partition $(S_1^i,\ldots,S_3^i)$ of $[n]$ so that if $Q_i[(a_1,\ldots,a_d),(b_1,\ldots,b_d)] \neq 0$ then $a_1,b_1 \in S_1^i, a_2,b_2 \in S_2^i$, and $a_j, b_j \in S_3^i$ for all $j > 2$.

  We show that every matrix $\|Q_i\|$ has bounded spectral norm.
  To save on indices, let $N = Q_i$.
  Let $(S_1,S_2,S_3)$ be the partition of $[n]$ corresponding to $N$.
  We bound $\E \Tr (N N^\top)^\ell$ for some $\ell$ to be chosen later.

  Let $\cR(N)$ be the set of contributing disjoint labeled $U$-ribbons $(R,F)$ of length $2\ell$ with $F(a_1^i), F(b_1^i) \in S_1, F(a_2^i), F(b_2^i) \in S_2$ and $F(a_j^i),F(b_j^i) \in S_3$ for $j > 2$.
  Then $\E \Tr (NN^\top)^\ell \leq O(\ell^\ell) |\cR(N)|$.
  (Here we have an inequality rather than an equality because some elements of $\cR(N)$ may correspond to entries of $N$ which are zero because they appeared in some other part of the partitioning scheme and $\ell^\ell \geq \ell!$ accounts for reorderings of the labels.)

  Supposing that $B$ contains a $2$-matching, by \pref{lem:ribbon-unique-bound}, each $(R,F) \in \cR(N)$ contains at most $(2d - 2)\ell + 2$ unique $\{ F(u) \, : \, u \in R \}$.
  So there are at most $n^{2\ell(d - 1) + 2}$ elements of $\cR(N)$.
  It follows by Markov's inequality that for any $\alpha > 0$,
  \[
    \Pr(\|N\| > \alpha) \leq \Pr(\Tr(NN^T)^\ell > \alpha^{2\ell}) \leq \frac{O(\ell^\ell) n^{2\ell(d-1) + 2}}{\alpha^{2\ell}}
  \]
  Choosing $\alpha \geq O(\ell) n^{d-1 + 10/\ell} (\log n)^{1/2\ell} 2^{d^2/2\ell}$ makes this at most $(\ell^\ell n^{10} \log (n)2^{d^2})^{-1}$.
  Choose $\ell = (\log n)^2$ so that there is such an $\alpha$ also satisfying $\alpha = O(n^{d-1} \log(n)^2 )$ (so long as $d \leq O(\log n)$ as assumed).

  Taking a union bound over the $\log n$ matrices $Q_i$, we get that
  \[
    \Pr(\text{exists $i$ with} \|Q_i \| > O(\log(n)^2 n^{d-1})) \leq n^{-10} 2^{-d^2}
  \]
  and so by the triangle inequality applied to $\|M\| = \|\sum_i Q_i\|$, we get
  \[
    \Pr(\|Q\| > O(n^{d-1} \log(n))^3) \leq n^{-10} 2^{-d^2}\mper
  \]

  The case that $B$ contains only a $1$-matching is similar, replacing the $(2d - 2)\ell + 2$ unique vertices in a contributing $B$-ribbon with $(2d - 1)\ell + 1$, again by \pref{lem:ribbon-unique-bound}.
\end{proof}

\begin{proof}[Proof of \pref{lem:correction-disjoint}]
  We first handle the case when $U$ is non empty. By \pref{lem:embed} it is enough to prove the analogous statement for the $n^2 \times n^2$ matrix, also by abuse of notation denoted $Q$,
  which is the sum of matrices (abusing notation again) $Q_k$ with entries given by
  \[
    Q^k[(a_1,a_2 ),(b_1,b_2 )] = \begin{cases}
      g_{k,a_1}g_{k,a_2}g_{k,b_1}g_{k,b_2}\prod_{(i,j) \in U} g_{a_i, b_j} & \text{ if $|\{ a_1,a_2,b_1,b_2 \}| = 4$}\\
      0 & \text{ otherwise}
     \end{cases}\mper
  \]
By multiplication with an appropriate permutation matrix (which cannot change the spectral norm), we may assume that $U$ contains the edge $(1,1)$.  We begin with \pref{itm:partition-fancy} from \pref{lem:partitioning}, whose hypotheses are satisfied by our convention $y_{a,a} = 0$.
  This gives a family $Q_1,\ldots,Q_r$ with $r = O(3^3 \log n) = O(\log n)$ so that $\sum_{i = 1}^r Q_i$ and a corresponding family of partitions $(S^1_1,S^1_2,T^1),\ldots,(S_1^r,S_2^r,T^r)$.
  \pref{lem:partitioning} guarantees that $Q_i[(a_1,a_2),(b_1,b_2)] = \sum_{k \in T} y_{k,a_1}y_{k,a_2}y_{k,b_1}y_{k,b_2} \prod_{(i,j) \in U} y_{a_i,b_j}$ for some $T \subseteq T^i$ when $a_1,b_1 \in S_1^i$ and $a_2,b_2 \in S_2^i$ and is zero otherwise.

  Fix some $i \in [r]$ and let $N = Q_i$ (to save on indices).
  We will bound $\E \Tr (NN^\top)^\ell$ for some $\ell$ to be chosen later.
  Let $(S_1,S_2,T)$ be the partition of $[n]$ corresponding to $N$.

  Let $\cR(N)$ be the set of contributing labeled fancy $U$-ribbons of length $2\ell$ so that for each $c_i \in R$, $F(c_i) \in T$, for each $a_i^1, b_i^1$ we have $F(a_i^1), F(b_i^1) \in S_1$, and for each $a_i^2, b_i^2$ we have $F(a_i^2), F(b_i^2) \in S_2$.

  Expanding $\E \Tr (NN^\top)^\ell$ as usual, we see that $\E \Tr (NN^\top)^\ell \leq \ell^\ell |\cR(N)|$.
   (As in the proof of \pref{lem:fourier-norm-mpw-disjoint}, we have an inequality rather than an equality because some entries of $N$ may not have a sum over all elements of $T$ if there is overlap with previous parts of the partitioning scheme.)
   By \pref{lem:fancy-ribbon-unique-bound}, $|\cR(N)| \leq {n \choose {3\ell + 2}} \leq n^{3\ell + 2}$.

   By Markov's inequality,
   \[
     \Pr(\|N\| > \alpha) \leq \Pr(\Tr(NN^\top)^\ell > \alpha^{2\ell}) \leq \frac{\ell^\ell n^{3\ell + 2}}{\alpha^{2\ell}}\mper
   \]
   Taking $\alpha \geq \ell n^{3/2 + 13\ell /2}$ guarantees that this is at most $n^{-11}$.
   If $\ell = \Theta(\log n)$, then there is such an $\alpha$ satisfying also $\alpha = O(n^{3/2} \log(n))$.

   By a union bound and triangle inequality, we then get
   \[
     \Pr(\|Q\| > O(n^{3/2} (\log n)^2 ) \leq O(n^{-10})\mper\qedhere
   \]

   The proof in the case of $U$ empty is similar, using \pref{lem:fancy-ribbon-unique-bound} in the empty $U$ case.
\end{proof}

\section{Analyzing Deviations for the Degree-$d$ MPW Operator}
\label{sec:degree-d}
In this section, we use the tools developed in \pref{sec:tools} to analyze the spectrum of the deviation matrix $D = \M' -E$ and prove \pref{lem:M'-is-PSD}. 

As noted in \pref{sec:degree-d}, we decompose $\M'= E+D$. For any $I,J \in \nchoose{d}$, $D(I,J)$ depends on a) $\deg(I \cup J)$ and b) whether $\sE_{ext}(I,J) \subseteq G$.
  If $D(I,J)$ depended only on b) above, then it could be decomposed into a sum of patterned matrices defined in \pref{sec:tools}; analyzing these is tractable.
Our first step is thus to get rid of the dependence on $\deg(I \cup J)$---the only part depending on the entire graph.
We will obtain a matrix $L$ that depends only on whether $\sE_{ext}(I,J) \subseteq G$ or not (and thus is ``locally random" in the sense of \cite{MPW15}).

Specifically, we write $D = L + \Delta$ where $L$ is the locally random part obtained by replacing $D(I,J)$ by $\E[ D(I,J) \mid \sE_{ext}(I,J) \subseteq G]$ whenever $\sE_{ext}(I,J) \subseteq G$ and an appropriate negative constant when $\sE_{ext}(I,J) \nsubseteq G$ (this makes the expectation of each entry over $G \sim G(n,\frac{1}{2})$ to be $0$). More concretely, following \cite{MPW15}, we define:
\begin{definition}
\label{def:Ddecomp}

\begin{equation} \alpha(i) = \frac{{\omega \choose {2d-i}}}{{{2d} \choose {2d-i}}} \cdot {{n-2d+i} \choose i} \cdot 2^{-d^2 - {i \choose 2}},
\end{equation} 
and $p(i) = 2^{-(d-i)^2}$ for each $i$. We set $L(I,J)$ for every $I, J \in \nchoose{d}$ to be
\[
L(I,J) = \begin{cases}
\alpha(|I \cap J|) \cdot \frac{(1-p(|I \cap J|))}{p(|I \cap J|)} \text { if } \sE(I \cup J) \setminus (\sE(I) \cup \sE(J)) \subseteq G\\
- \alpha(|I \cap J|) \text{, otherwise}.
\end{cases}
\]
We define $\Delta = D-L$.
\end{definition}
The idea behind the definition is that $L(I,J) = \E_{G \sim G(n,\frac{1}{2}}[ \M'(I,J) \mid \sE_{int}(I,J) \subseteq G]$ whenever $\sE_{ext}(I,J) \subseteq G$ and in the other case, chosen to make $\E_{G}[ L(I,J)] = 0$. We will analyze $L$ and $\Delta$ separately. The proof of \pref{lem:M'-is-PSD} is broken into two main pieces. Each piece analyzes the action of $L$ and $\Delta$ split across various eigenspaces $V_0, V_1, \ldots, V_d$ of the matrix $E$. Such fine grained analysis for the case of $d=2$ was done in \cite{DM15}. A few points of distinctions from \cite{DM15} are in order at this point. 

The first is regarding the high level approach. The approach of \cite{DM15} used explicit expressions for a canonical set of eigenvectors in $V_1$ to obtain similar conclusions as us for the case of $d=2$. This approach gets unwieldy very quickly because the explicit entries of eigenvectors for $V_i$ for $i >1$ are hard to work with \cite{Bannai-Ito}. We tackle this issue by developing an argument that doesn't need explicit entries of the eigenvectors. Instead, we use basic representation theory (\pref{sec:patterned}) to identify a set of symmetries satisfies by vectors in $V_i$ for each $i$ and use it obtain the conclusions we require. 

Second, \cite{DM15} deal with the optimization version of the degree $4$ SOS program which, as noted in the introduction, could be potentially weaker than the one we analyze here (and thus our lower bound is technically stronger). This simplifies the analysis in \cite{DM15} a little bit as the matrix $\Delta$ defined above is identically zero for the operator analyzed. We explicitly work with the feasibility version of the degree 4 SOS program and thus, must deal with the additional complexity of handling $\Delta$. It turns out that we have to do a fine grained analysis of the $\Delta$ matrix itself. The decomposition we use for $\Delta$ is somewhat different from the case of $L$ even though, the analysis of each piece of the decomposition proceeds similar to the case of $L$. 

Third, for the special case of $d=2$, essentially the only matrix one has to analyze is the $L_0$, the matrix obtained by zeroing out all entries $(I,J)$ in $L$ such that $I \cap J \neq \emptyset$: a uniform bound on spectral norm of the remaining component suffices. However, for higher $d$, one has to deal with the ``non-disjoint" entries with some care and an argument analogous to the one in \cite{DM15} fails to show PSDness of $\M'$ beyond $\omega \approx n^{\frac{1}{2d}}$ giving no asymptotic improvement over \cite{MPW15}. 

Finally, our argument for analyzing the spectral norms of each of the pieces also needs to be much more general than in case of \cite{DM15} to handle higher degrees. For this, we identify a simple combinatorial structure (size of maximum matchings in appropriate bipartite graphs on $2d$ vertices) that controls the bounds and could also be used to obtain slick proofs of the conclusions required in \cite{DM15} in the context of analyzing $\M'$ for $d=2$. Our combinatorial argument itself is a generalization of the one given by \cite{DM15} for this case.

We now go on to describe the two lemmas that encapsulate the technical heart of the proof of \pref{lem:M'-is-PSD}. The first does a fine grained analysis of spectrum of $L$. In the following, let $\cM'$ be the filled-in matrix for the degree-$d$ MPW operator at clique size $\omega$ (\pref{def:MPW-Mprime}), $E = \E_{G \sim G(n,1/2)} [\cM']$, $D = \cM' - E$, $L$ be as \pref{def:Ddecomp}, $\Pi_i$ be the projectors to the spaces $V_i$ of \pref{fact:johnsondecomp}.
\begin{lemma}[Bounding Blocks of $L$] \label{lem:dev-d}
With probability at least $1- \frac{1}{n}$ over the draw of $G \sim G(n,\frac{1}{2})$, 
each $0 \leq i \leq d$ satisfies
\begin{enumerate}
\item  $$ \left|\Pi_i L \Pi_j \right| \leq 2^{2d} \tilde{O}(\omega^{2d} n^{d-\frac{1}{2}}),$$

\item If $i, j \geq 2$, then 
$$ \left| \Pi_i L \Pi_j \right| \leq 2^{2d} \tilde{O}(\omega^{2d} n^{d-1}).$$

\end{enumerate}
\end{lemma}

The next lemma does a (even more) fine grained analysis of the spectrum of $\Delta$:

\begin{lemma}\label{lem:Delta-bound}
With probability at least $1-\frac{1}{n}$ over the draw of $G \sim G(n,\frac{1}{2})$, for each $0 \leq i \leq d$:
$$ \left| \Pi_i \Delta \Pi_j \right| \leq \tilde{O}( 2^{O(d)} \omega^{2d-\min\{i,j\}} n^{d-\frac{1}{2}}) + \tO( 2^{O(d)} \omega^{2d - q} n^{d-1}).$$
\end{lemma}

We can now use \pref{lem:dev-d} and \pref{lem:Delta-bound} to complete the proof of \pref{lem:M'-is-PSD}.

\begin{proof}[Proof of \pref{lem:M'-is-PSD}]
For each $0 \leq i,j \leq d$, we compute $\Pi_i \M' \Pi_j$ and use \pref{lem:blockpsd}. We write $\M' = E + L + \Delta$. First, $\Pi_i^{\dagger} E \Pi_j = 0$ whenever $i \neq j$ as $\Pi_i$ are projectors to eigenspaces of $E$. Let $\lambda_0, \lambda_1, \ldots, \lambda_d$ be the eigenvalues of $E$ on eigenspaces $V_0, V_1, \ldots, V_d$. From \pref{lem:Eeigenvals}, we have: $$\lambda_j \geq 2^{-O(d^2)} \cdot n^d \cdot \omega^{2d-j}.$$ Thus, $$\Pi_i E \Pi_i \geq 2^{-O(d^2)} n^d \omega^{2d-j}.$$ 

In what follows, all our statements hold with probability at least $1-O(1)/n$:

Using \pref{lem:dev-d} and \pref{lem:Delta-bound}, for every $i \geq 2$, $$\| \Pi_i (L+\Delta) \Pi_i \| \leq \tilde{O} (\omega^{2d} n^{d-1}) + \tilde{O}(\omega^{2d-i} n^{d-\frac{1}{2}}).$$ On the other hand, when $i \leq 2$, $$ \| \Pi_i (L+\Delta) \Pi_i \| \leq \tilde{O}(\omega^{2d} n^{d-\frac{1}{2}}).$$ Then, it is easy to check that for any $\omega = O(n^{\frac{1}{d+1}})$, $$\Pi_i \M' \Pi_i = \Pi_i (E+ L + \Delta)  \Pi_i \geq 2^{-O(d^2)} \omega^{2d-i} n^d.$$

Next, we bound the cross terms $|\Pi_i (L + \Delta) \Pi_j |$ for $i \neq j$. Again, using \pref{lem:dev-d} and \pref{lem:Delta-bound}, we have for $i, j \geq 2$:

$$|\Pi_i (L + \Delta) \Pi_j| \leq \tilde{O}( \omega^{2d} n^{d-1}) + \tO( \omega^{2d-\min\{i, j \}} n^{d-\frac{1}{2}} ).$$

For $\omega \leq O(n^{\frac{1}{d+1}})$, it is again easy to check that, for $i, j \geq 2$, the above expression is at most $\frac{2}{d} \sqrt{\lambda_i \lambda_j}$. 

In the case when one of $i,j$ is at most $1$, we have the bound $$|\Pi_i (L + \Delta) \Pi_j| \leq \tilde{O}( \omega^{2d} n^{d-\frac{1}{2}}).$$ In this case, it is easy to check that so long as $\omega = O( n^{\frac{1}{d+1}} / \poly \log{(n)})$, $$ |\Pi_i (L + \Delta) \Pi_j| \leq \tilde{O}( \omega^{2d} n^{d-\frac{1}{2}}) \leq \frac{2}{d} \sqrt{\lambda_i \lambda_j}.$$

By an application of \pref{lem:blockpsd}, the proof is complete.
\end{proof}

\subsection{Proof of \pref{lem:dev-d}}
\paragraph{Proof Plan} We first describe the high level idea of the proof.

We start by decomposing $L = \sum_{q = 0}^d L_{q}$ where $L_q(I,J) = L(I,J)$ if $|I \cap J| = q$ and $0$ otherwise.
Notice that each $L_q$ then is obtained by a scaling an appropriate $0/1$ matrix.

Most illuminating is the \emph{disjoint} case $L_0$, which is nonzero only at entries $I,J$ with $I \cap J = \emptyset$.
For any disjoint $I,J$, $L_0(I,J)$ depends only whether $\sE_{ext}(I,J) \subseteq G$, which, one could write as an appropriately scaled AND function of the indicators $g_e$ of edges $e \in \sE_{ext}(I,J)$.
We can expand this AND function in the monomial (parities of subsets of $g_e$ variables) basis. Each such monomial corresponds to the bipartite graph $B$ that contains the pairs $e \in \sE_{ext}(I,J)$ that constitute the monomial. This gives a decomposition of $L_0$ into $2^{d^2}-1$ (since the constant term is $0$, $L$ being zero mean) components, $L_0^B$ for each non empty, labeled bipartite graph on $[d] \times [d]$.

We can bound the spectral norm of each of the pieces $L_0^B$ by direct application of tools derived in \pref{sec:spectralnorm}. The main work in this section goes into showing that depending upon the structure of $B$, an appropriate selection of subspaces $V_i$ lie in left or right kernels of $L_0^{B}$. Thus, for a fixed term $\Pi_i L \Pi_j$, some $L_0^B$ do not contribute. We identify the maximum spectral norm among contributing terms to obtain the final bound.

To accomplish this goal, we rely heavily on the tools built in \pref{sec:patterned} which give us a handle on the symmetries of the eigenspaces $V_0,V_1, \ldots, V_d$. This requires some work based on representation theory of finite groups and is presented in \pref{lem:eigenbasisstructure} and \pref{lem:kernelsymoperators}. 

The case of $L_{q}$ for $q\neq 0$ needs even finer decomposition. We decompose each $L_q$ ($q > 0$) into matrices that identify the ``pattern" of the $q$ intersecting vertices.
In \cite{MPW15} a similar idea is used to reduce the task of bounding the spectral norm of $L_q$ to a calculation similar to one in the case of $L_0$. However, unlike \cite{MPW15}, we also require properties of the kernels of the components of the decomposition. After restricting to a fixed intersection pattern of $q$ vertices, we thus resort to using a generalization of the kernel analysis used for the $L_0$ case.
%
We now proceed with the proof plan as described beginning with the decomposition of each $L_q$.
\subsubsection{Decomposing $L$} \label{sec:D-decomp}
We start by decomposing $L$ further as $L = \sum_{q = 0}^d L_q$ where  for any $I, J \in \nchoose{d}$,
\[
L_{q} = \begin{cases}
L(I,J) &\text{ if } |I \cap J| = q,\\
0 & \text{ otherwise }.
\end{cases}
\]

\paragraph{Decomposing $L_0$}
Recall that $\B$ is the set of all bipartite labeled graphs with left and right vertex sets labeled by $[d]$ and $[d]$.
Recall also from \pref{sec:patterned} that for any $I,J \subseteq [n]$ with $|I| = |J| = d$, the graph $\zeta_{I,J}(B)$ is a copy of $B$ on vertex sets $I,J$ where the correspondence between $I$ and $[d]$ and $J$ and $[d]$ is determined by the sorting map $\zeta$.
Finally, recall that for a graph $G$ on $[n]$, we let $g_b$ be the $\pm 1$ indicator for the presence of edge $b$ in $G$, and by convention $g_{b} = 0$ when $b = (i,i)$ for any $i \in [n]$.
For any $B \in \B$, define an $\nchoose{d} \times \nchoose{d}$ matrix $$\tilde{L}_0^{B}(I,J) = \alpha(0) \cdot (2^{d^2} -1) \Pi_{b \in \zeta_{I,J}(B)} g_b.$$

The idea is to write $L_0$ as a sum of such matrices $\tilde{L}^{B}_0$ with the entries corresponding $I, J$ where $|I \cap J| \neq 0$ zeroed out.
Thus, define $L_0^{B}$ to be the matrix with $(I,J)$ entry given by 
\[
L_0^B(I,J) = \begin{cases}
\tilde{L}_0^B(I,J)  &\text{ if } |I \cap J| = 0\\
0 &\text{ otherwise}.
\end{cases}
\]
We think of $L_0$ as a rescaling and centering of a $0/1$ matrix whose entries are the AND of the $\pm 1$ indicators for the edges in $\sE_{ext}(I,J)$.
Decomposing these ANDs into monomials over those $\pm 1$ indicators, we see that each monomial corresponds exactly to one bipartite graph $B$, and the centering of $L_0$ corresponds to removing the constant monomial, which corresponds to the empty bipartite graph.
Every other monomial recieves equal weight $2^{-d^2}$ in this expansion, and so from these observations it becomes routine to verify that $$2^{-d^2} \sum_{B \in \B} L_0^{B} = L_0.$$

\paragraph{Decomposing $L_q$} Similarly, we further decompose $L_{q}$ for $q > 0$.
Here things are a bit more involved.
Let us motivate our decomposition by understanding the structure of the matrix $L_q$ for $q > 0$ a little bit.
Consider an entry $(I,J)$ such that $I \cap J = K$.
Then, $\sE_{ext}(I , J) \subseteq \sE_{ext} (I \setminus K, J \setminus K)$.
Thus, the edge structure in the bipartite subgraph on vertex sets $I \setminus K$ and $J \setminus K$ decides the value of $L(I,J)$ for any graph $G$ and we can hope to a get a patterned matrix.
We now follow this intuition.

Recall the sorting maps $\zeta_I: [d] \rightarrow I$ and $\zeta_J: [d] \rightarrow J$.
Letting $Z_\ell, Z_r \subseteq [d]$ be subsets of size $q$, we define $L^{Z_\ell,Z_r}_{q}$ such that:
\[
L^{Z_\ell,Z_r}_{q} ( I,J) = \begin{cases}
L_{q}(I,J) & \text{ if } \zeta_{I}(Z_\ell) = \zeta_{J}(Z_r)\\
0 & \text{ otherwise}.
\end{cases}
\]
That is, $L^{Z_\ell,Z_r}_{q}$ is the ``part" of $L_q$ where any $I,J$ intersect in a (size $q$) subset given by $\zeta_{I}(Z_\ell)$ and $\zeta_{J}(Z_r)$.
It is then easy to see that $L_{q} = \sum_{Z_\ell,Z_r} L^{Z_\ell,Z_r}_{q}.$
Next, we decompose each $L^{Z_\ell,Z_r}_{q}$ further based on non-empty labeled bipartite graphs $B \in \B_{Z_\ell,Z_r}$ for each $Z_\ell,Z_r$.

We now define a matrix which is nonzero only on entries which intersect in \emph{at least} $q$ places: 
\[
\tilde{L}_q^{B,Z_\ell,Z_r} (I,J) = \begin{cases}
\alpha(q) \cdot (2^{(d-q)^2}-1) \Pi_{b \in \zeta_{I,J}(B)} g_b & \text{ if } \zeta_I(Z_\ell) = \zeta_J(Z_r)\\
0 & \text{ otherwise.}
\end{cases}
\]
Again, as before, the actual decomposition needs to zero out the entries $(I,J)$ such that $|I \cap J| \neq q$.
Thus, we define $L_q^{B,Z_\ell,Z_r}$ by zeroing entries of $\tilde L_q^{B,Z_\ell, Z_r}$ which intersect also outside of $Z_\ell, Z_r$:
\[
L_q^{B,Z_\ell,Z_r}(I,J) = \begin{cases}
\tilde{L}_q^{B,Z_\ell,Z_r}, \text{ if} |\zeta_I([d] \setminus Z_\ell) \cap \zeta_J([d] \setminus Z_r)| = 0\\
0, \text{ otherwise.}
\end{cases}
\]
Finally, it is again easy to verify that: $$L_q = 2^{-(d-q)^2} \sum_{Z_\ell,Z_r \subseteq [d]} \sum_{B \in \B_{Z_\ell,Z_r}} L_q^{B,Z_\ell,Z_r}.$$
%

\subsubsection{Spectral Analysis of $L$}
In order to prove Lemma \ref{lem:dev-d}, we will first use the decomposition described in the previous section to write $L_q$ as a sum of appropriate patterned matrices.
We will then partition the sum into groups, each group corresponding to an equivalence class of (left or right) similar bipartite graphs $B$.
We will infer some properties about the kernel and finally use the spectral norm bounds from \pref{sec:spectralnorm} to complete the proof. 

More concretely, let $(B,Z_\ell,Z_r)$ be a $q$ pattern (as defined in \pref{sec:patterned}).
From our decomposition from the previous section, we have:
\begin{equation}
   L_q = 2^{-(d-q)^2} \sum_{Z_\ell,Z_r} \sum_{B \in \B_{Z_\ell,Z_r}} L_q^{B,Z_\ell,Z_r} \label{eq:maindecomp}
\end{equation}
Our idea is to analyze appropriate collections of $L_q^{B,Z_\ell,Z_r}$ separately.
When $q = 0$, $Z_\ell,Z_r$ are redundant (being $\emptyset$) and thus $L_0^{B,Z_\ell,Z_r}= L_0^{B}$ in that special case.
In the first step, we observe that $\tilde{L}_q^{B,Z_\ell,Z_r}$ has some symmetries that are helpful to us.
We thus want to deal with the sums of $\tilde{L}_q^{B,Z_\ell,Z_r}$ instead of $L_q^{B,Z_\ell,Z_r}$.
To justify this, we start by showing that the difference of the above two matrices has small norm. 

\begin{claim} \label{claim:Lltilde}
For any $(B,Z_\ell,Z_r)$, $$|| L_q^{B,Z_\ell,Z_r} - \tilde{L}_q^{B,Z_\ell,Z_r} || \leq  \tilde{O}(\omega^{2d-q} n^{d-1}).$$
\end{claim}

Next, we bound each $\| \tilde{L}_{q}^{B,Z_\ell,Z_r} \|$ using the machinery from \pref{sec:spectralnorm}.
\begin{claim}[Norm Bounds on Pieces] \label{claim:norm-bounds}
With probability at least $1-1/n^{10}$ over the draw of $G \sim G(n,\frac{1}{2})$, 
\begin{enumerate}

\item $$\|\tilde{L}^{B,Z_\ell,Z_r}_q\| = \tilde{O}(\omega^{2d-q} \cdot n^{d-\frac{1}{2}} ).$$

\item  If $|B_{\ell}|, |B_r| \geq 2$, then: $$\| \tilde{L}^{B,Z_\ell,Z_r}_q \| = \tilde{O}(\omega^{2d-q} \cdot n^{d-1}).$$

\end{enumerate}
\end{claim}

At first, it should be worrisome that some of the $\tilde{L}_q^{B,Z_\ell,Z_r}$ have norms that are much larger than what we need (in the second claim of \pref{lem:dev-d}).  What comes to our rescue is the fact that the components $\tilde{L}_q^{B,Z_\ell,Z_r}$ that have large norm do not contribute to quadratic forms on $\left|\Pi_i^{\dagger} L \Pi_j\right|$ when $i, j$ are at least $2$. The crucial observation that allows us to conclude this is based on the observation that $\tilde{L}_q^{B,Z_\ell,Z_r}$ are patterned matrices in the sense of \pref{def:patterned} and thus clubbing all $(B',Z_\ell',Z_r')$ that are (left or right) similar to $(B,Z_\ell,Z_r)$, we can show that certain $V_i$ lie in their kernels. More specifically:

\begin{claim} \label{claim:disjoint-label}
For $t > q+|B_{\ell}|$, $$\Pi_t^{\dagger} \left( \sum_{(B,Z_\ell',Z_r') \sim_{\ell} (B,Z_\ell,Z_r)} \tilde{L}_q^{B',Z_\ell',Z_r'}\right)  = 0.$$ Similarly, for any  $w > q+ |B_r |$, $$ \left(\sum_{(B',Z_\ell',Z_r') \sim_{r} (B,Z_\ell,Z_r)} \tilde{L}_q^{B,Z_\ell,Z_r} \right) \Pi_{w} = 0.$$
\end{claim}

Before proving the three claims above, we show how they imply \pref{lem:dev-d}:
\begin{proof}
We use \eqref{eq:maindecomp} to write: 
\begin{equation}
L_q = 2^{-(d-q)^2} \left( \sum_{B,Z_\ell,Z_r} \tilde{L}_q^{B,Z_\ell,Z_r} + \sum_{B,Z_\ell,Z_r} (L_q^{B,Z_\ell,Z_r} - \tilde{L}_q^{B,Z_\ell,Z_r}) \right) \label{eq:centralforthisproof}
\end{equation}
\begin{align*} 
L_q &= 2^{-(d-q)^2} \left( \sum_{B,Z_\ell,Z_r} \tilde{L}_q^{B,Z_\ell,Z_r} + \sum_{B,Z_\ell,Z_r} (L_q^{B,Z_\ell,Z_r} - \tilde{L}_q^{B,Z_\ell,Z_r}) \right)\\
&\leq  2^{-(d-q)^2} \left( \sum_{B,Z_\ell,Z_r} \|\tilde{L}_q^{B,Z_\ell,Z_r}\| \right) + 2^{-(d-q)^2}\cdot \sum_{B,Z_\ell,Z_r} \| L_q^{B,Z_\ell,Z_r} - \tilde{L}_q^{B,Z_\ell,Z_r} \|\\
&\text{ Using \pref{claim:norm-bounds}}\\
&\leq 2^{2d} \cdot \tilde{O}(\omega^{2d-q} n^{d-\frac{1}{2}}).
\end{align*}

For the second part, fix an $i \geq 2$. We first show that some terms in the decomposition in \eqref{eq:centralforthisproof} do not contribute to $\left|\Pi_i^{\dagger} L_q \Pi_i\right|$.  

Consider any bipartite graph $B$ such that $|B_{\ell}| < 2$. Then, we have from Claim \ref{claim:disjoint-label}, $\left( \sum_{(B',Z_\ell',Z_r') \sim_{\ell} (B,Z_\ell,Z_r)}L_q^{B',Z_\ell',Z_r'} \right) \Pi_t = 0$ for every $t \geq 2$. Thus, $|\Pi_i^{\dagger} \sum_{B: |B_{\ell}| < 2} L_q^{B,Z_\ell,Z_r} \Pi_i |= 0$ for any $i$ and every $t \geq 2$. Similarly, when $|B_r| \geq 2$, $|\Pi_i^{\dagger} L_q^{B,Z_\ell,Z_r} \Pi_i |= 0$.  On the other hand, when both $|B_{\ell}|, |B_r| \geq 2$, from Claim \ref{claim:norm-bounds}, we have (with high probability over the draw of $G \sim G(n,\frac{1}{2})$), $\|L_q^{B,Z_\ell,Z_r} \| \tilde{O}(\omega^{2d-q} \cdot n^{d-1}).$ Thus:

Thus, for $i \geq 2$, we have:
\begin{align*}
||\Pi_i^{\dagger} L_q \Pi_i||  &= 2^{-(d-q)^2} \cdot \| \sum_{B, Z_\ell,Z_r} \Pi_i^{\dagger} L_q^{B,Z_\ell,Z_r} \Pi_i \|\\
&\leq 2^{-(d-q)^2} \sum_{B,Z_\ell,Z_r} |\Pi_i^{\dagger} \tilde{L}_q^{B,Z_\ell,Z_r} \Pi_i | + 2^{-(d-q)^2} \cdot \sum_{B ,Z_\ell,Z_r} | L_q^{B,Z_\ell,Z_r} - \tilde{L}_q^{B,Z_\ell,Z_r} |\\
&= 2^{-(d-q)^2} \sum_{B:|B_{\ell}| \geq 2, |B_r| \geq 2,Z_\ell,Z_r} |\Pi_i^{\dagger} \tilde{L}_q^{B,Z_\ell,Z_r} \Pi_i| + \tilde{O}(\omega^{2d} \cdot n^{d-1})\\
&\leq \tilde{O}(\omega^{2d-q} \cdot n^{d-\frac{1}{2}}) .
\end{align*}
Similarly, for $i+j \geq 2$ we must have $i, j \geq 2$. Thus, by a calculation similar to above, $\|\Pi_i^{\dagger} L_q \Pi_j\| \leq \tilde{O}(\omega^{2d} \cdot n^{d-\frac{1}{2}}).$ 
\end{proof}
In the remaining part of this section, we complete the proofs of Claims \ref{claim:disjoint-label} and \ref{claim:norm-bounds}. 

\subsubsection{Proof of Claims}
In this section, we obtain quick proofs of the three claims above using the tools developed in \pref{sec:tools}.

We first prove \pref{claim:Lltilde}.
\begin{proof}[Proof of \pref{claim:Lltilde}]
The proof is by appealing to \pref{fact:Gershgorin}. Observe that 
\[
|\tilde{L}_q^{B,Z_\ell,Z_r}(I,J) - L_q^{B,Z_\ell,Z_r}(I,J)| \leq \begin{cases} 
0 \text{ , if } |I \cap J | \leq q,\\
\alpha(q) \cdot (2^{(d-q)^2}-1) 2^{-(d-q)^2} \text{ , otherwise}.
\end{cases}
\]

We now estimate $$ \max_{I \in \nchoose{d}} \sum_{J \in \nchoose{d}} |\tilde{L}_q^{B,Z_\ell,Z_r}(I,J) - L_q^{B,Z_\ell,Z_r}(I,J)| \leq 2^d n^{d-q-1} \cdot \alpha(q) \cdot (2^{(d-q)^2}-1) 2^{-(d-q)^2} .$$ The claim now follows from \pref{fact:Gershgorin}.
\end{proof}

The next is a direct application of \pref{lem:kernelsymoperators}.
\begin{proof}[Proof of Claim \ref{claim:disjoint-label}]
The main observation is that $\tilde{L}^{B,Z_\ell,Z_r}_q$ is a patterned matrix (with a $q$-pattern $(B,Z_\ell,Z_r)$) in the sense of \pref{def:patterned}. The result then follows immediately by appealing to \pref{lem:kernelsymoperators}.
\end{proof}

Finally, we prove Claim \ref{claim:norm-bounds} using \pref{lem:fourier-norm-mpw-disjoint}.
\begin{proof}[Proof of Claim \ref{claim:norm-bounds}]
First, consider the case of $\tilde{L}_0^{B}$ ($Z_\ell = Z_r = \emptyset$ in this case). We write $||\tilde{L}_0^{B}|| \leq ||L_0^B|| + ||\tilde{L}_0^{B} - L_0^B|| .$ For the second term, we can appeal to \pref{claim:Lltilde}. For the first term, observe that by a direct application of \pref{lem:fourier-norm-mpw-disjoint}, $||L_0^B|| \leq \alpha(0) (2^{d^2}-1) 2^{-d^2} \cdot \tilde{O}(n^{d-\frac{1}{2}}).$ Further, when $|B_{\ell}|, |B_r| \geq 2$, then, $B$ has a $2$-matching and thus, by another application of \pref{lem:fourier-norm-mpw-disjoint}, we obtain that in this case, $||L_0^B|| \leq \alpha(0) \cdot (2^{d^2}-1) 2^{-d^2} \cdot \tilde{O}(n^{d-1}).$ The proof is thus complete for the case of $q = 0$. 

We now reduce the computation for the more general case to similar calculations by appealing to the idea of lifts. Consider the matrix $R \in \R^{\nchoose{d-q} \times \nchoose{d-q}}$ given by 
\[
R(I',J') = \begin{cases}
L_q^{B,Z_\ell,Z_r} (I' \cup K, J' \cup K) \text{ if } I' \cap J' = \emptyset.\\
0 \text{,  otherwise,}
\end{cases} 
\]
where $K \subseteq [n]$ is some fixed subset of size $q$ such that $K \cap I' = K \cap J' = \emptyset$.

Then, $L_q^{B,Z_\ell,Z_r} = R^{(q)}$. Thus, using \pref{fact:liftsofnondisjoint}, $\|L_q^{B,Z_\ell,Z_r}\| \leq 2^{2d} \|R\|$. Since $R$ has non zero entries only when the row and column indices are disjoint sets, we can apply \pref{lem:fourier-norm-mpw-disjoint} to $R$ to obtain $\|R\| \leq \alpha(q) (2^{(d-q)^2}-1) 2^{-(d-q)^2} \cdot \tilde{O}(n^{d-q-\frac{1}{2}})$. Further, when $|B_\ell|, |B_r|\geq 2$, again by an application of \pref{lem:fourier-norm-mpw-disjoint}, we have:$\|R\| \leq \alpha(q) (2^{(d-q)^2}-1) 2^{-(d-q)^2} \cdot \tilde{O}(n^{d-q-1})$. Now, using $\|\tilde{L}_q^{B,Z_\ell,Z_r}\| \leq \| L_q^{B,Z_\ell,Z_r} \| + \| \tilde{L}_q^{B,Z_\ell,Z_r} - L_q^{B,Z_\ell,Z_r}\|$ and using \pref{claim:Lltilde} completes the proof.
\end{proof}

\subsection{Proof of \pref{lem:Delta-bound}}
We now move on to analyzing the spectrum of the matrix $\Delta$.


The high level plan of the proof is similar to that of \pref{lem:dev-d}. We define $\Delta_i$ for each $0 \leq i \leq d$ as follows:
\[\Delta_q(I,J) = \begin{cases}
\Delta(I,J) \text{, if } |I \cap J| = q\\
0 \text{,  otherwise.}
\end{cases}
\]
We further split $\Delta_q = \sum_{K: |K| = q} \Delta_{q,K},$ where:
\[
\Delta_{q,K}(I ,J) = \begin{cases}
\Delta_q(I, J) \text{, if} I \cap J = K\\
0 \text{,  otherwise. }
\end{cases}
\]
First, we observe that $\Delta_0 = 0$. This is because $\deg(I \cup J)$ when $I$ and $J$ are disjoint is exactly $1$ and doesn't depend on the graph. Thus, $\Delta = \sum_{i = 1}^n \Delta_i.$ As before we would like to spot patterned matrices in each $\Delta_i$ to show that appropriate eigenspaces $V_j$ lie in the kernel of $\Delta_i$. In case of $\Delta_q$, however, there's a difference how this needs to be done. This is because each entry $(I,J)$ of $\Delta_i$ potentially depends on the edges from every vertex in the graph $G$ to $I$ and $J$. This is unlike the case of $L$ where the $(I,J)$ entry depends only on the edges between $I$ and $J$ (in fact that's the reason we separated $L$ from $\Delta$ in the analysis). Nevertheless, we give a decomposition below that will help us make claims similar to the ones in the case of analyzing $L$ in this case too.

Let us first explain the main idea in the decomposition. The entry $\Delta_q(I,J)$ depends on two events: a) whether $\sE_{ext}(I,J) \subseteq G$ and b) the number of subsets $S \subseteq [n]$ of size $|S| = q$ that form $2d$-cliques with $I \cup J$. The main observation that motivates our decomposition is the following: in the event that $\sE_{ext}(I,J) \subseteq G$, the deviation in $\deg(I \cup J)$ is completely captured (up to low order terms) by just the number of vertices $s$ that has an edge to all of $I \cup J$ in $G$. This allows us to write entries of $\Delta_q$ as a sum of contribution to the deviation due to each vertex $s$ separately. For the case of $q = 1$, this argument is in fact exact and there are no low order terms. When $q > 1$, the contributions due to individual vertices contribute the bulk of the deviation and only low order terms remain. 

From here on, we are in a situation similar to the one encountered in analyzing $L$ in the previous subsection. We show that the components in the decomposition with large spectral norm do not contribute to quadratic forms over eigenspaces with small eigenvalues of the expectation matrix $E$ using the idea of patterned matrices from \pref{sec:patterned}. We show that the remaining components have small spectral norm using the combinatorial techniques combined with the trace moment method developed in \pref{sec:spectralnorm}. 


We now proceed to make the ideas above more precise. We begin by some notation and a definition. We first define $e^1_{s,K} \in \R^{\nchoose{d}}$ for any $s \in [n]$ and $K \subseteq [n]$, $|K| = q$ as follows:
\[
e^1_{s,K}(I,J) = \begin{cases}
0 & \text{, if } I \cap J \neq K \text{ or } s \in I \cup J\\
2^{|I|-q}-1 \text{ otherwise and if } \sE_{ext}(\{s\}, I \setminus J) \subseteq G\\
-1 \text{,  otherwise}.
\end{cases}
\]

Similarly, we define $e^2_s \in \R^{\nchoose{d}}$ for any $s \in [n]$ and $K \subseteq [n]$, $|K| = q$:\\
\[
e^2_{s,K}(I,J) = \begin{cases}
0 & \text{, if } I \cap J \neq K \text{ or } s \in I \cup J\\
2^{|J|-q}-1 \text{ otherwise and if } \sE_{ext}(\{s\}, J \setminus I) \subseteq G\\
-1 \text{,  otherwise}.
\end{cases}
\]

Next, we define: $e_{s,K} \in \R^{\nchoose{d}}$ for any $s \in [n]$ and $K \subseteq [n]$ satisfying $|K| = q$:

\[
e^3_{s,K}(I,J) = \begin{cases}
0 & \text{, if } I \cap J \neq K \text{ or } s \in I \cup J\\
2^{q}-1 \text{ otherwise and if } \sE_{ext}(\{s\}, K ) \subseteq G\\
-1 \text{,  otherwise}.
\end{cases}
\]

Finally, we define
\[
e^{1,2}_{K}(I,J) = \begin{cases}
0 & \text{ , if } I \cap J \neq K  \\
2^{|I| + |J| -q} -1 \text{ otherwise and if } \sE_{ext}(I,J) \subseteq G\\
-1 \text{, otherwise.}
\end{cases}
\]
 Using the matrices above, we can show the following approximate factorization for the entries of $\Delta_{q,K}$:
 \begin{lemma}\label{lem:Delta-decomp}
For every $I,J$ such that $|I \cap J| = K$,  $$ \left| \Delta_{q,K}(I, J) - \eta \frac{{ {{\omega} \choose {2d-q}}}}{{{2d} \choose {2d-q}}} (1+e_K^{1,2}(I,J))  \sum_{s \in [n]} \left( (1+e^1_{s,K})(1+e^2_{s,K})(1+e^3_{s,K}) -1 \right)\right| \leq 2^{O(d)} \cdot \tilde{O}(\omega^{2d-q} \cdot n^{q-1}),$$ for $\eta = 2^{-2|I \cup J| + {{|I \cup J| + 1} \choose 2} - {{2d} \choose 2}} \frac{(n - |I|)^{2d - |I| - 1}}{(2d - |I| - 1)!} .$
\end{lemma}
\begin{proof}[Proof of \pref{lem:Delta-decomp}]

Let $A_{I \cup J}$ be the set of vertices $s$ in $G$ not in $I \cup J$ so that $(s,i) \in G$ for all $i \in I \cup J$.
By definition, if $I \cup J$ is a clique,
\[
  \sum_{s \in [n]} (1+e_{s,K}^{1,2}(I,J))\left( (1+e^1_{s,K})(1+e^2_{s,K})(1+e^3_{s,K}) -1 \right) = 2^{2|I \cup J|} (|A_I| - 2^{-|I \cup J|} (n - |I \cup J|))\mper
\]
Applying the scaling $\eta$, we get
\[
  \eta  (1+e^{1,2}(I,J))  \sum_{s \in [n]} \left( (1+e^1_{s,K})(1+e^2_{s,K})(1+e^3_{s,K}) -1 \right) = \left(|A_I| - \frac{n-|I|}{2^{|I|}}\right)\frac{2^{{{|I|+1} \choose 2} - {{2d} \choose 2}}(n-|I|)^{2d-|I|-1}}{(2d-|I|-1)!}
\]
and hence the lemma follows from \pref{thm:cliqueboundtheoremtwo}.
\end{proof}

We will need another definition before proceeding: For each $i$,
\[
e^i_{s,q} \eqdef \sum_{K: |K| = q} e^i_{s,K} 
\]

As in the case of analyzing $L$, we define the filled in versions of $e^i_{s,K}$ as follows:

\[
\tilde{e}^1_{s,K}(I,J) \eqdef \begin{cases}
0 & \text{, if } K \nsubseteq I \cap J \text{ or } s \in I \cup J\\
2^{|I|-q}-1 \text{ otherwise and if } \sE_{ext}(\{s\}, I) \subseteq G\\
-1 \text{,  otherwise}.
\end{cases}
\]
\[
\tilde{e}^2_{s,K}(I,J) \eqdef \begin{cases}
0 & \text{, if } K \nsubseteq I \cap J \text{ or } s \in I \cup J\\
2^{|J|-q}-1 \text{ otherwise and if } \sE_{ext}(\{s\}, J) \subseteq G\\
-1 \text{,  otherwise}.
\end{cases}
\]
\[
\tilde{e}^1_{s,K}(I,J) \eqdef \begin{cases}
0 & \text{, if } K \nsubseteq I \cap J \text{ or } s \in I \cup J\\
2^{q}-1 \text{ otherwise and if } \sE_{ext}(\{s\}, K) \subseteq G\\
-1 \text{,  otherwise}.
\end{cases}
\]

We start by giving norm bounds on all the matrices involved in the decomposition in \pref{lem:Delta-decomp}.

\begin{lemma} \label{lem:Delta-all-norms}
\begin{enumerate}
\item For each $i \in \{1,2,3\}$, $$\left \| \sum_{s \in [n], K:|K| = q} e^i_{s,K} \right \| \leq \tilde{O}(2^{O(d)} \cdot n^{d-\frac{1}{2}}).$$
\item For each $i \in \{1,2,3\}$, $$\left \| \sum_{s \in [n], K:|K| = q} e^i_{s,K} - \tilde{e}^i_{s,K} \right \| \leq \tO(2^{O(d)} \cdot n^{d-{3/2}}).$$
\item $$ \left \| \sum_{K:|K| = q} (1+e_{K}^{1,2}(I,J)) \odot \sum_{s \in [n]} \left( (1+e^1_{s,K}) \odot (1+e^{2}_{s,K}) \odot (1+e^3_{s,K}) - \sum_{i = 1}^{3} e^i_{s,K} \right) \right \| \leq \tO( 2^{O(d)} \cdot n^{d-1}).$$
\end{enumerate}
\end{lemma}

\begin{proof}
Fix a $q$ and let $Q $ be any matrix in $\R^{\nchoose{d} \times \nchoose{d}}$ that appears in the statement of the lemma above. Let $R \in \R^{\nchoose{d-q} \times \nchoose{d-q}}$ be defined by 
\[
R(I,J) = \begin{cases}
\sum_{s \in [n], K: |K| = q}  Q (I \cup K, J \cup K) \text {, if} I \cap J = \emptyset\\
0 \text{ , otherwise.}
\end{cases}
\]
Then, $\sum_{s \in [n], K:|K| = q} Q = R^{(q)}$ in the sense of \pref{def:liftsofnondisjoint}. Thus, by \pref{fact:liftsofnondisjoint} $\| R^{(q)} \| \leq 2^{2d} \|R\|.$ Thus, we focus on bounding $\|R\|$ in the following. We will use the trace moment method for this purpose and the argument is similar to the ones made in \pref{sec:spectralnorm} to develop the general purpose spectral concentration results. For this reason, we will be a bit more terse than in the case of the other applications of the trace moment method before. We set up the notation for the general $R$ as above and specialize the combinatorial reasoning for each of the specific matrices involved later.

We expand \begin{multline} \E[ \Tr((RR^{\dagger})^{\ell}] = \E [ \sum_{K_1, K_2, \ldots, K_{2\ell}} \sum_{I_1, I_2, \ldots, I_{2\ell}, s_1, s_2, \ldots, s_{2\ell}} R(I_1\setminus K_1, I_2 \setminus K_2) R(I_3 \setminus K_3,I_2 \setminus K_2)\\ \cdots R(I_{2\ell-1} \setminus K_{2\ell-1}, I_{2\ell} \setminus K_{2\ell}) \cdot R(I_{2\ell} \setminus K_{2\ell},I_1 \setminus K_1)].\end{multline} We now investigate when does a term in the expansion above contribute a non-zero value to the LHS.

First consider the case of $Q = \sum_{s,K} e^{i}_{s,K}$. Fix $i = 1$, the other cases are similar. $e^{1}_{s,K} (I,J)$ is a function of the variables $g_b$ for $b \in \sE_{ext}(\{s\},I)$ (whenever $I \cap J = K$). Writing $e^1_{s,K}(I,J)$ as a polynomial in $g_b$ for $b \in \sE_{ext}(\{s\},I)$ we observe that: $\E_{G}[e^1_{s,K}] = 0$ and that all coefficients of degree $j$ polynomials are equal for every $j$ and at most $2^{d}$. We decompose the matrix $e^1_{s,K}$ so that for each $(I,J)$ we only pick one of the (corresponding) terms in the polynomial expansion of each entry in $g_b$ described above. 

In the expansion of the expected trace above, then, for any such matrix that appears in the decomposition, each term is a (scaled) product of $g_b$ variables for some $b$. For the expectation of such a term to be non zero, each $g_b$ must occur an even number of times. Consider the case of a matrix in the decomposition of $e^1_{s,K}$  with entries being some (corresponding) monomials of degree $1$ for concreteness. The case of other matrices is similar. Fix any term. Let $T$ be the set of all vertices that appear in some $I_{i}$ for $i \leq 2\ell$ and are part of some $b$ for $g_b$ that appears in the term. Then, by a random partitioning argument based on \pref{itm:partition-fancy}, we can first  assume that all $\{s_1, s_2, \ldots, s_{2\ell}\}$ doesn't intersect $T$ (and lose a logarithmic factor in the spectral norm upper bound). It is now immediate that every $s \in \{s_1, s_2, \ldots, s_{2\ell} \}$ for a term to have non zero expectation (otherwise, some $g_b$ will not appear twice in the product sequence describing the term). Thus, the number of distinct vertices in any term with a non zero expectation is $\ell + (d-1) 2\ell = (d-\frac{1}{2}) (2\ell)$. The number of possible terms with the same set of distinct vertices is at most $(2\ell)!$. Finally, each term contributes at most $2^{d}$. Thus, we can upper bound the expected trace of such a matrix by $2^{O(d)} \cdot (2\ell)! \cdot n^{d-\frac{1}{2}} \poly \log {(n)}.$ We now do the standard step of using the Markov's inequality to obtain an upper estimate on $\Tr((QQ^{\dagger})^{2\ell})$ that holds with probability $1-1/n$, take $(2\ell)^{th}$ root and finally use $\ell = O(\log{(n)}$ to obtain the desired bound. Finally, matrix in the decomposition based on polynomial expansion in $g_b$ of the entries of $e^1_{s,K}$ can be similarly upper bounded in spectral norm completing the analysis of this case by an application of the triangle inequality. 

The case of the matrix $ \sum_{s \in [n], K:|K| = q} e^i_{s,K} - \tilde{e}^i_{s,K}$ is similar, except that it is now a $(q+1)^{th}$ lift of some matrix. Repeating the reasoning as above, (except each entry $(I,J)$ being described by sets of size $d-2$ instead of $d-1$), we obtain the stated upper bounded in the statement of the lemma. 

Finally, we now proceed to the analysis for the third case. We first split the Hadamard product into a sum and analyze each term separately. We sketch the main difference from above for the combinatorial picture for the term $Q = \sum_{s,K} e^1_{s,K}e^2_{s,K}$ here. The other arguments are similar. We again write the $\E[ \Tr((QQ^{\dagger})^{\ell})]$ as a sum over $K_1, \ldots, K_{2\ell}, s_1, s_2, \ldots, s_{2\ell}$ and $I_1, I_2, \ldots, I_{2\ell}$ as above. By expanding out each of $e^1_{s,K}$ and $e^2_{s,K}$ as polynomials in appropriate $g_b$ variables, we observe that the least degree of any term in the expansion is at least $2$ (i.e. involves a product of at least $2$ $g_b$s). Decomposing the matrix so that each entry gets the corresponding monomial in the polynomial expansion in terms of $g_b$s as above, we now consider the matrix with the term involving a scaled product of exactly $2$ $g_b$s. The other matrices in the decomposition can be handled similarly. By a random partitioning argument (\pref{itm:partition-fancy}) as before, we can assume that $\{s_1, s_2, \ldots s_{2\ell}\}$ are disjoint from $T$, the subset of vertices from $I_i$ for $i \leq 2\ell$ that appear in $b$ for some $g_b$ and lose a $O(\log{(n)})$ factor in the estimate on the norm. Again reasoning as before that for a non zero expectation, the term must have each $g_b$ appear an even number of times. Since each $g_b$ must appear at least twice (whenever it appears at least once), the number of distinct $g_b$s that appear in a term that contributes non zero expectation is at most $2\ell$. On the other hand, we now observe that the $\{s_1, s_2, \ldots, s_{2\ell}\} \cup T$ are all connected via a path using $b$ from $g_b$ that appear in the term and thus, the number of distinct elements in $\{s_1, s_2, \ldots, s_{2\ell}\} \cup T$ are at most $2\ell+1$. Thus, a non zero contributing term has a total of at most $2\ell(d-2) + 2\ell = 2\ell(d-1)$. Arguing in the standard way as done above, this now yields a norm estimate of $\tilde{O}(2^{O(d)} \cdot n^{d-1})$ as required. 

\end{proof}

Next, we show find out the spaces $\tilde{e}^i_{s,K}$ contribute to. Towards this, for each $i$, we define : $$\tilde{e}^i_{s,q} \eqdef \sum_{K: |K| = q} \tilde{e}^i_{s,K}.$$ Then, we have:

\begin{lemma} \label{lem:Delta-Kernel}
For any $t ,t' > q$, we have:

\begin{enumerate}

\item $$ \Pi_t \tilde{e}^2_{s,q} = 0,$$
\item $$ \tilde{e}^1_{s,q} \Pi_{t} = 0,$$
\item $$ \Pi_{t} \tilde{e}^3_{s,q} \Pi_{t'} = 0.$$
\end{enumerate}

\end{lemma}

\begin{proof}
We only do the proof for the first case, the others are similar. The idea is again to use \pref{lem:eigenbasisstructure}. Let $v \in \R^{\nchoose{d}}$. We will show that $ \tilde{e}^2_{s,q} \cdot v = w$ such that there exists a vector $u \in \R^{\nchoose{q}}$ such that for every $I$, $w_{I} = \sum_{I' \subseteq I \text{, }  |I'| = q} u_{I'}.$ An application of \pref{lem:eigenbasisstructure} then completes the proof of part (1).

We have:
\begin{align*}
w_I &= \sum_{J \in \nchoose{d}} \tilde{e}^2_{s,q} (I,J) v_J \\
&= \sum_{K:|K| = q} \sum_{J \in \nchoose{d}} \tilde{e}^2_{s,K} (I,J) v_J\\
&= \sum_{K: K \subseteq I \text{, } |K| = q} \sum_{J \in \nchoose{d}} \tilde{e}^2_{s,K} (I,J) v_J.
\end{align*}

The proof now follows by observing that $\sum_{J \in \nchoose{d}} \tilde{e}^2_{s,K} (I,J) v_J$ depends only on $K$. Thus, we set $u_{K}$ for every $|K| = q$ by $u_K = \sum_{J \in \nchoose{d}} \tilde{e}^2_{s,K} (I,J) v_J$ for any $I$ such that $K \subseteq I$.

\end{proof}

We can now complete the proof of \pref{lem:Delta-bound} using \pref{lem:Delta-all-norms}and \pref{lem:Delta-Kernel}.
\begin{proof}[Proof of \pref{lem:Delta-bound}]
For each $q$, let $A_q$ be the expression given by \pref{lem:Delta-decomp} to approximate the entries of $\Delta_q$. Then, we have for any $i,j$:

\begin{align*}
\Pi_i \Delta \Pi_j &= \sum_{q = 1}^{d} \Pi_i \Delta_q \Pi_j\\
&= \sum_{q = 1}^{d} \Pi_i \left( \Delta_q -A_q \right)  \Pi_j + \sum_{q = 1}^{d} \Pi_i A_q \Pi_j \\
\end{align*}

By a simple application of \pref{fact:Gershgorin}, it is easy to observe that $\| \Delta_q -A_q \| \leq 2^{O(d)} \tilde{O}(\omega^{2d-q} \cdot n^{d-1}).$
An application of \pref{lem:Delta-Kernel} and \pref{lem:Delta-all-norms} yields that the terms that contribute to $\Pi_{i} A_q \Pi_j$ have norm at most $2^{O(d)} \tilde{O}(\omega^{2d-\min\{i,j\}} n^{d-\frac{1}{2}}) + 2^{O(d)} \tilde{O}( \omega^{2d-q} n^{d-1}).$ This completes the proof.

\end{proof}

\section{Analyzing Deviations for the Corrected Degree-4 Operator}
\label{sec:deg4-psd}
The goal of this section is to prove \pref{lem:deg4-PSD-claim}, that is, to show that $\cN' \succ 0$. The proof is organized into $5$ main claims that we next present.

We first show that it is enough to prove PSDness of a somewhat simplified matrix $\cN$.
$\cN$ is produced by two simplifications to $\cN'$.
First, to take care of the zero rows as in \pref{sec:degree-d}, we work with a matrix where we ``fill in'' the entries carefully.
Second, $\cN$ and $\M'$ are equal on all entries $(I,J)$ such that $|I \cup J| \leq 3$; in other words, the correction affects only the homogeneous degree $4$ parts.
Specifically, let $\cR \in \R^{\nchoose{2} \times \nchoose{2}}$ be defined so that:
\[
\cR(I,J) = \begin{cases} \gamma (\tfrac {\omega} n)^5 C_4 \cdot \sum_{s \in [n]} \Pi_{a \in I \cup J} r_s(a) & \text{ if $|I \cup J| = 4$ and $\sE_{ext}(I,J) \subseteq G$}\\
  0&  \text{ otherwise}\mper
  \end{cases}
\]
(Recall that $C_4$ is the number of $4$-cliques in $G$.)
We then set $$\cN = \M' + \cR\mcom$$ where $\cM'$ is the filled in MPW matrix (\pref{def:MPW-Mprime}).

Our first claim shows that it is enough to prove PSDness of $\cN$:
\begin{lemma}
  \label{lem:N4prime}
For any $\gamma,c$ there is $\omega_0 = \Omega(\sqrt {n/\log(n)c \gamma})$ so that for any $\omega \leq \omega_0$ with probability $1 - O(n^{-10})$ if $\cN \succ \omega^2 n^2 / c \cdot I$ then $\cN' \succ 0$.
\end{lemma}
Next, we decompose the matrix $\cN$ appropriately and study the spectrum of each of the pieces.
Towards this, we define $\tilde \cR_0 \in \R^{\nchoose{2} \times \nchoose{2}}$ as follows.
For every $I,J$:
\[
  \tilde \cR_0(I,J) =
  \frac 1 {16} \sum_{s \in [n]} \gamma (\tfrac {\omega} n)^5 C_4 \cdot \Pi_{a \in (I \cup J) \setminus (I \cap J)} r_s(a)
\]
Recall that from  \pref{def:Ddecomp}, we know that $\cM' = E+ L + \Delta$. By writing $\cR = \tilde \cR_0 + (\cR  - \tilde \cR_0)$, we obtain the decomposition:
\begin{equation}
\cN =  (E + \tilde \cR_0) + L + \Delta + (\cR - \tilde \cR_0). \label{eq:deg4-decomp-N}
\end{equation}
In what follows, we will analyze each piece of the decomposition above separately on a carefully constructed decomposition of $\R^{\nchoose{2}}$.
We now proceed and construct this decomposition.
Recall $\R^{\nchoose{2}} = V_0 \oplus V_1 \oplus V_2$ where $V_0, V_1,V_2$ are the eigenspaces of the matrix $E = \E[ \M'] = \E[ \cN]$ from \pref{fact:johnsondecomp} and \pref{lem:Eeigenvals}.
Let $r_s \in \R^n$ be as described in \pref{def:corrected-moms}.
With slight abuse of notation, we write $r_s^{\otimes 2}$ for the vector in $\R^{\nchoose{2}}$ such that for every $I \in \nchoose{2}$, $$r_s^{\otimes 2} (I) = \Pi_{i \in I} r_s(i).$$

We now define a new decomposition by splitting $V_2$ further and write $\R^{\nchoose{2}} = W_0 \oplus W_1 \oplus W_{1.5} \oplus W_2$, where:
\begin{align}
    W_0 & = V_0\mcom \notag \\
    W_1 & = V_1\mcom \notag \\
    W_{1.5} & \text{ s.t. } W_{1.5} \perp (W_0 \oplus W_1) \text{ and } W_0 \oplus W_1 \oplus W_2 = V_0 \oplus V_1 \oplus \Span \{ r_s^{\oplus 2} \} \notag \\
    W_2 & = (W_1 \oplus W_2 \oplus W_3)^\perp\mper  \label{def:new-decomp}
\end{align}
Let $\Pi_{W_a}$ be the projector to $W_a$ for every $a \in \{0,1,1.5,2\}$. 

We are now ready to analyze the spectrum of each piece from \pref{eq:deg4-decomp-N}.
First, we analyze the spectrum of $(E+ \tilde \cR_0)$:
\begin{lemma}
  \label{lem:deg4-ER}
  For every $\gamma$ there is $\omega_0 = \Omega(\sqrt{\gamma n}/ \log(n)^2)$ so that for $\omega \leq \omega_0$,
  with probability $1 - O(n^{-10})$,
  $$(E + \tilde \cR_0) \succeq \Omega(\omega^4n^2) \cdot \Pi_{W_0} + \Omega(\omega^3n^2) \cdot \Pi_{W_1} + \Omega(\gamma \omega^5 n) \cdot \Pi_{W_{1.5}} + \Omega(\omega^2n^2) \cdot \Pi_{W_2}.$$
\end{lemma}

Next, we analyze the spectrum of $L$.
Here at last we see the main technical improvement in these corrected moments---the cross term between $V_1$ and $V_2$ has become a cross term between $W_1$ and $W_2$ and has much-reduced norm.
\begin{lemma}
  \label{lem:deg4-main}
  For any $\omega =\tO(\sqrt{n})$ there is $p = \polylog n$ so that, with probability $1 - O(n^{-10})$ the following bounds hold.
  \begin{enumerate}
    \item Diagonal Terms:
    \begin{align*}
      \|\Pi_{W_a} L \Pi_{W_a}\| & \leq O(p\omega^4 n^{3/2}) \quad \text{for $a \in \{0,1,1.5\}$}\\
      \|\Pi_{W_2} L \Pi_{W_2}\| & \leq O(p\omega^4 n)\mper
    \end{align*}
    \item Off-Diagonal Terms:
     \begin{align*}
      \| \Pi_{W_a} L \Pi_{W_b} \| & \leq O(p\omega^4 n^{3/2}) \quad \text{for $a,b \in \{0, 1, 1.5, 2\}$}\\
      \| \Pi_{W_{a}} L \Pi_{W_2} \| & \leq O(p\omega^3 n^{3/2}) \quad \text{for $a \in \{1,1.5\}$}\mper
    \end{align*}
  \end{enumerate}
\end{lemma}

Next, we bound the spectral norm of $\Delta$.
This is a direct corollary of the more general bound in \pref{lem:Delta-bound},
but to have a self-contained proof of the degree-4 case we also give a proof later in this section.
\begin{lemma}
  \label{lem:deg4-delta}
  Let $\cM'$ be the ${n \choose 2} \times {n \choose 2}$ filled-in matrix for the degree-4 MPW moments with clique size $\omega$ (\pref{def:MPW-Mprime}).
  Let $\Delta$ be as in \pref{def:Ddecomp}.
  With probability $1 - O(n^{-10})$, $\|\Delta \| \leq O(\omega^3 n^{3/2} \log(n)^2)$.
\end{lemma}

Finally, we bound the spectral norm of the last piece $(\cR - \tilde \cR_0)$:
\begin{lemma}
 \label{lem:deg4-corr-dev}
 Let $G \sim G(n,1/2)$.
 With probability $1 - O(n^{-10})$, $\|\cR - \cR_0 \| \leq O(\gamma \omega^5 n^{1/2} \log(n)^2)$.
\end{lemma}

The proofs of these lemmas follow, but first we complete the proof of \pref{lem:deg4-PSD-claim} and hence of \pref{thm:deg-4-sqrtn}.
\begin{proof}[Proof of \pref{lem:deg4-PSD-claim}]
  By \pref{lem:N4prime}, it will be enough to exhibit $c,\gamma = \polylog n$ and $\omega_1 = \Omega(\sqrt{n/\log(n) c\gamma})$ so that $\cN \succ (\omega^2 n^2 / c) \cdot I$ with probability $ 1- O(n^{-10})$ when $\omega \leq \omega_1$.
  Then our final bound will be given by the minimum of $\omega_1$ and $\omega_0$ of \pref{lem:N4prime}.
  (Recall that $\gamma$ is a parameter inside $\cN$.)
  In the following, all that we claim happens with probability at least $1-n^{-9}$ by a union bound.

  So let $\omega_0 \in \R$; we will choose it later.
  We will find $c,\gamma$ and conditions on $\omega_0$ so that the conditions of \pref{lem:blockpsd} hold for $\cN - ({\omega^2 n^2}/c)\cdot I$.

  First of all, by  \pref{lem:deg4-corr-dev}, for every $\gamma$ there is $c = O(\min\{n/\omega^3\gamma, 1 \})$ so that
  \[
    E - (\omega^2 n^2 /c) \cdot I \succeq \Omega(\omega^4n^2) \cdot \Pi_{W_0} + \Omega(\omega^3n^2) \cdot \Pi_{W_1} + \Omega(\gamma \omega^5 n) \cdot \Pi_{W_{1.5}} + \Omega(\omega^2n^2) \cdot \Pi_{W_2}\mper
  \]
  We assume $c = c(\gamma)$ is chosen in this way.

  For any $\gamma$, by \pref{lem:deg4-corr-dev} if we choose $\omega \leq \sqrt{n} /\gamma \log(n)^2$ then $\|\cR - \cR_0 \| \leq o(\omega^2 n^2)$.
  Adding $\cR - \cR_0$ to the previous equation,
  \[
  (E + \cR_0) + (\cR - \cR_0) - ({\omega^2 n^2}/{c}) \cdot I \succeq \Omega(\omega^4n^2) \cdot \Pi_{W_0} + \Omega(\omega^3n^2) \cdot \Pi_{W_1} + \Omega(\gamma \omega^5 n) \cdot \Pi_{W_{1.5}} + \Omega(\omega^2n^2) \cdot \Pi_{W_2}\mper
  \]
  By the same reasoning using \pref{lem:deg4-delta} to add $\Delta$ to the previous equation, we have for the same choice of $\omega$:
  \begin{equation}
    \cN - L - ({\omega^2 n^2}/{c}) \cdot I \succeq \Omega(\omega^4n^2) \cdot \Pi_{W_0} + \Omega(\omega^3n^2) \cdot \Pi_{W_1} + \Omega(\gamma \omega^5 n) \cdot \Pi_{W_{1.5}} + \Omega(\omega^2n^2) \cdot \Pi_{W_2}\mper \label{eq:deg4-0}
  \end{equation}
  So it just remains to add $L$ to the left-hand side.

  We decompose $L$ as:
  \[
    L = (\Pi_{W_0} + \Pi_{W_1} + \Pi_{W_{1.5}}) L (\Pi_{W_0} + \Pi_{W_1} + \Pi_{W_{1.5}}) + \Pi_{W_2} L + L \Pi_{W_2}\mper
  \]
  Let $p$ be as in \pref{lem:deg4-main}, which implies that
  \begin{equation}
    (\Pi_{W_0} + \Pi_{W_1} + \Pi_{W_{1.5}}) L (\Pi_{W_0} + \Pi_{W_1} + \Pi_{W_{1.5}}) \preceq O(p \omega^4 n^{3/2}) (\Pi_{W_0} + \Pi_{W_1} + \Pi_{W_{1.5}})\mper \label{eq:deg4-1}
  \end{equation}
  Choosing $\gamma = O(p^2 \log(n))$ we get that this is $o(\sqrt{\omega^4 n^2 \cdot \gamma \omega^5 n})$.
  So using \pref{lem:matrix-cs} to add \pref{eq:deg4-0} and \pref{eq:deg4-1}, we obtain
  \begin{equation}
    \cN - \Pi_{W_2} L - L \Pi_{W_2} - ({\omega^2 n^2}/{c}) \cdot I \succeq \Omega(\omega^4n^2) \cdot \Pi_{W_0} + \Omega(\omega^3n^2) \cdot \Pi_{W_1} + \Omega(\gamma \omega^5 n) \cdot \Pi_{W_{1.5}} + \Omega(\omega^2n^2) \cdot \Pi_{W_2}\mper\label{eq:deg4-2}
  \end{equation}
  We break $\Pi_{W_2} L$ apart as
  \[
    \Pi_{W_2} L = \Pi_{W_2} L \Pi_{W_0} + \Pi_{W_2} L \Pi_{W_2} + \Pi_{W_2} L (\Pi_{W_1} + \Pi_{W_{1.5}})\mper
  \]
  By \pref{lem:deg4-main},
  \begin{align*}
     \|\Pi_{W_2} L \Pi_{W_0} \| & \leq O(p\omega^4 n^{3/2}) = o(\sqrt{\omega^4 n^2 \cdot \omega^2 n^2}) & \quad \text{ for } \omega \leq \sqrt{n} / p\log(n)\\
     \|\Pi_{W_2} L \Pi_{W_2} \| & \leq O(p \omega^4 n) = o(\omega^2 n^2) & \quad \text{ for } \omega \leq \sqrt{n} / p \log(n)\\
     \| \Pi_{W_2} L (\Pi_{W_1} + \Pi_{W_{1.5}}) \| & \leq O(p \omega^3 n^{3/2}) = o(\sqrt{ \omega^2 n^2 \cdot \gamma \omega^5 n}) & \quad \text{ for } \gamma \geq p^2 \log p\mper
  \end{align*}
  Together with \pref{lem:matrix-cs} and \pref{eq:deg4-2}, this implies the lemma, for $\gamma \geq p^2 \log p$, $c = c(\gamma)$ as above, and $\omega_0 \leq \min \{\sqrt n / p \log(n)^2, \sqrt n / \gamma \log(n)^2 \}$.
\end{proof}

%
%

\subsection{Proof of Diagonal and Off-Diagonal Norm Bounds (\pref{lem:deg4-main})}
Here we prove \pref{lem:deg4-main}.

\begin{proof}[Proof of \pref{lem:deg4-main}]
  We start with the easy parts.
  Note that $\Pi_{W_2} L \Pi_{W_2} = \Pi_{W_2} \Pi_{V_2} L \Pi_{V_2} \Pi_{W_2}$, so the bound
  \[
    \|\Pi_{W_2} L \Pi_{W_2}\| \leq \tO(\omega^4 n)
  \]
  is immediate from \pref{lem:dev-d}.
  The same theorem also implies that $\|L\| \leq \tO(\omega^4 n^{3/2})$,
  and since projectors are contractive this finishes the bounds on the diagonal terms and the first part of the off-diagonal bound (for cross terms among $W_0, W_1, W_{1.5}$).

  We just have to prove that $\| \Pi_{W_{1}} L \Pi_{W_2} \| \leq \tO(\omega^4 n)$.
  We will use the patterned matrix machinery to show this.
  We recall the decomposition of $L$ from \pref{sec:degree-d} as $L = L_0 + L_1 + L_2$. 
  The main point is to show that $\|\Pi_{W_1} L_0 \Pi_{W_2} \| \leq \tO(\omega^4 n)$, so we postpone to the end of the proof the case of $L_1$ and $L_2$.

  Recall again that $L_0$ can be further decomposed as $L_0 = O(1) \cdot \sum_{B \in \cB} L_0^B$ (again, see \pref{sec:degree-d} for definitions).
  Finally, for each $L_0^B$ there is a $\tilde L_0^B$ so that $\|L_0^B - \tilde L_0^B \| \leq \tO(\omega^4 n)$ with probability $1 - O(n^{-10})$ (see \pref{claim:Lltilde}).
  Since $|\cB| = O(1)$, it is actually enough for us to show that
  \[
    \Norm{ \Pi_{W_1} \Paren{\sum_{B \in \cB} \tilde L_0^B } \Pi_{W_2} } \leq \tO(\omega^4 n)\mper
  \]

  Let $\cB_1$ be the bipartite graphs on $[2] \times [2]$ for which at least one right-hand vertex has degree $0$.
  Then by \pref{lem:kernelsymoperators}, $\sum_{B \in \cB} \tilde L_0^B \Pi_{W_2} = 0$, since $W_2 \subseteq V_2$.

  We are left with $\cB_2$, the family of bipartite graphs where every right-hand-side vertex has nonzero degree.
  Using \pref{lem:fourier-norm-mpw-disjoint} on the matrices $L_0^B$, we get that if $B$ has no vertices of degree $0$ then $\|L_0^B \| \leq \tO(\omega^4 n)$ with probability $1 - O(n^{-10})$.
  Since as above $\|L_0^B - \tilde L_0^B \| \leq \tO(\omega^4 n)$ with similar probability, we get $\| \tilde L_0^B \| \leq \tO(\omega^4 n)$ for these $B$.

  The only remaining graphs in $\cB$ are the two graphs $B_1, B_2$ with exactly one vertex on the left-hand side of degree $0$.
  It is not hard to check that rows $\{ a, b \}$ of the matrices $\tilde L_0^{B_1}$ and $\tilde L_0^{B_2}$ are matrices are $r_a^{\tensor 2}$ and $r_b^{\tensor 2}$, respestively.
  Since $W_2 \perp r_s^{\tensor 2}$, we get $\tilde L_0^{B_1} \Pi_{W_2} = \tilde L_0^{B_2} \Pi_{W_2} = 0$.

  It remains to handle $L_1$ and $L_2$.
  By \pref{claim:norm-bounds} together with \pref{claim:Lltilde}, each satisfies $\|L_1\|, \|L_2\| \leq \tO(\omega^3 n^{3/2})$.
  Since $\omega \leq \sqrt n$, the lemma now follows.
\end{proof}

\subsection{Lower-Degree Cleanup (\pref{lem:N4prime})}
\label{sec:deg4-lowdeg-cleanup}
In this section we prove \pref{lem:N4prime}.
We start by bounding the difference between our pseudoexpectation $\pE$ and the MPW operator on polynomials of degree less than $4$.
This lemma is a consequence of \pref{lem:deg4-scalar} and the Gershgorin circle theorem (\pref{fact:Gershgorin}).
\begin{lemma}
  \label{lem:nondisjoint-entry}
  Let $G \sim G(n,1/2)$.
  Let $\omega$ be a real parameter.
  Let $\pE$ be as given in \pref{def:corrected-moms}.
  Let $\pE_0$ be the MPW operator for clique size $\omega$.
  Suppose $\omega \leq \sqrt n$ and $\gamma = O(\omega)$.
  Then with probability at least $1 - O(n^{-20})$,
  \begin{itemize}
    \item Every $i,j,k \in [n]$ with $i,j,k$ all distinct satisfies
    \[
      |\pE x_i x_j x_j - \pE_0 x_i x_j x_k| \leq O(\gamma \log(n) \omega^4 / n^4)\mper
    \]
    \item Every $i,j \in [n]$ with $i \neq j$ satisfies $|\pE x_i x_j - \pE_0 x_i x_j | \leq o(\omega^2 / n^2)$.
  \end{itemize}
\end{lemma}
\begin{proof}[Proof of \pref{lem:nondisjoint-entry}]
  Recall that we obtained $\pE$ by starting with the clique-size $\omega$ MPW operator on multilinear homogeneous degree-4 polynomials,
   adding a correction operator on those same polynomials, and then infering values of $\pE$ on lower-degree polynomials via the constraint $\sum_i x_i = \omega'$.
  There are two primary sources of the difference between our operator $\pE$ and the MPW operator on polynomials of lower degree.
  The dominant one is the propogation to lower degree polynomials of the correction operator $\cL$ (recall \pref{def:update}).
  The second is that the degree-4 values coming from the MPW part of our operator $\pE$ are propogated downwards using the constraint $\sum_i x_i = \omega'$ rather than $\sum_i x_i = \omega$ as is done to define the rest of the MPW operator.

  We start with the degree $3$ bound.
  Denote by $\pE_0$ the MPW operator for clique size $\omega$.
  Consider $i,j,k \in [n]$ all distinct, and recall that
  \begin{align*}
    \pE x_i x_j x_k & = \frac 1 {\omega' - 3} \sum_{\ell \neq i,j,k} \pE x_i x_j x_k x_\ell \\
    & = \frac 1 {\omega'-3} \sum_{\ell \neq i,j,k} \Paren{\pE_0 x_i x_j x_k x_\ell + \cL x_i x_j x_k x_\ell}\\
    & = \Paren{\frac 1 {\omega' -3} - \frac 1 {\omega - 3}} \Paren{\sum_{\ell \neq i,j,k} \pE_0 x_i x_j x_k x_\ell} + \pE_0 x_i x_j x_k + \frac 1 {\omega'-3} \sum_{\ell \neq i,j,k} \cL x_i x_j x_k x_\ell\mper
  \end{align*}
  Thus,
  \[
    \pE x_i x_j x_k - \pE_0 x_i x_j x_k = \Paren{\frac 1 {\omega' -3} - \frac 1 {\omega - 3}} \Paren{\sum_{\ell \neq i,j,k} \pE_0 x_i x_j x_k x_\ell} + \frac 1 {\omega'-3} \sum_{\ell \neq i,j,k} \cL x_i x_j x_k x_\ell\mper
  \]
  By \pref{lem:deg4-scalar}, with probability $1 - O(n^{-25})$ every $i,j,k$ satisfies $|\sum_{\ell \neq i,j,k} \cL x_i x_j x_k x_\ell | \leq O(\gamma \omega^5 \log(n) / n^4)$.
  At the same time, we know by \pref{lem:deg4-constraints} together with \pref{lem:deg4-scalar} that $|\omega' - \omega| \leq O(\gamma \omega^2 \log(n)^2 / n^{5/2})$.
  This implies that $|1 / (\omega' - 3) - 1/(\omega -3)| \leq O(\gamma \log(n)^2 / n^{5/2})$ (all with probability at least $1 - O(n^{-25})$).
  In conjunction with the preceeding, it implies also that every $i,j,k$ satisfies $(1/(\omega'-3))|\sum_{\ell \neq i,j,k} \cL x_i x_j x_k x_\ell | \leq O(\gamma \log(n) \omega^4 / n^4)$.
  Together with the trivial bound $|\sum_{\ell \neq i,j,k} \pE_0 x_i x_j x_k x_\ell | \leq O(\omega^4 / n^3)$ with probability $1 - n^{-\omega(1)}$ (following from \cite[Theorem 10.3]{MPW15}), all this implies that with probability $1 - O(n^{-25})$,
  \[
    |\pE x_i x_j x_k - \pE_0 x_i x_j x_k| \leq O(\gamma \log(n)^2 \omega^4 /n^{11/2}) + \tO(\gamma \log(n) \omega^4 / n^4) = \tO(\gamma \log(n) \omega^4 / n^4)\mper
  \]

  We turn to the degree-two bound.
  Fix $i \neq j \in [n]$.
  We expand $\pE x_i x_j - \pE_0 x_i x_j$.
  \begin{align*}
    & \pE x_i x_j - \pE_0 x_i x_j\\
     & = \frac 1 {(\omega' - 2)(\omega' -3)} \sum_{k \neq i,j} \pE x_i x_j x_k - \frac 1 {(\omega - 2)(\omega -3)} \sum_{k \neq i,j} \pE_0 x_i x_j x_k\\
    & = \frac 1 {(\omega' - 2)(\omega' - 3)} \sum_{k \neq i,j} (\pE x_i x_j x_k - \pE_0 x_i x_j x_k) - \Paren{\frac 1 {(\omega - 2)(\omega - 3)} - \frac 1 {(\omega'-2 )(\omega'-3)}} \sum_{k \neq i,j} \pE_0 x_i x_j x_k\mper
  \end{align*}

  With probability $1 - O(n^{-20})$ when $G \sim G(n,1/2)$ by \pref{lem:deg4-scalar}, we get that $1/(\omega' -2)(\omega'-3) = O(1/\omega^2)$.
  By the same,
  \[
    \Abs{\frac 1 {(\omega - 2)(\omega - 3)} - \frac 1 {(\omega'-2 )(\omega'-3)}} \leq \tO(\gamma /\omega n^{5/2})\mper
  \]
  Together with the bound from earlier in this proof on $|\pE x_i x_j x_k - \pE_0 x_i x_j x_k|$ and the trivial bound $|\sum_{k \neq i,j} \pE_0 x_i x_j x_k| \leq \tO(\omega^3/n^2)$, we obtain
  \begin{align*}
    |\pE x_i x_j - \pE_0 x_i x_j| & \leq \tO(\gamma \omega^2/ n^3) + \tO(\gamma \omega^2 / n^{9/2})
  \end{align*}
  which is $o(\omega^2 / n^2)$ for $\omega = o(\sqrt n)$ and $\gamma = o(\omega)$.
\end{proof}

\begin{proof}[Proof of \pref{lem:N4prime}]
  We first observe that an eigenvalue lower bound on $\cN$ implies the same on the (principal) submatrix indexed only by cliques (in this case, edges) in $G$.
  This submatrix is equal to $C_4 \cN' + \Err$, where
  \[
    \Err(I,J) = \begin{cases}
    			C_4 \pE x^I x^J - C_4 \pE_0 x^I x^J & \text{ if $|I \cup J| < 4$}\\
			0 & \text{ otherwise}
		\end{cases}
  \]
  We break $\Err$ into two parts so that $\Err_2 + \Err_3 = \Err$:
  \[
    \Err_2(I,J) =\begin{cases}
    			\Err(I,J) & \text{ if $|I \cup J| = 2$}\\
			0 & \text{ otherwise}
		\end{cases}
		\quad
		\text{and}
		\quad
    \Err_3(I,J) =\begin{cases}
    			\Err(I,J) & \text{ if $|I \cup J| = 3$}\\
			0 & \text{ otherwise}
		\end{cases}
  \]
  Note that $\Err_3$ consists of the off-diagonal nonzero entries of $\Err$, while $\Err_2$ contains the diagonal entries of $\Err$.
  We start by showing a bound on $\|\Err_3\|$.
  \begin{align*}
    \|\Err_3\| & \leq \max_{I \in {n \choose 2}} \sum_{J \neq I} |\Err_3(I,J)|\\
    & = \max_{\{i_1,i_2\} \in {n \choose 2}} \sum_{k \notin \{i_1,i_2\}} |\Err_3(I,\{k,i_1\})| + |\Err_3(I,\{k,i_2\})| \quad \text{(definition of $\Err_3$)}\\
    & \leq O(n) \cdot \max_{i,j,k \text{ all not equal}} C_4 | \pE x_i x_j x_k - \pE_0 x_i x_j x_k |\\
    & \leq O(n) \cdot C_4 \cdot  O(\gamma \log(n) \omega^4 / n^4) \quad \text{w.p. $1 - O(n^{-20})$ by \pref{lem:nondisjoint-entry}}\\
    & \leq O(\gamma \log(n) \omega^4 n) \quad \text{w.p. $1 - O(n^{-20})$ by $C_4 \approx n^4$, see \cite[Theorem 10.3]{MPW15}}
  \end{align*}

  Next we bound $\Err_2$.
  Since it is diagonal, it is enough to give an entrywise bound.
  \begin{align*}
    \|\Err_2\| & \leq \max_{I \in {n \choose 2}} \Err_2(I,I)\\
    & = \max_{i\neq j} C_4(\pE x_i x_j - \pE_0 x_i x_j) \quad \text{by definition of $\Err_2$}\\
    & \leq C_4 \cdot o(\omega^2 / n^2) \quad \text{w.p. $1 - O(n^{-20})$ by \pref{lem:nondisjoint-entry}}\\
    & \leq o(\omega^2 n^2) \quad \text{w.p. $1 - O(n^{-20})$ by $C_4 \approx n^4$, see \cite[Theorem 10.3]{MPW15}}\mper
  \end{align*}

  Fix $\gamma,c \in \R$.
  Suppose $\cN \succeq (\omega^2 n^2/c) \cdot I$.
  Then for $\cN + \Err$ to be PSD it is enough to have $\|\Err\| \leq \omega^2 n^2 /c$.
  There is by the above bounds a universal constant $C$ so that it is enough to have
  $C c \gamma \log(n) \omega^4 n \leq \omega^2 n^2$, or rearranging, $\omega \leq \sqrt{n/Cc \log(n) \gamma}$.
\end{proof}

\subsection{Eigenvalue Lower Bound for the Correction}

The following is the main claim for this section.

\begin{lemma}
  \label{lem:cor-eval-lb}
  Let $G \sim G(n,1/2)$.
  Let 
  \[
    r_s(j) = \begin{cases}
    1 & \text{ if $s \sim j$}\\
    -1 & \text{ if $s \not \sim j$}\\
    0 & \text{ if $s = j$}
    \end{cases}\mper
  \]
  Let $r_s^{\tensor 2} \in \R^{\nchoose{2}}$ be the vector with entries $r_s^{\tensor 2}(\{i,j\}) = r_s(i) r_s(j)$.
  Let $V = \Span\{r_s^{\tensor 2} \}_{ s\in [n]}$.
  Let $\Pi_V$ be the projector to $V$.
  With probability at least $1 - O(n^{-10})$ over the sample of $G$,
  \[
    \sum_s (r_s^{\tensor 2})(r_s^{\tensor 2})^\top \succeq \Omega(n^2) \Pi_V\mper
  \]
\end{lemma}

We will need the following graph theoretic machinery for the moment method bound.
\begin{definition}
  Let $G$ be a graph on $[n]$.
  A diamond ribbon $R$ of length $2\ell$ is a graph on vertices $s_1,\ldots,s_\ell, t_1,\ldots,t_\ell, u_1,\ldots,u_{2\ell}, v_1,\ldots,v_{2\ell}$.
  It has edges
  \[
    (s_i, u_{2i}), (s_i, v_{2i}), (u_{2i}, t_i), (v_{2i},t_i), (t_i, u_{2i + 1}), (t_i, v_{2i + 1}), (u_{2i+1}, s_{i+1}),(v_{2i+1},s_{i+1})
  \]
  where addition is modulo $2\ell$.

  A labeled diamond ribbon $(R,F)$ is a a diamond ribbon of length $2\ell$ together with a labeling $F : R \rightarrow G$ of vertices in $R$ with vertices from $G$.
  We insist that for $(x,y) \in R$ an edge that $F(x) \neq F(y)$.

  The labeled diamond ribbon $(R,F)$ is contributing if no element of the multiset $\{(F(x),F(y)) \text{ such that } (x,y) \in R\}$ occurs with odd multiplicity.
  It is \emph{disjoint} if the sets
  \[
    \{F(s_i)\}, \{F(t_i)\}, \{F(u_i)\}_{i \text{ odd}}, \{F(v_i)\}_{i\text{ even}}, \{F(v_i)\}_{i \text{ odd}}, \{ F(v_i) \}_{i \text{ even}}
  \]
  are disjoint.
\end{definition}

\begin{lemma}
  Let $(R,F)$ be a contributing disjoint labeled diamond ribbon of length $2\ell$.
  Then it contains at most $3\ell + O(1)$ distinct labels.
\end{lemma}
\begin{proof}
  By our disjointness assumption, every element of the multiset $\{F(s_i), F(t_i) \}$ must occur with multiplicity at least two and similarly for $\{F(u_i)\}$ and $\{F(v_i) \}$.
\end{proof}

\begin{proof}[Proof of \pref{lem:cor-eval-lb}]
  Note that the matrix $R = \sum_s (r_s^{\tensor 2})(r_s^{\tensor 2})^\top$ has row and column spaces both $V$.
  Note also that it factors as $SS^\top$, where $S$ is the $\nchoose{2} \times n$ matrix whose columns are the vectors $r_s^{\tensor 2}$.
  Thus, it will be enough to show that
  \[
    S^\top S \succeq \Omega(n^2) I \quad \text{w.p. $1 - O(n^{-10})$}
  \]
  where here $I$ is the $n \times n$ identity matrix.

  For this, consider the matrix $S^\top S$ indexed by vertices $s,t \in [n]$.
  It has entries
  \[
    S^\top S (s,t) = \iprod{r_s^{\tensor 2}, r_t^{\tensor 2}}
  \] 
  and in particular,
  \[
    S^\top S (s,s) = \iprod{r_s^{\tensor 2}, r_s^{\tensor 2}} = \nchoose{2}
  \]
  So, zeroing this matrix on the diagonal, it is enough to prove that
  \[
    \Norm{S^\top S - \nchoose{2} I} \leq o(n^2) \quad \text{w.p. $1 - O(n^{-10})$}\mper
  \]
  
  Let $H \seteq S^\top S - \nchoose{2} I$.
  Let $H_{i,j}$ for $i \neq j$ be given by
  \[
    H_{i,j} (s,t) = \begin{cases}
    r_s(i) r_s(j) r_t(i) r_t(j) & \text{ if $s \neq t$}\\
    0 & \text{ otherwise}
    \end{cases}\mper
  \]
  Then $H = \sum_{i \neq j} H_{i,j}$.
  Note that $H_{i,j}(s,t) = 0$ if $i \in \{s,t \}$ or $j \in \{s,t \}$.
  Thus the obvious generalization of \pref{lem:partitioning} to the two-parameter family $H_{i,j}$ applies.
  This gives us a family of matrices $H^1 ,\ldots,H^r$ for some $r = O(\log n)$ and a corresponding family of partitions $(S_1^1,S_2^2,S_3^1,S_4^1,S_5^1,S_6^1),\ldots,(S_1^r,\ldots,S_6^r)$ of $[n]$.

  These will be such that $\sum_{i = 1}^r H^i = H$, and
  \[
    H^i(s,t) = \sum_{(j,k) \in {S_i}'} r_s(j) r_s(k) r_t(j) r_t(k)
  \]
  where ${S^i}' \subseteq S_5^i \times S_6^i$ is the subset giving indices so that the corresponding summand has not occurred in any $i'' < i$.
  We will bound $\E \Tr (H^i)^{2\ell}$ for some $\ell$ to be chosen later.
  Every term in the expansion of this quantity corresponds to a disjoint labeled diamond ribbon of length $2\ell$, and the number of nonzero terms is at most the number of contributing disjoint labeled diamond ribbons.
  So $\E \Tr(H^i)^{2\ell} \leq n^{3\ell + O(1)}$.
  The rest follows by standard manipulations.
\end{proof}
\subsection{Proofs of Remaining Lemmas}
\begin{proof}[Proof of \pref{lem:deg4-ER}]
  By \pref{lem:Eeigenvals},
  \[
    E \succeq \Omega(\omega^4n^2) \cdot \Pi_{W_0} + \Omega(\omega^3n^2) \cdot \Pi_{W_1} + \Omega(\omega^2n^2) \cdot \Pi_{W_2}\mper
  \]

  Let $W = W_0 \oplus W_1$.
  By \pref{lem:cor-eval-lb} and \cite[Theorem 10.3]{MPW15} (saying that $C_4 \approx n^4$),
  with probability $1 - O(n^{-10})$ we get that $\tilde \cR_0 \succeq \Omega(\gamma \omega^5 n) \Pi_{\Span \{ r_s^{\tensor 2}\} }$.
  We make the observation that $\tilde \cR_0 = (\Pi_W + \Pi_{W_{1.5}})\tilde \cR_0(\Pi_W + \Pi_{W_{1.5}})$ and that $\Pi_{W_{1.5}} \Pi_{\Span \{ r_s^{\tensor 2}\} } \Pi_{W_{1.5}} = \Pi_{W_{1.5}}$.
  So we just need to handle the term $\Pi_W \tilde \cR_0 \Pi_{W_{1.5}} + \Pi_{W_{1.5}} \tilde \cR_0 \Pi_W$.

  Together, \pref{lem:correction-disjoint} and \pref{lem:deg4-scalar} imply that $\|\tilde \cR_0 \| \leq O(\gamma \omega^5 n \log(n)^2)$ with probability $1 - O(n^{-10})$.
  Thus to ensure that $\|\tilde \cR_0 \| \leq o(\sqrt{\omega^3 n^2} \cdot \gamma \omega^5 n)$ it is enough to choose $\omega \leq \sqrt{\gamma n}/\log(n)^2$.
  This is enough to apply \pref{lem:matrix-cs} and conclude the proof.
\end{proof}

\begin{proof}[Proof of \pref{lem:deg4-delta}]
  Note that for $I,J$ disjoint we have $\Delta(I,J) = 0$.
  We bound the maximal sum across any row of $\Delta$.
  With probability $1 - O(n^{-10})$ every off-diagonal entry off $\Delta$ is at most $O(\omega^3 \sqrt n \log(n)^2)$ in absolute value \cite[Theorem 10.1]{MPW15}.
  For each $I \in \nchoose{2}$, we then get $\sum_{J \neq I} |\Delta(I,J) | \leq O(\omega^3 n^{3/2} \log(n)^2)$.
  At the same time,
  The diagonal entries are each at most $O(\omega^2 n^{3/2}\log(n)^2)$ with similar probability, again by \cite[Theorem 10.1]{MPW15}.
\end{proof}

\begin{proof}[Proof of \pref{lem:deg4-corr-dev}]
  $\cR$ and $\tilde \cR_0$ differ in two respects.
  We first bound the spectral norm of the part of $\tilde \cR_0$ on non-disjoint entries.
  Let
  \[
    \tilde \cR_0^{(3)}(I,J) = \begin{cases} \tilde \cR_0(I,J) & \text{ if $|I \cap J| = 3$} \\ 0 & \text{ otherwise} \end{cases}\mper
  \]
  It follows from \pref{lem:deg4-scalar} and a row-sum bound that $\|\tilde \cR_0^{(3)}\| \leq O(\gamma \omega^5 \sqrt n \log(n)^2)$ with probability $ 1 - O(n^{-10})$.
  A similar analysis holds for the analogous matrix $\tilde \cR_0^{(2)}$.

  Let $\cR_0$ be given by
  \[
  \cR_0(I,J) =
  \begin{cases}
    \frac 1 {16} \sum_{s \in [n]} \gamma (\tfrac {\omega} n)^5 C_4 \cdot \Pi_{a \in (I \cup J) \setminus (I \cap J)} r_s(a) & \text{ if $|I \cap J| = 4$}\\
    0 & \text{ otherwise}
  \end{cases}\mper
  \]
  Now it is enough to bound $\| \cR_0 - \cR \|$.
  This is the deviation introduced by zeroing non-clique entries.
  Note that each entry $I,J$ of $\cR$ can be decomposed as a function of the underlying graph edge variables:
  \[
    \cR(I,J) = \cR_0 + \gamma C_4 (\tfrac{\omega} n)^5 \frac 1 {16} \sum_{\text{nonempty } S \subseteq \sE_{ext}(I,J)} \sum_s \prod_{(u,v) \in S} g_{u,v} \prod_{u \in I \cup J} r_s(u)\mcom
  \]
  where we recall that for an edge $(u,v)$, the variable $g_{(u,v)}$ is the $\pm 1$ indicator for that edge.
  Each of the entries in the sum over nonempty $S$ above corresponds to a matrix of the form bounded in \pref{lem:correction-disjoint}, where we conclude that
  each has spectral norm at most $O(n^{3/2} \log(n)^2)$ with probability $1 - O(n^{-10})$.
  We conclude (also using $C_4 \approx n^4$, see \cite[Theorem 10.3]{MPW15}) that $\|\cR_0 - \cR\| \leq O(\gamma \omega^5 \sqrt n \log(n)^2)$ as desired.
\end{proof}


\section{Concentration of $\deg_G(I)$}\label{sec:degreevariancecalculations}
In this section, we prove the following large-deviation bounds on the number of $x$-cliques a random $G(n,1/2)$ graph contains and on $\deg_G(I)$.
Similar results (which are likely sufficient for our needs) appear in the literature; see \cite{Rucinski88, Vu01, JLR11book} for instance.
We provide these proofs for completeness.
A coarser concentration result for $\deg_G(I)$ appears in \cite{MPW15}.

\begin{definition}
For a graph $G$, define $N_x(G)$ to be the number of $x$-cliques in $G$.
\end{definition}
Unless otherwise specified, in this section $G \sim G(n,1/2)$.

This first theorem gives the large deviation bound for the number of cliques of size $x$ in $G$.
\begin{theorem}\label{thm:cliqueboundtheoremone}
For all $\epsilon \in (0,1)$, for all $x$, if $n > {x^2}(2e - e\ln{\epsilon})(2e + 2 - e\ln{\epsilon})$ then
$$\Pr\left[\left|N_x(G) - 2^{-{x \choose 2}}{n \choose x}\right| > 
e(2 - \ln{\epsilon})\frac{x^2}{x!}n^{x-1}\right] < \epsilon.$$
(Note that $2^{-{x \choose 2}} {n \choose x} = \E N_x(G)$.)
\end{theorem}

We also want a large deviations inequality for the number of cliques of size $2d$ that a clique of size $d' < 2d$ participates in in $G$.  Moreover, to carry out the eigenspace splitting arguments needed for \pref{lem:Delta-bound}, we want to know the dependence of this deviation on the number of vertices adjacent to every vertex in the $d'$-clique.
The following theorem serves both these purposes.
\begin{definition}
Given any $I \subseteq [n]$, let $A_I$ be the set of all vertices not in $I$ which are adjacent to all 
vertices in $I$.
\end{definition}
\begin{theorem}\label{thm:cliqueboundtheoremtwo}
There is a universal constant $C$ so that for any $I \subseteq [n]$ of size at most $2d$, if $I$ is a clique in $G$ then for any $\epsilon \in (0,1)$, if $n \geq 100{d^2}2^{2d}(3 - \ln{\epsilon})^2$ then
$$\Pr\left[\,\left|\deg_G(I) - 2^{{{|I|} \choose 2} - {{2d} \choose 2}}{{n-|I|} \choose {{2d} - |I|}}\right| > 
C(3 - \ln{\epsilon})n^{d - |I| - \frac{1}{2}}\,\right] < \epsilon$$
More precisely, if $|I| < d$ then
\begin{align*}
\Pr &\left[\,\left|\deg_G(I) - 2^{{{|I|} \choose 2} - {{2d} \choose 2}}{{n-|I|} \choose {{2d} - |I|}} - \left(|A_I| - \frac{n-|I|}{2^{|I|}}\right)\frac{2^{{{|I|+1} \choose 2} - {{2d} \choose 2}}(n-|I|)^{2d-|I|-1}}{(2d-|I|-1)!}\right| > \right. \\
&\;\;\;\;\left. C 2^{2d}(3 - \ln{\epsilon})^2n^{2d-|I|-1}\,\right] < \epsilon.
\end{align*}
\end{theorem}
The key lemma in proving \pref{thm:cliqueboundtheoremone} is the following, which bounds how often subsets of $G$ of size $x$ share (potential) edges.
\begin{lemma}\label{lem:cliqueboundlemma}
If $x \geq 2$ and $n \geq {x^2}q(q+2)$ then there are at most 
$2n^{xq - q}\left(\frac{{x^2}}{x!}\right)^{q}$ multi-sets $S = \{V_1, \cdots, V_q\}$ of subsets $V_i \subseteq [n]$ of size $|V_i| = x$ such that
that for all $j$ there exists an $i \neq j$ such that $|V_i \cap V_j| \geq 2$.
\end{lemma}
Using \pref{lem:cliqueboundlemma}, we can bound the deviation of the number of $x$-cliques in $G$ from its expected value; we carry this out now.
\begin{definition}
Define $X = \sum_{V: V \subseteq G, |V| = x}{\left(1_{V} - 2^{-{x \choose 2}}\right)}$ where $1_V = 1$ if $V$ is a clique in $G$ and 
$0$ otherwise.
\end{definition}
\begin{proposition} \label{prop:degG-1}
$X = N_x(G) - 2^{-{x \choose 2}}{n \choose x}$
\end{proposition}
\begin{proof}
By observation.
\end{proof}
\begin{corollary}\label{cliqueboundcorollaryone}
If $x \geq 2$ and $n \geq {x^2}q(q+2)$ then $E[X^q] < 2q!\left(\frac{{x^2}}{x!}n^{x-1}\right)^q$.
\end{corollary}
\begin{proof}
By \pref{prop:degG-1},
$\E[X^q] = \sum_{V_1, \cdots, V_q: \atop \forall i, V_i \subseteq V(G), |V_i| = x}
{\E\left[\prod_{i=1}^{q}{\left(1_{V_i} - 2^{-{x \choose 2}}\right)}\right]}$. Note that all terms of this sum have value less than $1$. Furthermore, for all nonzero terms in this sum, 
for all $j$ there must be an $i$ such that $|V_i \cap V_j| \geq 2$, since the sets $V_i$ and $V_j$ must share a potential edge in order for $1_{V_i}$ and $1_{V_j}$ not to be independent.
Thus, this sum 
is at most the number of ordered multi-sets of $q$ $x$-cliques $\{V_1, \cdots, V_q\}$ where for all $j$ there 
is an $i$ such that $|V_i \cap V_j| \geq 2$. In turn, this is at most $q!$ times the number 
of unordered multi-sets of such $q$ $x$-cliques. By \pref{lem:cliqueboundlemma}, this is at most 
$2q!\left(\frac{{x^2}}{x!}n^{x-1}\right)^q$, as needed.
\end{proof}
We are now ready to prove \pref{thm:cliqueboundtheoremone}.
\ignore{ which we restate here for convenience.
\vskip.1in
\noindent
{\bf \pref{thm:cliqueboundtheoremone}.}
{\it
For all $\epsilon \in (0,1)$, for all $x$, if $n > {x^2}(2e - e\ln{\epsilon})(2e + 2 - e\ln{\epsilon})$ then 
$$P\left[\left|N_x(G) - 2^{-{x \choose 2}}{n \choose x}\right| > 
e(2 - \ln{\epsilon})\frac{x^2}{x!}n^{x-1}\right] < \epsilon$$
}}
\begin{proof}[Proof of \pref{thm:cliqueboundtheoremone}]
The result is trivial for $x = 0$ and $x = 1$ so we may assume that $x \geq 2$. 
Using Corollary \ref{cliqueboundcorollaryone} and Markov's inequality, 
$$\epsilon > \Pr\left[X^q > \frac{E[X^q]}{\epsilon}\right] \geq 
\Pr\left[X^q > \frac{2q!\left(\frac{{x^2}}{x!}n^{x-1}\right)^q}{\epsilon}\right] 
= \Pr\left[|X| > \sqrt[q]{\frac{2q!}{\epsilon}}\frac{x^2}{x!}n^{x-1}\right]$$

Thus, we just need to give an upper bound on $\min_{\{\text{positive even q}\}}\{\sqrt[q]{\frac{2q!}{\epsilon}}\}$. 
For all positive even $q$, $2q! \leq q^q$ so this expression is upper bounded by $\frac{q}{\sqrt[q]{\epsilon}}$. 
We now try to minimize $\frac{q}{\sqrt[q]{\epsilon}}$ over all positive even $q$. Taking the derivative of this expression 
with respect to $q$ yields $\frac{1}{\sqrt[q]{\epsilon}} + \frac{\ln{\epsilon}}{q\sqrt[q]{\epsilon}}$. Setting this 
to $0$ yields $q = -\ln{\epsilon}$. However, we require $q$ to be even so we take $q$ to be the smallest positive even 
integer which is greater than $-\ln{\epsilon}$. Now $q < 2 - \ln{\epsilon}$ and 
$\sqrt[q]{\epsilon} \geq \sqrt[(-\ln{\epsilon})]{\epsilon} = \left(e^{\ln{\epsilon}}\right)^{\frac{1}{-\ln{\epsilon}}} = \frac{1}{e}$. 
Putting everything together, for this $q$, $\sqrt[q]{\frac{2q!}{\epsilon}} \leq \frac{q}{\sqrt[q]{\epsilon}} < 2e - e\ln{\epsilon}$. 
Plugging this in gives 
$$\epsilon > P\left[|X| > \sqrt[q]{\frac{2q!}{\epsilon}}\frac{x^2}{x!}n^{x-1}\right] \geq 
P\left[|X| > e(2-\ln{\epsilon})\frac{x^2}{x!}n^{x-1}\right]$$
All that is left is to check that $n \geq {x^2}q(q+2)$ for this $q$ to make sure that our application of 
Corollary \ref{cliqueboundcorollaryone} was valid. Since $n > {x^2}(2e - e\ln{\epsilon})(2e + 2 - e\ln{\epsilon})$, this holds, 
as needed.
\end{proof}
Now that we have proven \pref{thm:cliqueboundtheoremone}, we will derive \pref{thm:cliqueboundtheoremtwo} 
from \pref{thm:cliqueboundtheoremone}. The idea is that conditioned on $I$ being a clique, by 
\pref{thm:cliqueboundtheoremone}, $\deg_G(I)$ is primarily determined by $|A_I|$, which can be easily shown to be tightly concentrated around its expected value. 
We start with the following lemma
\begin{lemma}\label{lem:ndvavassumedlemma}
If $n > d$ then for any $I \subseteq [n]$ of size less than $d$, if we first determine all of the edges 
incident to elements of $I$ (which determines $A_I$) then if $I$ is a clique, when we look at the remainder of the graph, for any $\epsilon_1 \in (0,1)$, 
\begin{align*}
P\Big(&\left|\deg_G(I) - 2^{{{|I|} \choose 2} - {d \choose 2}}{{n-|I|} \choose {d - |I|}} - 
\left(|A_I| - \frac{n-|I|}{2^{|I|}}\right)\frac{2^{{{|I|+1} \choose 2} - {d \choose 2}}(n-|I|)^{d-|I|-1}}{(d-|I|-1)!}\right| > \\
&10(2 - \ln{\epsilon_1})n^{d-|I|-1} + (d - |I|)^2\left(\frac{2^{|I|}|A_I|}{n-|I|} - 1\right)^2
\frac{2^{{{|I|} \choose 2} - {d \choose 2}}(n-|I|)^{d-|I|}}{(d - |I|)!}\Big) < \epsilon_1
\end{align*} 
so long as the following conditions hold:
\begin{enumerate}
\item $(d - |I|)\left|\frac{2^{|I|}|A_I|}{n - |I|} - 1\right| \leq 1$
\item $|A_I| > {d^2}(2e - e\ln{\epsilon_1})(2e + 2 - e\ln{\epsilon_1})$
\end{enumerate}
\end{lemma}


To prove \pref{lem:ndvavassumedlemma} we require the following results; proofs of the more elementary ones are deferred to \pref{sec:deg-lems}.
\begin{proposition}\label{prop:factorialapproxprop}
For all nonnegative integers $n$ and $k$ where $k < n$, $0 \leq \frac{n^k}{k!} - {n \choose k} \leq \frac{k^2}{2n}\frac{n^k}{k!}$
\end{proposition}
\begin{proof}[Proof of \pref{prop:factorialapproxprop}]
Note that $n^k \geq \prod_{j=0}^{k-1}{(n-j)} \geq n^k\prod_{j=0}^{k-1}{(1 - \frac{j}{n})} \geq 
n^k(1 - \sum_{j=0}^{k-1}{\frac{j}{n}}) \geq n^k(1 - \frac{k^2}{2n})$

This implies that $0 \leq n^k - \prod_{j=0}^{k-1}{(n-j)} \leq \frac{k^2}{2n}{n^k}$ and dividing everything by 
$k!$ gives the claimed result.
\end{proof}

\begin{lemma}\label{lem:ndvpartonelemma}
If $n > d \geq |I|$ then $\left|2^{{{|I|} \choose 2} - {d \choose 2}}\frac{(n - |I|)^{d - |I|}}{(d-|I|)!} - 
2^{{{|I|} \choose 2} - {d \choose 2}}{{n-|I|} \choose {d - |I|}}\right| \leq 2^{{{|I|} \choose 2} - {d \choose 2}}n^{d - |I| - 1}$
\end{lemma}
\begin{proof}[Proof of \pref{lem:ndvpartonelemma}]
Applying \pref{prop:factorialapproxprop} on $n - |I|$ and $d - |I|$ gives 
$$\left|\frac{(n-|I|)^{d - |I|}}{(d - |I|)!} - {{n - |I|} \choose {d - |I|}}\right| \leq 
\frac{(d - |I|)^2}{2(n - |I|)}\frac{(n - |I|)^{d - |I|}}{(d - |I|)!} \leq (n - |I|)^{d - |I| - 1} \leq n^{d - |I| - 1}$$
Multiplying this equation by $2^{{{|I|} \choose 2} - {d \choose 2}}$ gives the claimed result.
\end{proof}

\begin{proposition}\label{prop:exponentprop}
For all nonnegative integers $x$ and $d$ where $x \leq d$, $x(d-x) + {x \choose 2} - {d \choose 2} = -{{d-x} \choose 2}$. 
\end{proposition}
\begin{proposition}\label{prop:degG-2}
For any nonnegative integer $k$ and any $x$ such that $|kx| \leq 1$, $\left|(1+x)^k - (1 + kx)\right| \leq {k^2}{x^2}$
\end{proposition}

Eventually in the course of proving \pref{lem:ndvavassumedlemma} we will break
$$\deg_G(I) - 2^{{{|I|} \choose 2} - {d \choose 2}}{{n-|I|} \choose {d - |I|}} - 
\left(|A_I| - \frac{n-|I|}{2^{|I|}}\right)\frac{2^{{{|I|+1} \choose 2} - {d \choose 2}}(n-|I|)^{d-|I|-1}}{(d-|I|-1)!}$$ 
into pieces.
The following lemmas offer the necessary bounds on each piece.

\begin{lemma}\label{lem:ndvparttwolemma}
If $(d - |I|)\left|\frac{2^{|I|}|A_I|}{n - |I|} - 1\right| \leq 1$ then
\begin{align*}&\left|2^{-{{d - |I|} \choose 2}}\frac{|A_I|^{d - |I|}}{(d - |I|)!} - 
\frac{2^{{{|I|} \choose 2} - {d \choose 2}}(n-|I|)^{d-|I|}}{(d - |I|)!} - 
\left(|A_I| - \frac{n-|I|}{2^{|I|}}\right)\frac{2^{{{|I|+1} \choose 2} - {d \choose 2}}(n-|I|)^{d-|I|-1}}{(d-|I|-1)!}\right| \\
&\leq (d - |I|)^2\left(\frac{2^{|I|}|A_I|}{n-|I|} - 1\right)^2
\frac{2^{{{|I|} \choose 2} - {d \choose 2}}(n-|I|)^{d-|I|}}{(d - |I|)!}\end{align*}
\end{lemma}
\begin{proof}[Proof of \pref{lem:ndvparttwolemma}]
Applying \pref{prop:degG-2} with $x = \frac{2^{|I|}|A_I|}{n-|I|} - 1$ and $k = (d - |I|)$, since 
$(d - |I|)\left|\frac{2^{|I|}|A_I|}{n - |I|} - 1\right| \leq 1$, 
$$\left|\frac{2^{|I|(d - |I|)}|A_I|^{d - |I|}}{(n - |I|)^{d - |I|}} - 1 - (d - |I|)
\left(\frac{2^{|I|}|A_I|}{n-|I|} - 1\right)\right| \leq 
(d - |I|)^2\left(\frac{2^{|I|}|A_I|}{n-|I|} - 1\right)^2$$
Multiplying this equation by $\frac{2^{{{|I|} \choose 2} - {d \choose 2}}(n-|I|)^{d-|I|}}{(d - |I|)!}$ and using 
\pref{prop:exponentprop} with $x = |I|$ gives
\begin{align*}&\left|2^{-{{d - |I|} \choose 2}}\frac{|A_I|^{d - |I|}}{(d - |I|)!} - 
\frac{2^{{{|I|} \choose 2} - {d \choose 2}}(n-|I|)^{d-|I|}}{(d - |I|)!} - 
(d - |I|)\left(\frac{2^{|I|}|A_I|}{n-|I|} - 1\right)\frac{2^{{{|I|} \choose 2} - {d \choose 2}}(n-|I|)^{d-|I|}}{(d - |I|)!}\right| \\
&= \left|2^{-{{d - |I|} \choose 2}}\frac{|A_I|^{d - |I|}}{(d - |I|)!} - 
\frac{2^{{{|I|} \choose 2} - {d \choose 2}}(n - |I|)^{d-|I|}}{(d - |I|)!} - 
\left(A_I - \frac{n-|I|}{2^{|I|}}\right)\frac{2^{{{|I|+1} \choose 2} - {d \choose 2}}(n - |I|)^{d-|I|-1}}{(d-|I|-1)!}\right| \\
&\leq (d - |I|)^2\left(\frac{2^{|I|}A_I}{n-|I|} - 1\right)^2\frac{2^{{{|I|} \choose 2} - {d \choose 2}}(n-|I|)^{d-|I|}}{(d - |I|)!}\end{align*}
\end{proof}

\begin{lemma}\label{lem:ndvpartthreelemma}
If $|A_I| > d - |I|$ then $\left|2^{-{{d - |I|} \choose 2}}{{|A_I|} \choose {d-|I|}} - 2^{-{{d - |I|} \choose 2}}
\frac{|A_I|^{d - |I|}}{(d - |I|)!}\right| \leq 2^{-{{d - |I|} \choose 2}}|A_I|^{d-|I|-1}$
\end{lemma}
\begin{proof}[Proof of \pref{lem:ndvpartthreelemma}]
Applying \pref{prop:factorialapproxprop} on $|A_I|$ and $d - |I|$ gives 
$$\left|\frac{|A_I|^{d-|I|}}{(d-|I|)!} - {{|A_I|} \choose {d-|I|}}\right| \leq \frac{(d-|I|)^2}{2|A_I|}\frac{|A_I|^{d-|I|}}{(d-|I|)!} 
\leq |A_I|^{d-|I|-1}$$
Multiplying this equation by $2^{-{{d - |I|} \choose 2}}$ gives the claimed result.
\end{proof}

\begin{lemma}\label{lem:ndvpartfourlemma}
For all $\epsilon_1 \in (0,1)$, if $|A_I| > {d^2}(2e - e\ln{\epsilon_1})(2e + 2 - e\ln{\epsilon_1})$ then 
$$\Pr\left[\left|\deg_G(I) - 2^{-{{d-|I|} \choose 2}}{|A_I| \choose {d-|I|}}\right| > 
e(2 - \ln{\epsilon_1})\frac{(d-|I|)^2}{(d-|I|)!}|A_I|^{d-|I|-1}\right] < \epsilon_1$$
\end{lemma}
\begin{proof}[Proof of \pref{lem:ndvpartfourlemma}]
This lemma follows immediately from applying \pref{thm:cliqueboundtheoremone} on the random graph $G$ restricted to the vertices $A_I$.
\end{proof}

\begin{proof}[Proof of \pref{lem:ndvavassumedlemma}]
We break $$\deg_G(I) - 2^{{{|I|} \choose 2} - {d \choose 2}}{{n-|I|} \choose {d - |I|}} - 
\left(|A_I| - \frac{n-|I|}{2^{|I|}}\right)\frac{2^{{{|I|+1} \choose 2} - {d \choose 2}}(n-|I|)^{d-|I|-1}}{(d-|I|-1)!}$$ 
into four parts and analyze each one separately.
\begin{enumerate}
\item $2^{{{|I|} \choose 2} - {d \choose 2}}\frac{(n - |I|)^{d - |I|}}{(d-|I|)!} - 2^{{{|I|} \choose 2} - 
{d \choose 2}}{{n-|I|} \choose {d - |I|}}$
\item $2^{-{{d - |I|} \choose 2}}\frac{|A_I|^{d - |I|}}{(d - |I|)!} - 
2^{{{|I|} \choose 2} - {d \choose 2}}\frac{(n - |I|)^{d-|I|}}{(d - |I|)!} - 
\left(|A_I| - \frac{n-|I|}{2^{|I|}}\right)\frac{2^{{{|I|+1} \choose 2} - {d \choose 2}}(n - |I|)^{d-|I|-1}}{(d-|I|-1)!}$
\item $2^{-{{d - |I|} \choose 2}}{{|A_I|} \choose {d-|I|}} - 2^{-{{d - |I|} \choose 2}}\frac{|A_I|^{d - |I|}}{(d - |I|)!}$
\item $\deg_G(I) - 2^{-{{d - |I|} \choose 2}}{{|A_I|} \choose {d-|I|}}$
\end{enumerate}
Combining \pref{lem:ndvpartonelemma}, \pref{lem:ndvparttwolemma}, \pref{lem:ndvpartthreelemma}, and \pref{lem:ndvpartfourlemma}, we have that under 
the given conditions, 
\begin{align*}
\Pr\Big(&\left|\deg_G(I) - 2^{{{|I|} \choose 2} - {d \choose 2}}{{n-|I|} \choose {d - |I|}} - 
\left(|A_I| - \frac{n-|I|}{2^{|I|}}\right)\frac{2^{{{|I|+1} \choose 2} - {d \choose 2}}(n-|I|)^{d-|I|-1}}{(d-|I|-1)!}\right| > \\
&2^{{{|I|} \choose 2} - {d \choose 2}}n^{d - |I| - 1} + (d - |I|)^2\left(\frac{2^{|I|}|A_I|}{n-|I|} - 1\right)^2
\frac{2^{{{|I|} \choose 2} - {d \choose 2}}(n-|I|)^{d-|I|}}{(d - |I|)!} + \\
&2^{-{{d - |I|} \choose 2}}|A_I|^{d-|I|-1} + e(2 - \ln{\epsilon_1})\frac{(d-|I|)^2}{(d-|I|)!}|A_I|^{d-|I|-1}\Big) < \epsilon_1
\end{align*}
The result now reduces to showing the following equation 
$$2^{{{|I|} \choose 2} - {d \choose 2}}n^{d - |I| - 1} + 
2^{-{{d - |I|} \choose 2}}|A_I|^{d-|I|-1} + e(2 - \ln{\epsilon_1})\frac{(d-|I|)^2}{(d-|I|)!}|A_I|^{d-|I|-1} \leq 
10(2 - \ln{\epsilon_1})n^{d-|I|-1}$$
which follows from the facts that $|I| < d$, $|A_I| \leq n$, and $\frac{(d-|I|)^2}{(d-|I|)!} \leq 2$.
\end{proof}
To use \pref{lem:ndvavassumedlemma} to prove \pref{thm:cliqueboundtheoremtwo}, we need probabilistic bounds 
on $|A_I|$.
\begin{lemma}\label{avboundslemma}
There is a universal constant $C$ so that
for all $\epsilon_2 \in (0,1)$, $P\left[||A_I| - 2^{-|I|}(n - |I|)| > C(2 - \ln{\epsilon_2})\sqrt{n}\right] < \epsilon_2$
\end{lemma}
\begin{proof}
The lemma follows from standard concentration of measure.
If we let $x_i$ be $1$ if $i \notin I$ and $i$ is adjancent to all 
vertices in $I$ and $0$ otherwise then $\sum_{i=1}^{n}{x_i} = A_I$.
The expected value of 
$A_I = \sum_{i=1}^{n}{x_i}$ is $2^{-|I|}(n - |I|)$, so by Bernstein's inequality there is $C$ so that for all 
$\epsilon_2 \in (0,1)$, $\Pr\left[|A_I - 2^{-|I|}(n - |I|)| > C(2 - \ln{\epsilon_2})\sqrt{n}\right] < \epsilon_2$. 
\end{proof}

We have all we need now to prove \pref{thm:cliqueboundtheoremtwo}.
\begin{proof}[Proof of \pref{thm:cliqueboundtheoremtwo}]
The result is trivial if $|I| = d$ and follows immediately from \pref{thm:cliqueboundtheoremone} if $|I| = 0$ 
so we may assume that $0 < |I| < d$.

In what follows, $C$ and $C'$ denotes universal constants which may vary from line to line.
Now recall that by \pref{lem:ndvavassumedlemma}, for any $\epsilon_1 \in (0,1)$, 
\begin{align*}
P\Big(&\left|\deg_G(I) - 2^{{{|I|} \choose 2} - {d \choose 2}}{{n-|I|} \choose {d - |I|}} - 
\left(|A_I| - \frac{n-|I|}{2^{|I|}}\right)\frac{2^{{{|I|+1} \choose 2} - {d \choose 2}}(n-|I|)^{d-|I|-1}}{(d-|I|-1)!}\right| > \\
&C (2 - \ln{\epsilon_1})n^{d-|I|-1} + (d - |I|)^2\left(\frac{2^{|I|}|A_I|}{n-|I|} - 1\right)^2
\frac{2^{{{|I|} \choose 2} - {d \choose 2}}(n-|I|)^{d-|I|}}{(d - |I|)!}\Big) < \epsilon_1
\end{align*} 
so long as the following conditions hold:
\begin{enumerate}
\item $(d - |I|)\left|\frac{2^{|I|}|A_I|}{n - |I|} - 1\right| \leq 1$
\item $|A_I| > {d^2}(2e - e\ln{\epsilon_1})(2e + 2 - e\ln{\epsilon_1})$
\end{enumerate}
Taking $\epsilon_1 = \epsilon_2 = \frac{\epsilon}{2}$, plugging Lemma \ref{avboundslemma} into these equations and using the 
union bound, we have that 
\begin{align*}
P\Big(&\left|\deg_G(I) - 2^{{{|I|} \choose 2} - {d \choose 2}}{{n-|I|} \choose {d - |I|}} - 
\left(|A_I| - \frac{n-|I|}{2^{|I|}}\right)\frac{2^{{{|I|+1} \choose 2} - {d \choose 2}}(n-|I|)^{d-|I|-1}}{(d-|I|-1)!}\right| > \\
&C (3 - \ln{\epsilon})n^{d-|I|-1} + (d - |I|)^2\left(\frac{2^{|I|}C'(3 - \ln{\epsilon})\sqrt{n}}{n-|I|}\right)^2
\frac{2^{{{|I|} \choose 2} - {d \choose 2}}(n-|I|)^{d-|I|}}{(d - |I|)!}\Big) < \epsilon
\end{align*} 
so long as the corresponding conditions hold. Assuming these conditions hold for now, since $|I| < \frac{n}{16}$, $|I| < d$, 
$\frac{(d-|I|)^2}{(d - |I|)!} \leq 2$, and $2^{|I|}2^{{|I|} \choose 2} = 2^{{|I|+1} \choose 2}$,
\begin{align*}
(d - |I|)^2\left(\frac{2^{|I|}C(3 - \ln{\epsilon})\sqrt{n}}{n-|I|}\right)^2
\frac{2^{{{|I|} \choose 2} - {d \choose 2}}(n-|I|)^{d-|I|}}{(d - |I|)!} &\leq 
C'\frac{2^{|I|}(3 - \ln{\epsilon})^2{n}}{(n - |I|)}(n-|I|)^{d-|I|-1} \\
&< C \cdot {2^d}(3 - \ln{\epsilon})^2{n}^{d-|I|-1}
\end{align*}
Plugging this in we have that
\begin{align*}
P\Big(&\left|\deg_G(I) - 2^{{{|I|} \choose 2} - {d \choose 2}}{{n-|I|} \choose {d - |I|}} - 
\left(|A_I| - \frac{n-|I|}{2^{|I|}}\right)\frac{2^{{{|I|+1} \choose 2} - {d \choose 2}}(n-|I|)^{d-|I|-1}}{(d-|I|-1)!}\right| > \\
&C 2^{d}(3 - \ln{\epsilon})^2n^{d-|I|-1}\Big) < \epsilon
\end{align*}
as needed. For the first part of Theorem \pref{thm:cliqueboundtheoremtwo}, note that this implies that 
\begin{align*}
P\Big(&\left|\deg_G(I) - 2^{{{|I|} \choose 2} - {d \choose 2}}{{n-|I|} \choose {d - |I|}}\right| > \\
&\left||A_I| - \frac{n-|I|}{2^{|I|}}\right|\frac{2^{{{|I|+1} \choose 2} - {d \choose 2}}(n-|I|)^{d-|I|-1}}{(d-|I|-1)!} 
+ C2^{d}(3 - \ln{\epsilon})^2n^{d-|I|-1}\Big) < \epsilon
\end{align*}
Plugging in Lemma \ref{avboundslemma} and noting that
$$C(3 - \ln{\epsilon})\sqrt{n}\frac{2^{{{|I|+1} \choose 2} - {d \choose 2}}(n-|I|)^{d-|I|-1}}{(d-|I|-1)!} 
< C'(3 - \ln{\epsilon})n^{d - |I| - \frac{1}{2}}$$
we have that 
$$\Pr\Big(\left|\deg_G(I) - 2^{{{|I|} \choose 2} - {d \choose 2}}{{n-|I|} \choose {d - |I|}}\right| > 
C(3 - \ln{\epsilon})n^{d - |I| - \frac{1}{2}} + C'2^{d}(3 - \ln{\epsilon})^2n^{d-|I|-1}\Big) < \epsilon$$
Taking $n \geq C{d^2}2^{2d}(3 - \ln{\epsilon})^2$ and $d \geq 2$, 
$$C2^{d}(3 - \ln{\epsilon})^2{n}^{d-|I|-1} \leq 
C'(3 - \ln{\epsilon})n^{d - |I| - \frac{1}{2}}$$
Plugging this in gives that 
$$\Pr\Big(\left|\deg_G(I) - 2^{{{|I|} \choose 2} - {d \choose 2}}{{n-|I|} \choose {d - |I|}}\right| > 
C(3 - \ln{\epsilon})n^{d - |I| - \frac{1}{2}}\Big) < \epsilon$$
as needed. All that is left is to check the conditions for \pref{lem:ndvavassumedlemma}, which are as follows.
\begin{enumerate}
\item $(d - |I|)\frac{2^{|I|}e(3 - \ln{\epsilon})\sqrt{n}}{n-|I|} \leq 1$
\item $2^{-|I|}(n - |I|) > {d^2}e(3 - \ln{\epsilon})(3e + 2 - e\ln{\epsilon}) + e(3 - \ln{\epsilon})\sqrt{n}$
\end{enumerate}
These conditions are true if $n \geq 4{d^2}2^{2d}(3 - \ln{\epsilon})^2$. To see this, note that since $d > |I| > 0$ and  
$|I| < \frac{n}{16}$
\begin{enumerate}
\item $(d - |I|)2^{|I|}e(3 - \ln{\epsilon})\sqrt{n} \leq d2^{|I|}e(3 - \ln{\epsilon})\sqrt{n} - |I| 
\leq \frac{e}{4}\sqrt{4{d^2}2^{2d}(3 - \ln{\epsilon})^2}\sqrt{n} - |I| \leq n - |I|$
\item $2^{|I|}{d^2}e(3 - \ln{\epsilon})(3e + 2 - e\ln{\epsilon}) < 2^{|I|}{d^2}e^2\frac{3e + 2}{3e}(3 - \ln{\epsilon})^2 < 
10 \cdot {2^{|I|}}{d^2}(3 - \ln{\epsilon})^2 \leq \frac{5}{16}n$
\item $2^{|I|}e(3 - \ln{\epsilon})\sqrt{n} < 4 \cdot {2^{|I|}}(3 - \ln{\epsilon})\left(2d{2^{d}}(3 - \ln{\epsilon})\right) \leq \frac{n}{2}$
\end{enumerate}
Dividing the first statement by $n - |I|$ gives the first condition. Using the second and third statements, 
$$2^{|I|}{d^2}e(3 - \ln{\epsilon})(3e + 2 - e\ln{\epsilon}) + 2^{|I|}e(3 - \ln{\epsilon})\sqrt{n} < \frac{13}{16}n < n - |I|$$
Dividing this by $2^{|I|}$ gives the second condition, as needed.
\end{proof}

\subsection{Proofs of Remaining Lemmas}
\label{sec:deg-lems}
\begin{proof}[Proof of \pref{lem:cliqueboundlemma}]
\begin{definition}
For each multi-set $S = \{V_1, \cdots, V_k\}$ of $k$ $x$-cliques, define the constraint graph 
$H_S$ as follows.
\begin{enumerate}
\item $V(H_S) = \{V_1, \cdots, V_k\}$
\item $E(H_S) = \{(V_i,V_j): |V_i \cap V_j| \geq 2\}$
\end{enumerate}
\end{definition}
Let's first bound the number of $S$ such that $H_S$ is connected.
\begin{lemma}\label{lem:connectedcaselemma}
For any $x,k \geq 2$, there are at most 
$n^{kx-2k+2}k!\frac{2}{x^4}\left(\frac{{x^4}}{2(x!)}\right)^{k}$ multi-sets $S$ of $k$ $x$-cliques 
such that $H_S$ is connected.
\end{lemma}
\begin{proof}
Since $H_S$ is connected, we can order $\{V_1, \cdots, V_k\}$ so that for all $j > 1$ there is an $i < j$ such that 
$(V_i,V_j) \in E(H_S)$. Assuming this is the case, we have at most ${n \choose x}$ choices for $V_1$. For each $j > 1$, 
there are two vertices in $V_j$ which are contained in some $i$ where $i < j$. There are at most $j$ choices for this 
$i$ and then there are ${d \choose 2}$ choices for which two vertices of $V_i$ are contained in $V_j$. There are at 
most ${{n} \choose {x-2}}$ choices for the other $x-2$ vertices of $V_j$ so for each $j > 1$ there are at most 
$j{x \choose 2}{{n} \choose {x-2}}$ choices for $V_j$. Putting everything together, there are at most 
$${n \choose x}{k!}\left({x \choose 2}{{n} \choose {x-2}}\right)^{k-1} \leq 
{n^x}\frac{k!}{x!}\left(\frac{{x^4}n^{x-2}}{2(x!)}\right)^{k-1} \leq 
n^{kx-2k+2}k!\frac{2}{x^4}\left(\frac{{x^4}}{2(x!)}\right)^{k}$$
multi-sets $S$ of $k$ $x$-cliques such that $H_S$ is connected.
\end{proof}
Now consider the number of multi-sets $S$ of $q$ $x$-cliques such that $H_S$ has $t$ connected 
components with sizes $s_1, \cdots, s_t$. Using \pref{lem:connectedcaselemma}, there are at most 
$$\prod_{i=1}^{t}{\left(n^{{s_i}x-2{s_i}+2}{s_i}!\frac{2}{x^4}\left(\frac{{x^4}}{2(x!)}\right)^{{s_i}}\right)} = 
\left(\prod_{i=1}^{t}{{s_i}!}\right){n^{xq - 2q + 2t}\left(\frac{2}{x^4}\right)^{t}\left(\frac{{x^4}}{2(x!)}\right)^{q}}$$
such $S$.

We now total this up over all possible $t,s_1, \cdots, s_t$. For the special case that all connected components 
of $H_S$ have size $2$, there are at most ${n^{xq - q}\left(\frac{{x^2}}{x!}\right)^{q}}$ such $S$. 
We will show that this term contributes more than all of the other terms combined, which implies 
that the total number of $S$ is at most $2{n^{xq - q}\left(\frac{{x^2}}{x!}\right)^{q}}$, as needed.

For a given $t$, $\prod_{i=1}^{t}{{s_i}!} \leq 2^{t-1}{(q+2-2t)!} \leq {2^t}{(q+2)^{q-2t}}$. 
Also, each of the $t$ components of $H_S$ must have at least two vertices so to determine the sizes 
$s_1,\cdots,s_t$ it is sufficient to decide how to distribute the $q-2t$ extra vertices among the $t$ connected 
components of $H_S$. There are at most ${{q - t - 1} \choose {t-1}} \leq q^{q - 2t}$ ways to do this. Thus, the total 
contribution for terms of a given $t$ is at most 
$${2^t}{(q(q+2))^{q-2t}}n^{xq - 2q + 2t}\left(\frac{2}{x^4}\right)^{t}\left(\frac{{x^4}}{2(x!)}\right)^{q} = 
\left(\frac{{x^2}q(q+2)}{2n}\right)^{q - 2t}\left({n^{xq - q}\left(\frac{{x^2}}{(x!)}\right)^{q}}\right)$$

Since $n \geq {x^2}q(q+2)$, the ${n^{xq - q}\left(\frac{{x^2}}{x!}\right)^{q}}$ term which comes from 
$t = \frac{q}{2}$ contributes more than all the other terms combined, as needed.
\end{proof}

\begin{proof}[Proof of \pref{prop:exponentprop}]
Rearranging this equation gives ${d \choose 2} = {x \choose 2} + dx + {{d-x} \choose 2}$ which just says that 
if we want to pick two elements in $[1,d]$ we can either pick two elements in $[1,x]$, one element from $[1,x]$ and 
one element from $[x+1,d]$, or two elements from $[x+1,d]$. 
\end{proof}

\begin{proof}[Proof of \pref{prop:degG-2}]
$$\left|(1+x)^k - (1 + kx)\right| \leq \sum_{j=2}^{k}{\left|{k \choose j}{x^j}\right|} 
\leq \sum_{j=2}^{k}{\left|\frac{1}{j!}{k^j}{x^j}\right|} \leq {k^2}{x^2}\sum_{j=2}^{k}{\frac{1}{j!}} \leq {k^2}{x^2}$$
\end{proof}


\section{Optimality of MPW Analysis}\label{sec:kelner}
In this section we sketch an argument due to Kelner showing that the MPW moments are not PSD when $\omega \gg n^{1/(d+1)}$.
\begin{theorem}
With high probability, the MW moments are not PSD when $\omega \gg n^{\frac{1}{d+1}}$. In particular, if $\omega \gg n^{\frac{1}{d+1}}$ then for all $s$, for some appropriately chosen $C$, with high probability, 
$$\tilde{E}[(C{\omega}^d{x_s}-\sum_{I: I \subseteq V, |I|=d}{\prod_{i \in I}{r_s(i)x_i}})^2] < 0$$ 
with high probability.
\end{theorem}
\begin{proof}[Proof Sketch]
We will be using the following proposition heavily.
\begin{proposition}
For all $I \subseteq V(G)$ such that $|I| < 2d$, $\tilde{E}[\sum_{j \notin I}{x_{I \cup j}}] = (\omega-|I|)\tilde{E}[x_I]$
\end{proposition}
\begin{proof}
We have the equation that $\sum_{j}{x_j} = \omega$ so 
$$\tilde{E}[\sum_{j}{x_{I \cup j}}] = |I|\tilde{E}[x_I] + \tilde{E}[\sum_{j \notin I}{x_{I \cup j}}] 
= {\omega}\tilde{E}[x_I]$$
and the result follows.
\end{proof}
\begin{corollary}\label{kelnerproofcorollary}
For all $I \subseteq V(G)$ and all $m$ such that $|I| + m \leq 2d$, 
$$\tilde{E}[\sum_{J: I \subseteq J, |J|=|I|+m}{x_{J}}] = \frac{\prod_{x=0}^{m-1}{(\omega-|I|-x)}}{m!}\tilde{E}[x_I] = \binom{\omega-|I|}{m}\tilde{E}[x_I]$$
\end{corollary}
\begin{proof}
This result follows from repeatedly expanding out $(\omega-|K|)\tilde{E}[x_K]=\sum_{j \notin K}{\tilde{E}[x_{K \cup j}]}$. For any given $J$ such that $I \subseteq J$ and $|J| = |I|+m$ there are $m!$ different ways to reach $J$ from $I$ which gives us the $m!$ term.
\end{proof}

With this corollary in hand, we expand out $\tilde{E}[(C{\omega}^d{x_s}-\sum_{I: I \subseteq V \setminus \{s\}, \atop |I|=d}{\prod_{i \in I}{r_s(i)x_i}})^2]$, obtaining
$$\tilde{E}[C^2{\omega}^{2d}x_s] - 2C{\omega}^d\tilde{E}[\sum_{I: I \subseteq V \setminus \{s\}, \atop |I|=d}{x_{I \cup \{s\}}}] + \tilde{E}[\sum_{I,J: I,J \subseteq V \setminus \{s\}, \atop |I|=|J|=d}{\left(\prod_{i \in I \Delta J}{r_s(i)}\right)x_{I \cup J}}]$$

Using Corollary \ref{kelnerproofcorollary}, 
$$2C{\omega}^d\tilde{E}[\sum_{I: I \subseteq V \setminus \{s\}, \atop |I|=d}{x_{I \cup \{s\}}}] = 2C{\omega}^d\binom{\omega-1}{d}\tilde{E}[x_s]$$

Thus combining the first two terms gives 
$$\left(C^2{\omega}^{2d} - 2C{\omega}^d\binom{\omega-1}{d}\right)\tilde{E}[x_s]$$
From our concnetration bounds on $deg_G(s)$, with high probability $\tilde{E}(x_s)$ is $\frac{\omega}{n}(1 \pm O(\frac{log(n)}{\sqrt{n}}))$. Thus, taking $C$ to be a sufficiently small constant (which will depend on $d$), with high probability the first two terms are $-\Omega(\frac{{\omega}^{2d+1}}{n})$

Analyzing the third term is trickier. For the third term, taking $K = I \Delta J$ for each term,
\begin{align*}
\tilde{E}[\sum_{I,J: I,J \subseteq V \setminus \{s\}, \atop |I|=|J|=d}{\left(\prod_{i \in I \Delta J}{r_s(i)}\right)x_{I \cup J}}] &= \sum_{x=0}^{d}{\sum_{I,J: I,J \subseteq V \setminus \{s\}, \atop 
|I|=|J|=d, |I \Delta J|=2x}{\left(\prod_{i \in I \Delta J}{r_s(i)}\right)\tilde{E}[x_{I \cup J}]}} \\
&= \sum_{x=0}^{d}{\sum_{K:K \subseteq V \setminus \{s\}, \atop |K| = 2x}{\binom{\omega-d-x}{d-x}\left(\prod_{i \in K}{r_s(i)}\right){\tilde{E}[x_K]}}}
\end{align*}

We will now analyze the expected value and variance of this expression. However, before doing so there is a subtle issue we must deal with. ${\tilde{E}[x_K]}$ is not completely independent of $\left(\prod_{i \in K}{r_s(i)}\right)$. What saves us is that the dependence is small enough to be negligible. 

For each $K \subseteq V(G) \setminus \{s\}$, define $y_K$ to be the expected value of $\tilde{E}[x_K]$ if we preserve all of the edges of $G$ which are not incident with $s$ but reselect the edges of $G$ incident with $s$ randomly. From our concentration results on $deg_G(K)$, with high probability, for all $K$, $|\tilde{E}[x_K]-y_K|$ is $O(\frac{\omega^{|K|}log(n)}{n^{|K|+1}})$. We now write 
\begin{align*}
\sum_{x=0}^{d}{\sum_{K:K \subseteq V \setminus \{s\}, \atop |K| = 2x}{}}&{\binom{\omega-d-x}{d-x}\left(\prod_{i \in K}{r_s(i)}\right){\tilde{E}[x_K]}} = \\
&\sum_{x=0}^{d}{\sum_{K:K \subseteq V \setminus \{s\}, \atop |K| = 2x}{\binom{\omega-d-x}{d-x}\left(\prod_{i \in K}{r_s(i)}\right){y_K}}} \\
& + \sum_{x=0}^{d}{\sum_{K:K \subseteq V \setminus \{s\}, \atop |K| = 2x}{\binom{\omega-d-x}{d-x}\left(\prod_{i \in K}{r_s(i)}\right){(\tilde{E}[x_K]-y_K)}}}
\end{align*}
For the second part, for each $x \in [0,d]$ there are at most $\binom{n}{2x}$ $K$ such that $K \subseteq V(G) \setminus \{s\}$ and $|K|=2x$. Thus, 
\begin{align*}
\sum_{x=0}^{d}{\sum_{K:K \subseteq V \setminus \{s\}, \atop |K| = 2x}{}}{}&{\binom{\omega-d-x}{d-x}\left(\prod_{i \in K}{r_s(i)}\right){(\tilde{E}[x_K]-y_K)}} \\
&\leq \sum_{x=0}^{d}{\binom{n}{2x}\binom{\omega-d-x}{d-x}\max_{K: |K|=2x}{|\tilde{E}[x_K]-y_K|}}
\end{align*}
From our concentration bounds, with high probability, for all $x$ the corresponding term on the right is $O(\frac{{\omega}^{d+x}\lg{n}}{n})$
which is $O(\frac{\omega^{d+x}}{n})$. For all $x \leq d$, this is much smaller than $\Omega(\frac{{\omega}^{2d+1}}{n})$. Thus we may ignore the second part.

For the first part, the values $r_s(i)$ are completely independent of the values $y_{K}$. When we take the expected value of 
$$\sum_{x=0}^{d}{\sum_{K:K \subseteq V \setminus \{s\}, \atop |K| = 2x}{\binom{\omega-d-x}{d-x}\left(\prod_{i \in K}{r_s(i)}\right){y_K}}}$$
over the edges incident to $s$, only the $x=0$ term remains and we obtain $\binom{\omega-d}{d}$. When we take the variance of $$\sum_{x=0}^{d}{\sum_{K:K \subseteq V \setminus \{s\}, \atop |K| = 2x}{\binom{\omega-d-x}{d-x}\left(\prod_{i \in K}{r_s(i)}\right){y_K}}}$$
over the edges incident to $s$, only the square terms remain so we obtain
$$\sum_{x=0}^{d}{\sum_{K:K \subseteq V \setminus \{s\}, \atop |K| = 2x}{{\binom{\omega-d-x}{d-x}}^2{(y_K)^2}}}$$
From our concentration bounds on the $y_K$, with high probability this is 
$${\binom{\omega-d}{d}}^2 + \sum_{x=1}^{d}{{\binom{\omega-d-x}{d-x}}^2\binom{n-1}{2x}O\left(\frac{\omega^{4x}}{n^{4x}}\right)}$$
which is 
$${\binom{\omega-d}{d}}^2 + O\left(\frac{\omega^{2d+2}}{n^2}\right)$$
Thus, with high probability $$\sum_{x=0}^{d}{\sum_{K:K \subseteq V \setminus \{s\}, \atop |K| = 2x}{\binom{\omega-d-x}{d-x}\left(\prod_{i \in K}{r_s(i)}\right){y_K}}}$$
is $\binom{\omega-d}{d}\left(1 \pm O(\frac{\omega}{n})\right)$ which is $O(\omega^{d})$.
Putting everything together, if $\omega \gg n^{\frac{1}{d+1}}$ then for some appropriately chosen $C$, with high probabbility, $\tilde{E}[(C{\omega}^d{x_s}-\sum_{I: I \subseteq V, |I|=d}{\prod_{i \in I}{r_s(i)x_i}})^2] < 0$
\end{proof}

\section{Acknowledgements}
We thank Jonathan Kelner, Ankur Moitra, Aviad Rubinstein, and David Steurer for many helpful discussions.
We especially thank Boaz Barak for introducing the problem to us, for numerous conversations which led to many of the proofs in this paper, and for invaluable help and advice in writing it.

S.B.H. acknowledges the support of an NSF Graduate Research Fellowship under award no. 1144153.

\addreferencesection
\bibliographystyle{amsalpha}
\bibliography{bib/mr,bib/dblp,bib/scholar,bib/tensor-pca}

\end{document}